\ifdefined\isfocs
\documentclass[10pt, conference, compsocconf]{IEEEtran}
\else
\documentclass[11pt]{article}
\usepackage[margin=1in]{geometry}
\fi
\usepackage{amsmath}
\usepackage{amsthm}
\usepackage{amssymb}
\usepackage{algorithm}
\usepackage{subfig}
\usepackage{color}
\usepackage[english]{babel}
\usepackage{graphicx}
\usepackage{wrapfig,epsfig}
\usepackage{epstopdf}
\usepackage{url}
\usepackage{graphicx}
\usepackage{color}
\usepackage{epstopdf}
\usepackage{algpseudocode}
\usepackage{scrextend}
\usepackage[T1]{fontenc}
\usepackage{bbm}
\usepackage{comment}
\usepackage{multicol} 
\usepackage{paralist}

\usepackage{tikz}
\usepackage{hyperref}
\hypersetup{colorlinks=true,citecolor=red,linkcolor=blue} 
\usetikzlibrary{arrows}

\ifdefined\isfocs
\title{Algorithms and Hardness for Linear Algebra on Geometric Graphs}
\else
\title{Algorithms and Hardness for \\ Linear Algebra on Geometric Graphs\thanks{A preliminary version of this paper appeared in the Proceedings of 61st Annual IEEE Symposium on Foundations of Computer Science (FOCS 2020).}}
\fi

\ifdefined\isfocs

\author{
\IEEEauthorblockN{Josh Alman}
\IEEEauthorblockA{\texttt{jalman@seas.harvard.edu}\\
Harvard U.}
\and
\IEEEauthorblockN{Timothy Chu}
\IEEEauthorblockA{\texttt{tzchu@andrew.cmu.edu}\\
Carnegie Mellon U.}
\and
\IEEEauthorblockN{Aaron Schild}
\IEEEauthorblockA{\texttt{aschild@uw.edu}\\
U. of Washington}
\and
\IEEEauthorblockN{Zhao Song}
\IEEEauthorblockA{\texttt{zhaos@utexas.edu}\\
Columbia, Princeton, IAS}
}
\else
\author{
Josh Alman\thanks{\texttt{jalman@seas.harvard.edu}, Harvard University, supported by a Rabin fellowship. Part of the work done while the author was a graduate student at MIT. Supported in part by a Michael O. Rabin Postdoctoral Fellowship.} 
\quad
Timothy Chu\thanks{\texttt{tzchu@andrew.cmu.edu}, Carnegie Mellon University.}
\quad
Aaron Schild\thanks{\texttt{aschild@uw.edu}, University of Washington. Part of the work done while the author was a graduate student at UC Berkeley.  Supported in part by ONR \#N00014-17-1-2429, a Packard fellowship, and the University of Washington. }
\quad
Zhao Song\thanks{\texttt{magic.linuxkde@gmail.com}, Columbia University, Princeton University and Institute for Advanced Study. Part of the work done while the author was a graduate student at UT-Austin. Supported in part by Ma Huateng Foundation, Schmidt Foundation, Simons Foundation, NSF, DARPA/SRC, Google and Amazon.}
}
\fi
\date{}

\newtheorem{theorem}{Theorem}[section]
\newtheorem{lemma}[theorem]{Lemma}
\newtheorem{definition}[theorem]{Definition}

\newtheorem{proposition}[theorem]{Proposition}
\newtheorem{corollary}[theorem]{Corollary}
\newtheorem{conjecture}[theorem]{Conjecture}

\newtheorem{fact}[theorem]{Fact}
\newtheorem{remark}[theorem]{Remark}

\newtheorem{example}[theorem]{Example}
\newtheorem{problem}[theorem]{Problem}

\newtheorem{hypothesis}[theorem]{Hypothesis}

\newcommand{\E}{\mathbf{E}}
\newcommand{\T}{\mathcal{T}}
\newcommand{\Err}{\mathrm{Err}}
\newcommand{\poly}{\mathrm{poly}}

\newcommand{\tr}{\mathrm{tr}}
\newcommand{\polylog}{\mathrm{polylog}}
\newcommand{\Ceff}{\mathtt{Ceff}}

\newcommand{\wt}{\widetilde}
\newcommand{\ov}{\overline}

\renewcommand{\tilde}{\wt}

\renewcommand{\k}{\mathsf{K}}
\newcommand{\Reff}{\mathtt{Reff}}

\newcommand{\KAdjE}{\mathsf{KAdjE}}
\newcommand{\KLapE}{\mathsf{KLapE}}

\newcommand{\WSPD}{\mathrm{WSPD}}
\newcommand{\dist}{\mathrm{dist}}
\renewcommand{\d}{\mathrm{d}}

\newcommand{\SETH}{\mathsf{SETH}}
\newcommand{\ANN}{\mathsf{ANN}}

\newcommand{\LowDiamSet}{\textsc{LowDiamSet}}
\newcommand{\ReffQuery}{\textsc{ReffQuery}}
\newcommand{\ReffPreproc}{\textsc{ReffPreproc}}
\newcommand{\UnweightedCover}{\textsc{UnweightedCover}}
\newcommand{\BoundedCover}{\textsc{BoundedCover}}
\newcommand{\TwoBoundedCover}{\textsc{TwoBoundedCover}}
\newcommand{\LogCover}{\textsc{LogCover}}
\newcommand{\DesiredCover}{\textsc{DesiredCover}}
\newcommand{\IntervalFamily}{\textsc{IntervalFamily}}
\newcommand{\OversamplingWithCover}{\textsc{OversamplingWithCover}}

\def\Z{{\mathbb Z}}
\def\N{{\mathbb N}}
\def\R{{\mathbb R}}

\def\eps{\varepsilon}
\renewcommand{\epsilon}{\eps}

\definecolor{b2}{RGB}{51,153,255}
\definecolor{mygreen}{RGB}{80,180,0}

\newcommand{\Josh}[1]{{\color{b2}[Josh: #1]}}
\newcommand{\josh}[1]{{\color{b2}[Josh: #1]}}
\newcommand{\Aaron}[1]{{\color{mygreen}[Aaron: #1]}}
\renewcommand{\Josh}[1]{{\color{b2}}}
\renewcommand{\josh}[1]{{\color{b2}}}
\renewcommand{\Aaron}[1]{{\color{brown}}}

\begin{document}

\ifdefined\isfocs
 \maketitle
 \begin{abstract}
For a function $\mathsf{K} : \mathbb{R}^{d} \times \mathbb{R}^{d} \to \mathbb{R}_{\geq 0}$, and a set $P = \{ x_1, \ldots, x_n\} \subset \mathbb{R}^d$ of $n$ points, the $\mathsf{K}$ graph $G_P$ of $P$ is the complete graph on $n$ nodes where the weight between nodes $i$ and $j$ is given by $\mathsf{K}(x_i, x_j)$. In this paper, we initiate the study of when efficient spectral graph theory is possible on these graphs. We investigate whether or not it is possible to solve the following problems in $n^{1+o(1)}$ time for a $\mathsf{K}$-graph $G_P$ when $d < n^{o(1)}$:

\begin{itemize}
\item Multiply a given vector by the adjacency matrix or Laplacian matrix of $G_P$
\item Find a spectral sparsifier of $G_P$
\item Solve a Laplacian system in $G_P$'s Laplacian matrix
\end{itemize}

For each of these problems, we consider all functions of the form $\mathsf{K}(u,v) = f(\|u-v\|_2^2)$ for a function $f:\mathbb{R} \rightarrow \mathbb{R}$. We provide algorithms and comparable hardness results for many such $\mathsf{K}$, including the Gaussian kernel, Neural tangent kernels, and more. For example, in dimension $d = \Omega(\log n)$, we show that there is a parameter associated with the function $f$ for which low parameter values imply $n^{1+o(1)}$ time algorithms for all three of these problems and high parameter values imply the nonexistence of subquadratic time algorithms assuming Strong Exponential Time Hypothesis ($\mathsf{SETH}$), given natural assumptions on $f$.

As part of our results, we also show that the exponential dependence on the dimension $d$ in the celebrated fast multipole method of Greengard and Rokhlin cannot be improved, assuming $\mathsf{SETH}$, for a broad class of functions $f$. To the best of our knowledge, this is the first formal limitation proven about fast multipole methods.

  \end{abstract}
\else

\begin{titlepage}
 \maketitle
  \begin{abstract}
For a function $\mathsf{K} : \mathbb{R}^{d} \times \mathbb{R}^{d} \to \mathbb{R}_{\geq 0}$, and a set $P = \{ x_1, \ldots, x_n\} \subset \mathbb{R}^d$ of $n$ points, the $\mathsf{K}$ graph $G_P$ of $P$ is the complete graph on $n$ nodes where the weight between nodes $i$ and $j$ is given by $\mathsf{K}(x_i, x_j)$. In this paper, we initiate the study of when efficient spectral graph theory is possible on these graphs. We investigate whether or not it is possible to solve the following problems in $n^{1+o(1)}$ time for a $\mathsf{K}$-graph $G_P$ when $d < n^{o(1)}$:

\begin{itemize}
\item Multiply a given vector by the adjacency matrix or Laplacian matrix of $G_P$
\item Find a spectral sparsifier of $G_P$
\item Solve a Laplacian system in $G_P$'s Laplacian matrix
\end{itemize}

For each of these problems, we consider all functions of the form $\mathsf{K}(u,v) = f(\|u-v\|_2^2)$ for a function $f:\mathbb{R} \rightarrow \mathbb{R}$. We provide algorithms and comparable hardness results for many such $\mathsf{K}$, including the Gaussian kernel, Neural tangent kernels, and more. For example, in dimension $d = \Omega(\log n)$, we show that there is a parameter associated with the function $f$ for which low parameter values imply $n^{1+o(1)}$ time algorithms for all three of these problems and high parameter values imply the nonexistence of subquadratic time algorithms assuming Strong Exponential Time Hypothesis ($\mathsf{SETH}$), given natural assumptions on $f$.

As part of our results, we also show that the exponential dependence on the dimension $d$ in the celebrated fast multipole method of Greengard and Rokhlin cannot be improved, assuming $\mathsf{SETH}$, for a broad class of functions $f$. To the best of our knowledge, this is the first formal limitation proven about fast multipole methods.

  \end{abstract}
 \thispagestyle{empty}
 \end{titlepage}

\pagenumbering{roman}

{\hypersetup{linkcolor=black}
\tableofcontents
}

\newpage

\setcounter{page}{1}
\pagenumbering{arabic}
\fi

\section{Introduction}\label{sec:intro}

Linear algebra has a myriad of applications throughout computer science and physics. Consider the following seemingly unrelated tasks:

\vspace{1mm}
\begin{compactenum}
    \item \textbf{$n$-body simulation (one step)}: Given $n$ bodies $X$ located at points in $\mathbb{R}^d$, compute the gravitational force on each body induced by the other bodies.
    \item \textbf{Spectral clustering}: Given $n$ points $X$ in $\mathbb{R}^d$, partition $X$ by building a graph $G$ on the points in $X$, computing the top $k$ eigenvectors of the Laplacian matrix $L_G$ of $G$ for some $k\ge 1$ to embed $X$ into $\mathbb{R}^k$, and run $k$-means on the resulting points.
    \item \textbf{Semi-supervised learning}: Given $n$ points $X$ in $\mathbb{R}^d$ and a function $g:X\rightarrow \mathbb{R}$ whose values on some of $X$ are known, extend $g$ to the rest of $X$.
\end{compactenum}

\vspace{1mm}

Each of these tasks has seen much work throughout numerical analysis, theoretical computer science, and machine learning. The first task is a celebrated application of the fast multipole method of Greengard and Rokhlin~\cite{gr87, gr88, gr89}, voted one of the top ten algorithms of the twentieth century by the editors of \emph{Computing in Science and Engineering}~\cite{dongarra2000guest}.
The second task is \emph{spectral clustering} \cite{njw02, lwdh13}, a popular algorithm for clustering data. The third task is to label a full set of data given only a small set of partial labels~\cite{z05survey, csbz09, zl05}, which has seen increasing use in machine learning. One notable method for performing semi-supervised learning is the graph-based Laplacian regularizer method~\cite{lszlh19,zl05, bns06,z05}.

Popular techniques for each of these problems benefit from primitives in spectral graph theory on a special class of dense graphs called \emph{geometric graphs}. For a function $\k:\mathbb{R}^d\times \mathbb{R}^d\rightarrow \mathbb{R}$ and a set of points $X\subseteq \mathbb{R}^d$, the \emph{$\k$-graph on $X$} is a graph with vertex set $X$ and edges with weight $\k(u,v)$ for each pair $u,v\in X$. Adjacency matrix-vector multiplication, spectral sparsification, and Laplacian system solving in geometric graphs are directly relevant to each of the above problems, respectively:

\vspace{1mm}
\begin{compactenum}
\item \textbf{$n$-body simulation (one step)}: For each $i\in \{1,2,\hdots,d\}$, make a weighted graph $G_i$ on the points in $X$, in which the weight of the edge between the points $u,v\in X$ in $G_i$ is $\k_i(u,v) := (\frac{G_{\text{grav}} \cdot m_u \cdot m_v}{\|u - v\|_2^2})(\frac{v_i - u_i}{\|u - v\|_2})$,  where $G_{\text{grav}}$ is the gravitational constant and $m_x$ is the mass of the point $x\in X$. Let $A_i$ denote the weighted adjacency matrix of $G_i$. Then $A_i\textbf{1}$ is the vector of $i$th coordinates of force vectors. In particular, gravitational force can be computed by doing $O(d)$ adjacency matrix-vector multiplications, where each adjacency matrix is that of the $\k_i$-graph on $X$ for some $i$.

\item \textbf{Spectral clustering}: Make a $\k$ graph $G$ on $X$. In applications, $\k(u,v) = f(\|u-v\|_2^2)$, where $f$ is often chosen to be $f(z) = e^{-z}$~\cite{l07,njw02}. Instead of directly running a spectral clustering algorithm on $L_G$, one popular method is to construct a sparse matrix $M$ approximating $L_G$ and run spectral clustering on $M$ instead~\cite{chl16,csblc11, kmt12}. Standard sparsification methods in the literature are heuristical, and include the widely used Nystrom method which uniformly samples rows and columns from the original matrix~\cite{cjkmm13}. 

If $H$ is a spectral sparsifier of $G$, it has been suggested that spectral clustering with the top $k$ eigenvectors of $L_H$ performs just as well in practice as spectral clustering with the top $k$ eigenvectors of $L_G$~\cite{chl16}. 
One justification is that since $H$ is a spectral sparsifier of $G$, the eigenvalues of $L_H$ are at most a constant factor larger than those of $L_G$, so cuts with similar conductance guarantees are produced. Moreover, spectral clustering using sparse matrices like $L_H$ is known to be faster than spectral clustering on dense matrices like $L_G$ ~\cite{chl16, cjkmm13, kmt12}.

\item \textbf{Semi-supervised learning}: An important subroutine in semi-supervised learning is completion based on $\ell_2$-minimization~\cite{z05, z05survey, lszlh19}. Specifically, given values $g_v$ for $v\in Y$, where $Y$ is a subset of $X$, find the vector $g\in \mathbb{R}^n$ (variable over $X\setminus Y$) that minimizes
$\sum_{u,v\in X,u\ne v} \k(u,v) (g_u - g_v)^2.$
The vector $g$ can be found by solving a Laplacian system on the $\k$-graph for $X$.
\end{compactenum}
\vspace{1mm}

In the first, second, and third tasks above, a small number of calls to matrix-vector multiplication, spectral sparsification, and Laplacian system solving, respectively, were made on geometric graphs. One could solve these problems by first explicitly writing down the graph $G$ and then using near-linear time algorithms \cite{ss11,ckmpprx14} to multiply, sparsify, and solve systems. However, this requires a minimum of $\Omega(n^2)$ time, as $G$ is a dense graph.

In this paper, we initiate a theoretical study of the geometric graphs for which efficient spectral graph theory is possible. In particular, we attempt to determine for which (a) functions $\k$ and (b) dimensions $d$ there is a much faster, $n^{1+o(1)}$-time algorithm for each of (c) multiplication, sparsification, and Laplacian solving. Before describing our results, we elaborate on the choices of (a), (b), and (c) that we consider in this work.

We start by discussing the functions $\k$ that we consider (part (a)). Our results primarily focus on the class of functions of the form $\k(u,v) = f(\|u-v\|_2^2)$ for a function $f:\mathbb{R}_{\ge 0}\rightarrow \mathbb{R}$ for $u,v\in \mathbb{R}^d$. Study of these functions dates back at least eighty years, to the early work of Bochner, Schoenberg, and John Von Neumann on metric embeddings into Hilbert Spaces~\cite{b33, s37, ns41}. These choices of $\k$ are ubiquitous in applications, like the three described above, since they naturally capture many kernel functions $\k$ from statistics and machine learning, including the Gaussian kernel $(e^{-\|u-v\|_2^2})$, the exponential kernel $(e^{-\|u-v\|_2})$, the power kernel $(\|u-v\|_2^q)$ for both positive and negative $q$, the logarithmic kernel ($\log (\|u-v\|_2^q + c)$), and more~\cite{s10, z05, btb05}.  See Section~\ref{subsec:relatedwork} below for even more popular examples. In computational geometry, many transformations of distance functions are also captured by such functions $\k$, notably in the case when $\k(u,v) = \|u-v\|_2^q$~\cite{l82, as14,acx19, cms20}. 

We would also like to emphasize that many kernel functions which do not at first appear to be of the form $f(\|u-v\|_2^2)$ can be rearranged appropriately to be of this form. For instance, in
\ifdefined\isfocs
the full version
\else
Section~\ref{sec:ntk} below 
\fi
we show that the recently popular Neural Tangent Kernel is of this form, so our results apply to it as well.
That said, to emphasize that our results are very general, we will mention later how they also apply to some functions of the form $\k(u,v) = f(\langle u,v\rangle )$, including $\k(u,v) = |\langle u,v\rangle |$.

Next, we briefly elaborate on the problems that we consider (part (c)). For more details, see 
\ifdefined\isfocs
the full version. 
\else
Section \ref{sec:preli}. 
\fi
The points in $X$ are assumed to be real numbers stated with $\polylog(n)$ bits of precision. Our algorithms and hardness results pertain to algorithms that are allowed some degree of approximation. For an error parameter $\epsilon > 0$, our multiplication and Laplacian system solving algorithms produce solutions with $\epsilon$-additive error, and our sparsification algorithms produce a graph $H$ for which the Laplacian quadratic form $(1 \pm \epsilon)$-approximates that of $G$.

Matrix-vector multiplication, spectral sparsification, and Laplacian system solving are very natural linear algebraic problems in this setting, and have many applications beyond the three we have focused on ($n$-body simulation, spectral clustering, and semi-supervised learning). See Section~\ref{subsec:relatedwork} below where we expand on more applications.

Finally, we discuss dependencies on the dimension $d$ and the accuracy $\epsilon$ for which $n^{1+o(1)}$ algorithms are possible (part (b)). Define $\alpha$, a measure of the `diameter' of the point set and $f$, as
$$\alpha := \frac{ \max_{u,v \in X} f( \| u - v \|_2^2 ) }{ \min_{u,v \in X} f( \| u - v \|_2^2 ) } + \frac{ \max_{u,v \in X}  \| u - v \|_2^2  }{ \min_{u,v \in X}  \| u - v \|_2^2  }.$$
It is helpful to have the following two questions in mind when reading our results:

\vspace{2mm}
\begin{compactitem}
\item (High-dimensional algorithms, e.g. $d = \Theta(\log n)$) Is there an algorithm which runs in time $\poly(d, \log(n\alpha/\epsilon)) n^{1+o(1)}$ for multiplication and Laplacian solving? Is there an algorithm which runs in time $\poly(d, \log(n\alpha))n^{1+o(1)}$ for sparsification when $\epsilon=1/2$?

\item (Low-dimensional algorithms, e.g. $d = o(\log n)$) Is there an algorithm which runs in time $(\log(n\alpha/\epsilon))^{O(d)} n^{1+o(1)}$ for multiplication and Laplacian solving? Is there a sparsification algorithm which runs in time $(\log(n\alpha))^{O(d)} n^{1+o(1)}$ when $\epsilon=1/2$?
\end{compactitem}
\vspace{2mm}

We will see that there are many important functions $\k$ for which there are such efficient low-dimensional algorithms, but no such efficient high-dimensional algorithms. In other words, these functions $\k$ suffer from the classic `curse of dimensionality.' At the same time, other functions $\k$ will allow for efficient low-dimensional and high-dimensional algorithms, while others won't allow for either.

We now state our results. We will give very general classifications of functions $\k$ for which our results hold, but afterwards in Section~\ref{sec:resultstable} we summarize the results for a few particular functions $\k$ of interest. The main goal of our results is as follows:

\textbf{Goal}: For each problem of interest (part (c)) and dimension $d$ (part (b)), find a natural parameter $p_f > 0$ associated with the function $f$ for which the following dichotomy holds:

\begin{compactitem}
    \item If $p_f$ is high, then the problem cannot be solved in subquadratic time assuming $\SETH$ on points in dimension $d$.
    \item If $p_f$ is low, then the problem of interest can be solved in almost-linear time ($n^{1+o(1)}$ time) on points in dimension $d$.
\end{compactitem}

As we will see shortly, the two parameters $p_f$ which will characterize the difficulties of our problems of interest in most settings are the \emph{approximate degree} of $f$, and a parameter related to how \emph{multiplicatively Lipschitz} $f$ is. We define both of these in the next section.

\subsection{High-dimensional results}\label{subsubsec:highdimmult}

We begin in this subsection by stating our results about which functions have $\poly(d,\log(\alpha),\log(1/\epsilon)) \cdot$ $n^{1+o(1)}$-time algorithms for multiplication and Laplacian solving and $\poly(d,\log(\alpha),1/\epsilon)\cdot n^{1+o(1)}$-time algorithms for sparsification. When reading these results, it is helpful to think of $d = \Theta(\log n)$, $\alpha = 2^{\text{polylog}(n)}$, $\epsilon = 1/2^{\polylog(n)}$ for multiplication and Laplacian solving, and $\epsilon = 1/2$ for sparsification. With these parameter settings, $\poly(d)n^{1+o(1)}$ is almost-linear time, while $2^{O(d)}n^{1+o(1)}$ time is not. For results about algorithms with runtimes that are exponential in $d$, see 
\ifdefined\isfocs
the full version.
\else
Section \ref{sec:low-dim-summary}.
\fi

\subsubsection{Multiplication} 

In high dimensions, we give a full classification of when the matrix-vector multiplication problems are easy for kernels of the form $\k(u,v) = f(\|u - v\|_2^2)$ for some function $f:\mathbb{R}_{\ge 0}\rightarrow \mathbb{R}_{\ge 0}$ that is analytic on an interval. We show that the problem can be efficiently solved only when $\k$ is very well-approximated by a simple polynomial kernel. That is, we let $p_f$ denote the minimum degree of a polynomial that $\eps$-additively-approximates the function $f$.

\begin{theorem}
[Informal version of Theorem~\ref{thm:hardnessapprox} and Corollary~\ref{cor:kernelalg}]\label{thm:informal-high}
For any function $f : \R_+ \to \R_+$ which is analytic on an interval $(0,\delta)$ for any $\delta$ $>0$, and any $0 < \eps < 2^{-\polylog(n)}$, consider the following problem: given as input $x_1, \ldots, x_n \in \R^d$ with $d = \Theta(\log n)$ which define a $\k$ graph $G$ via $\k(x_i, x_j) = f(\|x_i -x_j\|_2^2)$, and a vector $y \in \{0,1\}^n$, compute an $\eps$-additive-approximation to $L_G \cdot y$. 
\begin{compactitem}
    \item If $f$ can be $\eps$-additively-approximated by a polynomial of degree at most $o(\log n)$, then the problem can be solved in $n^{1+o(1)}$ time.
    \item Otherwise, assuming $\SETH$, the problem requires time $n^{2 - o(1)}$.
\end{compactitem}
The same holds for $L_G$, the Laplacian matrix of $G$, replaced by $A_G$, the adjacency matrix of $G$.
\end{theorem}

While Theorem \ref{thm:informal-high} yields a parameter $p_f$ that characterizes hardness of multiplication in high dimensions, it is somewhat cumbersome to use, as it can be challenging to show that a function is far from a polynomial. We also show Theorem \ref{thm:hardnessapprox-easy}, which shows that if $f$ has a single point with large $\Theta(\log n)$-th derivative, then the problem requires time $n^{2-o(1)}$ assuming $\SETH$. The Strong Exponential Time Hypothesis ($\SETH$) is a common assumption in fine-grained complexity regarding the difficulty of solving the Boolean satisfiability problem; see 
\ifdefined\isfocs
the full version
\else
section~\ref{sec:fine_grained}
\fi 
for more details. Theorem~\ref{thm:informal-high} informally says that assuming $\SETH$, the curse of dimensionality is inherent in performing adjacency matrix-vector multiplication. In particular, we directly apply this result to the $n$-body problem discussed at the beginning:

\begin{corollary} \label{cor:nbodyhardness}
Assuming $\SETH$, in dimension $d = \Theta(\log n)$ one step of the $n$-body problem requires time $n^{2 - o(1)}$.
\end{corollary}

The fast multipole method of Greengard and Rokhlin~\cite{gr87, gr89} solves one step of this $n$-body problem in time $(\log(n/\epsilon))^{O(d)} n^{1+o(1)}$. Our Corollary~\ref{cor:nbodyhardness} shows that assuming $\SETH$, such an exponential dependence on $d$ is required and cannot be improved. To the best of our knowledge, this is the first time such a formal limitation on fast multipole methods has been proved. This hardness result also applies to fast multipole methods for other popular kernels, like the Gaussian kernel $\k(u,v) = \exp(-\|u - v\|_2^2)$, as well.

\subsubsection{Sparsification}

We next show that sparsification can be performed in almost-linear time in high dimensions for kernels that are ``multiplicatively Lipschitz'' functions of the $\ell_2$-distance. We say $f:\mathbb{R}_{\ge 0}\rightarrow \mathbb{R}_{\ge 0}$ is $(C,L)$-\emph{multiplicatively Lipschitz} for $C > 1, L > 1$ if for all $x\in \mathbb{R}_{\ge 0}$ and all $\rho\in (1/C,C)$,
$$C^{-L} f(x)\le f(\rho x)\le C^L f(x).$$
Here are some popular functions that are helpful to think about in the context of our results:\\
${}$\hspace{4mm}1. $f(z) = z^L$ for any positive or negative constant $L$. This function is $(C,|L|)$-multiplicatively Lipschitz for any $C > 1$. \\
${}$\hspace{4mm}2. $f(z) = e^{-z}$. This function is not $(C,L)$-multiplicatively Lipschitz for any $L > 1$ and $C > 1$. We call this the \emph{exponential function}.\\
${}$\hspace{4mm}3.  The piecewise function $f(z) = e^{-z}$ for $z\le L$ and $f(z) = e^{-L}$ for $z > L$. This function is $(C,O(L))$-multiplicatively Lipschitz for any $C > 1$. We call this a \emph{piecewise exponential function}. \\
${}$\hspace{4mm}4. The piecewise function $f(z) = 1$ for $z\le k$ and $f(z) = 0$ for $z > k$, where $k\in \mathbb{R}_{\ge 0}$. This function is not $(C,L)$-multiplicatively Lipschitz for any $C > 1$ or $L > 1$. This is a \emph{threshold function}.

We show that multiplicatively Lipschitz functions can be sparsified in $n^{1+o(1)} \poly(d)$ time: 

\begin{theorem}[Informal version of Theorem~\ref{thm:sparsify-lipschitz}]\label{thm:informal-sparsify-lipschitz}
For any function $f$ such that $f$ is $(2, L)$-multiplicatively Lipschitz, building a $(1\pm\epsilon)$-spectral sparsifier of the $\k$-graph on $n$ points in $\mathbb{R}^d$ where $\k(u,v) = f(\|u - v\|_2^2)$, with $O(n \log n / \eps^2)$ edges, can be done in time
\begin{align*}
 O(nd \sqrt{L \log n}) + n \log n \cdot 2^{O(\sqrt{L \log n})} \cdot (\log\alpha)/\eps^2 .
\end{align*}
\end{theorem}
This Theorem applies even when $d = \Omega(\log n)$. When $L$ is constant, the running time simplifies to $O(nd \sqrt{ \log n} + n^{1+o(1)} \log \alpha/ \eps^2)$. This covers the case when $f(x)$ is any rational function with non-negative coefficients, like $f(z) = z^L$ or $f(z) = z^{-L}$. 

It may seem more natural to instead define $L$-multiplicatively Lipschitz functions, without the parameter $C$, as functions with $\rho^{-L} f(x)\le f(\rho x)\le \rho^L f(x)$ 
for all $\rho$ and $x$. Indeed, an $L$-multiplicatively Lipschitz function is also $(C,L)$-multiplicative Lipschitz for any $C>1$, so our results show that efficient sparsification is possible for such functions. However, the parameter $C$ is necessary to characterize when efficient sparsification is possible. Indeed, as in Theorem~\ref{thm:informal-sparsify-lipschitz} above, it is sufficient for $f$ to be $(C,L)$-multiplicative Lipschitz for a $C$ that is bounded away from 1. To complement this result, we also show a lower bound for sparsification for any function $f$ which is not $(C,L)$-multiplicatively Lipschitz for any $L$ and sufficiently large $C$:

\begin{theorem}[Informal version of Theorem \ref{thm:high-spars-hard}]\label{thm:informal-high-spars-hard}
Consider an $L > 1$. There is some sufficiently large value $C_L > 1$ depending on $L$ such that for any decreasing function $f:\mathbb{R}_{\ge 0}\rightarrow \mathbb{R}_{\ge 0}$ that is not $(C_L,L)$-multiplicatively Lipschitz, no $O(n2^{L^{.48}})$-time algorithm for constructing an $O(1)$-spectral sparsifier of the $\k$-graph of a set of $n$ points in $O(\log n)$ dimensions exists assuming $\SETH$, where $\k(x,y) = f(\|x-y\|_2^2)$.
\end{theorem}

For example, when $L = \Theta(\log^{2+\delta} n)$ for some constant $\delta > 0$, Theorem~\ref{thm:informal-high-spars-hard} shows that there is a $C$ for which, whenever $f$ is not $(C,L)$-multiplicatively Lipschitz, the sparsification problem cannot be solved in time $n^{1+o(1)}$ assuming $\SETH$.

Bounding $C$ in terms of $L$ above is important. For example, if $C$ is small enough that $C^L = 2$, then $f$ could be close to constant. Such $\k$-graphs are easy to sparsify by uniformly sampling edges, so one cannot hope to show hardness for such functions.

Theorem~\ref{thm:informal-high-spars-hard} shows that geometric graphs for threshold functions, the exponential function, and the Gaussian kernel do not have efficient sparsification algorithms. Furthermore, this hardness result essentially completes the story of which \emph{decreasing} functions can be sparsified in high dimensions, modulo a gap of $L^{.48}$ versus $L$ in the exponent. The tractability landscape is likely much more complicated for non-decreasing functions. That said, many of the kernels used in practice, like the Gaussian kernel, are decreasing functions of distance, so our dichotomy applies to them.

We also show that our techniques for sparsification extend beyond kernels that are functions of $\ell_2$ norms; specifically $\k(u,v) = |\langle u,v\rangle|$:

\begin{lemma}[Informal version of Lemma \ref{lem:inner-sparsify}]\label{thm:sparsabsinnerproduct} 
The $\k(u,v) = |\langle u,v\rangle|$-graph on $n$ points in $\mathbb{R}^d$ can be $\epsilon$-approximately sparsified in $n^{1+o(1)} \poly(d) /\epsilon^2$ time.
\end{lemma}

\subsubsection{Laplacian solving}

Laplacian system solving has a similar tractability landscape to that of adjacency matrix multiplication. We prove the following algorithmic result for solving Laplacian systems:

\begin{theorem}[Informal version of Corollary~\ref{cor:exactmonomial} and Proposition~\ref{prop:woodbury}]\label{thm:informal-lapl-poly}
There is an algorithm that takes $ n^{1+o(1)} \poly(d,\log(n\alpha/\epsilon))$ time to $\epsilon$-approximately solve Laplacian systems on $n$-vertex $\k$-graphs, where $\k(u,v) = f(\|u-v\|_2^2)$ for some (nonnegative) polynomial $f$.\footnote{$f$ is a nonnegative function if $f(x) \geq 0$ for all $x \geq 0$.}
\end{theorem}

We show that this theorem is nearly tight via two hardness results. The first applies to multiplicatively Lipschitz kernels, while the second applies to kernels that are not multiplicatively Lipschitz. The second hardness result only works for kernels that are decreasing functions of $\ell_2$ distance. We now state our first hardness result: 

\begin{corollary}\label{cor:introsystemhardhighdim}
Consider a function $f$ that is $(2,o(\log n))$-multiplicatively Lipschitz for which $f$ cannot be $(\epsilon=2^{-\poly(\log  n)})$-approximated by a polynomial of degree at most $o(\log n)$. Then, assuming $\SETH$, there is no $n^{1+o(1)} \poly(d, \log(\alpha n/\epsilon))$-time algorithm for $\epsilon$-approximately solving Laplacian systems in the $\k$-graph on $n$ points, where $\k(u,v) = f(\|u-v\|_2^2)$.
\end{corollary}
\ifdefined\isfocs
In the full version, 
\else
In Section~\ref{sec:equivalence},
\fi
we will see, using an iterative refinement approach, that if a $\k$ graph can be efficiently sparsified, then there is an efficient Laplacian multiplier for $\k$ graphs if and only if there is an efficient Laplacian system solver for $\k$ graphs. Corollary~\ref{cor:introsystemhardhighdim} then follows using this connection: it describes functions which we have shown have efficient sparsifiers but not efficient multipliers.

Corollary~\ref{cor:introsystemhardhighdim}, which is the first of our two hardness results in this setting, applies to slowly-growing functions that do not have low-degree polynomial approximations, like $f(z) = 1/(1 + z)$. Next, we state our second hardness result:

\begin{theorem}[Informal version of Theorem \ref{thm:high-lsolve-hard}]\label{thm:informal-high-lsolve-hard}
Consider an $L > 1$. There is some sufficiently large value $C_L > 1$ depending on $L$ such that for any decreasing function $f:\mathbb{R}_{\ge 0}\rightarrow \mathbb{R}_{\ge 0}$ that is not $(C_L,L)$-multiplicatively Lipschitz, no $O(n  2^{L^{.48}} \log \alpha)$-time algorithm exists for solving Laplacian systems $2^{-\poly(\log n)}$ approximately in the $\k$-graph of a set of $n$ points in $O(\log n)$ dimensions assuming $\SETH$, where $\k(u,v) = f(\|u-v\|_2^2)$. 
\end{theorem}

This yields a quadratic time hardness result when $L = \Omega(\log^2 n)$. By comparison, the first hardness result, Corollary~\ref{cor:introsystemhardhighdim}, only applied for $L = o(\log n)$. In particular, this shows that for non-Lipschitz functions like the Gaussian kernel, the problem of solving Laplacian systems and, in particular, doing semi-supervised learning, cannot be done in almost-linear time assuming $\SETH$.

\subsection{Our Techniques}

\subsubsection{Multiplication}

Our goal in matrix-vector multiplication is, given points $P = \{ x_1, \ldots, x_n\} \subset \R^d$ and a vector $y \in \R^n$, to compute a $(1 \pm \eps)$-approximation to the vector $L_G \cdot y$ where $L_G$ is the Laplacian matrix of the $\k$ graph on $P$, for $\eps = n^{-\Omega(1)}$ (see
\ifdefined\isfocs
the full version
\else
Definition~\ref{def:approxmult}\fi for the precise error guarantees on $\eps$). We call this the $\k$ Laplacian Evaluation ($\KLapE$) problem. A related problem, in which the Laplacian matrix $L_G$ is replaced by the adjacency matrix $A_G$, is the $\k$ Adjacency Evaluation ($\KAdjE$) problem.

We begin by showing a simple, generic equivalence between $\KLapE$ and $\KAdjE$ for any $\k$: an algorithm for either one can be used as a black box to design an algorithm for the other with only negligible blowups to the running time and error. It thus suffices to design algorithms and prove lower bounds for $\KAdjE$.

\paragraph*{Algorithmic Techniques}

We use two primary algorithmic tools: the Fast Multipole Method (FMM), and a `kernel method' for approximating $A_G$ by a low-rank matrix.

FMM is an algorithmic technique for computing aggregate interactions between $n$ bodies which has applications in many different areas of science. Indeed, when the interactions between bodies is described by our function $\k$, then the problem solved by FMM coincides with our $\KAdjE$ problem. 

Most past work on FMM either considers the low-dimensional case, in which $d$ is a small constant, or else the low-error case, in which $\eps$ is a constant. Thus, much of the literature does not consider the simultaneous running time dependence of FMM on $\eps$ and $d$. In order to solve $\KAdjE$, we need to consider the high-dimensional, high-error case. We thus give a clean mathematical overview and detailed analysis of the running time of FMM in 
\ifdefined\isfocs
the full version,
\else Section~\ref{sec:fastmm},
\fi
following the seminal work of Greengard and Strain~\cite{gs91}, which may be of independent interest.  

As discussed in section~\ref{subsubsec:highdimmult} above, the running time of FMM depends exponentially on $d$, and so it is most useful in the low-dimensional setting. Our main algorithmic tool in high dimensions is a low-rank approximation technique: we show that when $f(x)$ can be approximated by a sufficiently low-degree polynomial (e.g. any degree $o(\log n)$ suffices in dimension $\Theta(\log n)$), then we can quickly find a low-rank approximation of the adjacency matrix $A_G$, and use this to efficiently multiply by a vector. Although this seems fairly simple, in Theorem~\ref{thm:informal-high} we show it is optimal: when such a low-rank approximation is not possible in high dimensions, then $\SETH$ implies that $n^{2 - o(1)}$ time is required for $\KAdjE$.

The simplest way to show that $f(x)$ can be approximated by a low-degree polynomial is by truncating its Taylor series. In fact, the FMM \emph{also} requires that a truncation of the Taylor series of $f$ gives a good approximation to $f$. By comparison, the FMM puts more lenient restrictions on what degree the series must be truncated to in low dimensions, but in exchange adds other constraints on $f$, including that $f$ must be monotone. \ifdefined\isfocs
See the full version 
\else
See Section~\ref{sec:fastmm_general} and Corollary~\ref{cor:kernelalg} 
\fi 
for more details.

\paragraph*{Lower Bound Techniques}

We now sketch the proof of Theorem~\ref{thm:informal-high}, our lower bound for $\KAdjE$ for many functions $\k$ in high enough dimensions (typically $d = \Omega(\log n)$), assuming $\SETH$. Although $\SETH$ is a hardness hypothesis about the Boolean satisfiability problem, a number of recent results~\cite{aw15,r18,c18,sm19} have showed that it implies hardness for a variety of \emph{nearest neighbor search} problems. Our lower bound approach is hence to show that $\KAdjE$ is useful for solving nearest neighbor search problems.

The high-level idea is as follows. Suppose we are given as input points $X = \{ x_1, \ldots, x_n\} \subset \{0,1\}^d$, and our goal is to find the closest pair of them. For each $\ell \in \{1,2,\ldots,d\}$, let $c_\ell$ denote the number of pairs of distinct points $x_i, x_j \in X$ with distance $\|x_i - x_j\|_2^2 = \ell$. Using an algorithm for $\KAdjE$ for our function $\k(x,y) = f(\| x-y \|_2^2)$, we can estimate
$$1^\top A_G 1 = \sum_{i \neq j} \k(x_i, x_j) = \sum_{\ell=1}^d c_\ell \cdot f(\ell).$$
Similarly, for any nonnegative reals $a,b \geq 0$, we can take an appropriate affine transformation of $X$ so that an algorithm for $\KAdjE$ can estimate
\begin{align}\label{introeq}\sum_{\ell=1}^d c_\ell \cdot f(a\cdot \ell + b).\end{align}

Suppose we pick real values $a_1, \ldots,a_d,b_1,\ldots,b_d \geq 0$ and define the $d \times d$ matrix $M$ by $M[i,\ell] = f(a_i \cdot \ell + b_i)$. By estimating the sum (\ref{introeq}) for each pair $(a_i,b_i)$, we get an estimate of the matrix-vector product $Mc$, where $c \in \mathbb{R}^d$ is the vector of the $c_\ell$ values. We show that if $M$ has a large enough determinant relative to the magnitudes of its entries, then one can recover an estimate of $c$ itself from this, and hence solve the nearest neighbor problem. 

The main tool we need for this approach is a way to pick $a_1, \ldots,a_d,b_1,\ldots,b_d$ for a function $f$ which cannot be approximated by a low degree polynomial so that $M$ has large determinant. We do this by decomposing $\det(M)$ in terms of the derivatives of $f$ using the Cauchy-Binet formula, and then noting that if $f$ cannot be approximated by a polynomial, then many of the contributions in this sum must be large. The specifics of this construction are quite technical; 
\ifdefined\isfocs
see the full version for the details.
\else
see section~\ref{sec:lbhighdim} for the details.
\fi

\paragraph*{Comparison with Previous Lower Bound Techniques}

Prior work (e.g. \cite{cs17}, \cite{bis17}, \cite{bcis18}) has shown $\SETH$-based fine-grained complexity results for matrix-related computations. For instance, \cite{bis17} showed hardness results for exact algorithms for many machine-learning related tasks, like kernel PCA and gradient computation in training neural networks, while \cite{cs17} and \cite{bcis18} showed hardness results for kernel density estimation. In all of this work, the authors are only able to show hardness for a limited set of kernels. For example, \cite{bis17} shows hardness for kernel PCA only for Gaussian kernels. These limitations arise from the technique used. To show hardness, \cite{bis17} exploits the fact that the Gaussian kernel decays rapidly to obtain a gap between the completeness and soundness cases in approximate nearest neighbors, just as we do for functions $f$ like $f(x) = (1/x)^{\Omega(\log n)}$. The hardness results of \cite{cs17} and \cite{bcis18} employ a similar idea.

As discussed in \emph{Lower Bound Techniques}, we circumvent these limitations by showing that applying the multiplication algorithm for one kernel a small number of times and linearly combining the results is enough to solve Hamming closest pair. This idea is enough to give a nearly tight characterization of the analytic kernels for which subquadratic-time multiplication is possible in dimension $d = \Theta(\log n)$. As a result, by combining with reductions similar to those from past work, our lower bound also applies to a variety of similar problems, including kernel PCA, for a much broader set of kernels than previously known; 
\ifdefined\isfocs
see the full version.
\else
see Section~\ref{subsec:kernelPCA} below for the details.
\fi

Our lower bound is also interesting when compared with the Online Matrix-Vector Multiplication (OMV) Conjecture of Henzinger et al.~\cite{henzinger2015unifying}. In the OMV problem, one is given an $n \times n$ matrix $M$ to preprocess, then afterwards one is given a stream $v_1, \ldots, v_n$ of length-$n$ vectors, and for each $v_i$, one must output $M \times v_i$ before being given $v_{i+1}$. The OMV Conjecture posits that one cannot solve this problem in total time $n^{3-\Omega(1)}$ for a general matrix $M$. 
At first glance, our lower bound may seem to have implications for the OMV Conjecture: For some kernels $\k$, our lower bound shows that for an input set of points $P$ and corresponding adjacency matrix $A_G$, and input vector $v_i$, there is no algorithm running in time $n^{2-\Omega(1)}$ for multiplying $A_G \times v_i$, so perhaps multiplying by $n$ vectors cannot be done in time $n^{3-\Omega(1)}$. However, this is not necessarily the case, since the OMV problem allows $O(n^{2.99})$ time for preprocessing $A_G$, which our lower bound does not incorporate. More broadly, the matrices $A_G$ which we study, which have very concise descriptions compared to general matrices, are likely not the best candidates for proving the OMV Conjecture. That said, perhaps our results can lead to a form of the OMV Conjecture for geometric graphs with concise descriptions.

\subsubsection{Sparsification}

\paragraph*{Algorithmic techniques}

Our algorithm for constructing high-dimensional sparsifiers for $\k(x,y) = f(\|x-y\|_2^2)$, when $f$ is a $(2, L)$ multiplicatively Lipschitz function, involves using three classic ideas: the Johnson Lindenstrauss lemma of random projection~\cite{jl84, im98}, the notion of well-separated pair decomposition from Callahan and Kosaraju~\cite{ck93, ck95}, and spectral sparsification via oversampling~\cite{ss11, kmp10}. Combining these techniques carefully gives us the bounds in Theorem~\ref{thm:informal-sparsify-lipschitz}.

To overcome the `curse of dimensionality', we use the Lindenstrauss lemma to project onto $\sqrt{L \log n}$ dimensions. This preserves all pairs distance, with a distortion of at most $2^{O(\sqrt{\log n/L})}$. Then, using a $1/2$-well-separated pair decomposition partitions the set of projected distances into bicliques, such that each biclique has edges that are no more than $2^{O(\sqrt{\log n/L})}$ larger than the smallest edge in the biclique. This ratio will upper bound the maximum leverage score of an edge in this biclique in the original $\k$-graph. 
Each biclique in the set of projected distances has a one-to-one correspondence to a biclique in the original $\k$-graph. Thus to sparsify our $\k$-graph, we sparsify each biclique in the $\k$-graph by uniform sampling, and take the union of the resulting sparsified bicliques. Due to the $(2, L)$-Lipschitz nature of our function, it is guaranteed that the longest edge in any biclique (measured using $\k(x,y)$) is at most $2^{O(\sqrt{L \log n})}$. This upper bounds the maximum leverage score of an edge in this biclique with respect to the $\k$-graph, which then can be used to upper bound the number of edges we need to sample from each biclique via uniform sampling.
We take the union of these sampled edges over all bicliques, which gives our results for high-dimensional sparsification summarized in Theorem~\ref{thm:informal-sparsify-lipschitz}.
\ifdefined\isfocs
Details can be found in the full version. 
\else
Details can be found in the proof of Theorem~\ref{thm:sparsify-lipschitz} in Section~\ref{sec:sparsify-lipschitz}.
\fi
When $L$ is constant, we get almost linear time sparsification algorithms.

For low dimensional sparsification, we skip the Johnson Lindenstrauss step, and use a $(1+1/L)$-well separated pair decomposition. This gives us a nearly linear time algorithm for sparsifying $(C, L)$ multiplicative Lipschitz functions, when $(2L)^{O(d)}$ is small, which covers the case when $d$ is constant and $L=n^{o(1)}$. 
\ifdefined\isfocs 
See the full version for details.
\else 
See Theorem~\ref{thm:low-sparsify-lipschitz} for details.
\fi

For $\k(u,v) = |\langle u,v\rangle|$, our sparsification algorithm is quite different from the multiplicative Lipschitz setting. In particular, the fact that $\k(u,v) = 0$ on a large family of pairs $u,v\in \mathbb{R}^d$ presents challenges. Luckily, though, this kernel does have some nice structure. For simplicity, just consider defining the $\k$-graph on a set of unit vectors. The weight of any edge in this graph is at most 1 by Cauchy-Schwarz. The key structural property of this graph is that for every set $S$ with $|S| > d+1$, there is a pair $u,v\in S$ for which the $u$-$v$ edge has weight at least $\Omega(1/d)$. In other words, the unweighted graph consisting of edges with weight between $\Omega(1/d)$ and 1 does not have independent sets with size greater than $d+1$. It turns out that all such graphs are dense  \ifdefined\isfocs(see the full version)\else(see Proposition \ref{prop:k-dep-dense})\fi and that all dense graphs have an expander subgraph consisting of a large fraction of the vertices \ifdefined\isfocs(see the full version)\else(see Proposition \ref{prop:dense-has-expander})\fi. Thus, if this expander could be found in $O(n)$ time, we could partition the graph into expander clusters, sparsify the expanders via uniform sampling, and sparsify the edges between expanders via uniform sampling. It is unclear to us how to identify this expander efficiently, so we instead identify clusters with effective resistance diameter $O(\poly(d\log n)/n)$. This can be done via uniform sampling and Johnson-Lindenstrauss \cite{ss11}. As part of the proof, we prove a novel Markov-style lower bound on the probability that effective resistances deviate too low in a randomly sampled graph, which may be of independent interest  
\ifdefined\isfocs
(see the full version).
\else(see Lemma \ref{lem:sparse-lower-tail}).
\fi

\paragraph*{Lower bound techniques}

To prove lower bounds on sparsification for decreasing functions that are not $(C_L,L)$-multiplicatively Lipschitz, we reduce from exact bichromatic nearest neighbors on two sets of points $A$ and $B$. In high dimensions, nearest neighbors is hard even for Hamming distance \cite{r18}, so we may assume that $A,B\subseteq \{0,1\}^d$. In low dimensions, we may assume that the coordinates of points in $A$ and $B$ consist of integers on at most $O(\log n)$ bits. In both cases, the set of possible distances between points in $A$ and $B$ is discrete. We take advantage of the discrete nature of these distance sets to prove a lower bound. In particular, $C_L$ is set so that $C_L$ is the smallest ratio between any two possible distances between points in $A$ and $B$. To see this in more detail, see 
\ifdefined\isfocs
the full version.
\else
Lemma \ref{lem:spars-reduc}.
\fi

Let $x_f$ be a point at which the function $f$ is not $(C_L,L)$-multiplicatively Lipschitz and suppose that we want to solve the decision problem of determining whether or not $\min_{a\in A, b\in B} \|a - b\|_2\le k$. We can do this using sparsification by scaling the points in $A$ and $B$ by a factor of $k/x_f$, sparsifying the $\k$-graph on the resulting points, and thresholding based on the total weight of the resulting $A$-$B$ cut. If there is a pair with distance at most $k$, there is an edge crossing the cut with weight at least $f(x_f)$ because $f$ is a decreasing function. Therefore, the sparsifier has total weight at least $f(x_f)/(1+\epsilon) = f(x_f)/2$ crossing the $A$-$B$ cut by the cut sparsification approximation guarantee. If there is not a pair with distance at most $k$, no edges crossing the cut with weight larger than $f(C_L x_f)\le C_L^{-L} f(x_f)\le ( 1/n^{10} ) \cdot f(x_f)$ by choice of $C_L$. Therefore, the total weight of the $A$-$B$ cut is at most $(1/n^8 ) \cdot f(x_f)$, which means that it is at most $( (1 + \epsilon)/n^8 ) \cdot f(x_f) < f(x_f)/4$ in the sparsifier. In particular, thresholding correctly solves the decision problem and one sparsification is enough to solve bichromatic nearest neighbors.

\subsection{Brief summary of our results in terms of \texorpdfstring{$p_f$}{}}

Before proceeding to the body of the paper, we summarize our results. Recall that we consider three linear-algebraic problems in this paper, along with two different dimension settings (low and high). This gives six different settings to consider. We now define $p_f$ in each of these settings. In all high-dimensional settings, we have found a definition of $p_f$ that characterizes the complexity of the problem. In some low-dimensional settings, we do not know of a suitable definition for $p_f$ and leave this as an open problem. For simplicity, we focus here only on decreasing functions $f$, although all of our algorithms, and most of our hardness results,  hold for more general functions as well.

\begin{table*}[!ht]
\centering
\begin{tabular}{|l|l|l|l|}
    \hline
    {\bf Dimension} & {\bf Multiplication} & {\bf Sparsification} & {\bf Solving} \\ \hline
     $d = \poly(\log n)$ & $f_1$ & $f_1,f_2$ & $f_1$\\ \hline
     $c^{\log^* n} < d < O(\log^{1-\delta} n)$ for  $\delta > 0$ & $f_1,f_2,f_3$ & $f_1,f_2,f_3$ & $f_1,f_2,f_3$ \\ \hline
\end{tabular}\caption{Functions among $f_1,f_2,f_3,f_4$ that have almost-linear time algorithms}\label{tab:examples}
\end{table*}

\begin{compactenum}
    \item Adjacency matrix-vector multiplication
    \begin{enumerate}
        \item High dimensions: $p_f$ is the minimum degree of any polynomial that $1/2^{\text{poly}(\log n)}$-additively approximates $f$. $p_f > \Omega(\log n)$ implies subquadratic-time hardness (Theorem \ref{thm:informal-high} part 2), while $p_f < o(\log n)$ implies an almost-linear time algorithm (Theorem \ref{thm:informal-high}, part 1).
        \item Low dimensions: Not completely understood. The fast multipole method yields an almost-linear time algorithm for some functions, like the Gaussian kernel \ifdefined\isfocs(see the full version)\else(Theorem \ref{thm:fast_gaussian_transform_intro})\fi, but functions exist that are hard in low dimensions \ifdefined\isfocs(see the full version).
        \else(Proposition \ref{prop:low-mult-hard}).\fi
    \end{enumerate}
    \item Sparsification. In both settings, $p_f$ is the minimum value for which $f$ is $(C, p_f)$-multiplicatively Lipschitz, where $C = 1 + 1/p_f^c$ for some constant $c > 0$ independent of $f$.
    \begin{enumerate}
        \item High dimensions: If $p_f > \Omega(\log^2 n)$ and $f$ is nonincreasing, then no subquadratic time algorithm exists (Theorem \ref{thm:informal-high-spars-hard}). If $p_f < o(\log n)$, then an almost-linear time algorithm for sparsification exists (Theorem \ref{thm:informal-sparsify-lipschitz}).
        \item Low dimensions: There is some constant $t > 1$ such that if $p_f > \Omega(n^t)$ and $f$ is nonincreasing, then no subquadratic time algorithm exists \ifdefined\isfocs(see the full version). \else(Theorem \ref{thm:informal-low-spars-hard}).\fi If $p_f < n^{o(1/d)}$, then there is a subquadratic time algorithm \ifdefined\isfocs(see the full version)\else(Theorem \ref{thm:informal-low-sparsify-lipschitz})\fi.
    \end{enumerate}
    \item Laplacian solving.
    \begin{enumerate}
        \item High dimensions: $p_f$ is the maximum of the $p_f$ values in the \emph{Adjacency matrix-vector multiplication} and \emph{Sparsification} settings, with hardness occurring for decreasing functions $f$ if $p_f > \Omega(\log^2 n)$ (Corollary \ref{cor:introsystemhardhighdim} combined with Theorem \ref{thm:informal-high-lsolve-hard}) and an algorithm existing when $p_f < o(\log n)$ (Theorem \ref{thm:informal-lapl-poly}).
        \item Low dimensions: Not completely understood, as in the low-dimensional multiplication setting. As in the sparsification setting, we are able to show that there is a constant $t$ such that if $f$ is nonincreasing and $p_f > \Omega(n^t)$ where $p_f$ is defined as in the \emph{Sparsifiction} setting, then no subquadratic time algorithm exists \ifdefined\isfocs(see the full version).
        \else
        (Theorem \ref{thm:low-lsolve-hard}).
        \fi
    \end{enumerate}
\end{compactenum}

Many of our results are not tight for two reasons: (a) some of the hardness results only apply to decreasing functions, and (b) there are gaps in $p_f$ values between the upper and lower bounds. However, neither of these concerns are important in most applications, as (a) weight often decreases as a function of distance and (b) $p_f$ values for natural functions are often either very low or very high. For example, $p_f > \Omega(\text{polylog}(n))$ for all problems for the Gaussian kernel ($f(x) = e^{-x}$), while $p_f = O(1)$ for sparsification and $p_f > \Omega(\text{polylog}(n))$ for multiplication for the gravitational potential ($f(x) = 1/x$). Resolving the gap may also be difficult, as for intermediate values of $p_f$, the true best running time is likely an intermediate running time of $n^{1 + c + o(1)}$ for some constant $0<c<1$. Nailing down and proving such a lower bound seems beyond the current techniques in fine-grained complexity.

\subsection{Summary of our Results on Examples} \label{sec:resultstable}

To understand our results better, we illustrate how they apply to some examples. For each of the functions $f_i$ given below, make the $\k$-graph, where $\k_i(u,v) = f_i(\|u-v\|_2^2)$:

\begin{compactenum}
\item $f_1(z) = z^k$ for a positive integer constant $k$.
\item $f_2(z) = z^c$ for a negative constant or a positive non-integer constant $c$.
\item $f_3(z) = e^{-z}$ (the Gaussian kernel).
\item $f_4(z) = 1$ if $z\le \theta$ and $f_4(z) = 0$ if $z > \theta$ for some parameter $\theta > 0$ (the threshold kernel).
\end{compactenum}

\josh{We need to put $f_5$ in the appropriate places in the table}

In Table \ref{tab:examples}, we summarize for which of the above functions there are efficient algorithms and for which we have hardness results. There are six regimes, corresponding to three problems (multiplication, sparsification, and solving) and two dimension regimes ($d = \poly(\log n)$ and $d = c^{\log^* n}$). A function is placed in a table cell if an almost-linear time algorithm exists, where runtimes are $n^{1+o(1)}(\log(\alpha n/\epsilon))^t$ in the case of multiplication and system solving and $n^{1+o(1)}(\log(\alpha n))^t/\epsilon^2$ in the case of sparsification for some $t\le O(\log^{1-\delta} n)$ for some $\delta > 0$. Moreover, for each of these functions $f_1, f_2, f_3,$ and $f_4$, if it does not appear in a table cell, then we show a lower bound that no subquadratic time algorithm exists in that regime assuming $\SETH$.

\josh{Note: we don't care if the stuff below is on page 11 or later}

\subsection{Other Related Work} \label{subsec:relatedwork}

\paragraph*{Linear Program Solvers}

Linear Program is a fundamental problem in convex optimization. There is a long list of work focused on designing fast algorithms for linear program \cite{d47,k80,k84,v87,v89_lp,ls14,ls15,cls19,lsz19,song19,b20,blss20,sy20,jswz20}. For the dense input matrix, the state-of-the-art algorithm \cite{jswz20} takes $n^{\max\{\omega,2+1/18\}} \log (1/\eps)$ time, $\omega$ is the exponent of matrix multiplication \cite{aw21}. The solver can run faster when matrix $A$ has some structures, e.g. Laplacian matrix. 

\paragraph*{Laplacian System Solvers}
It is well understood that a Laplacian linear system can be solved in time $\tilde{O} (m \log ( 1 / \eps ) )$, where $m$ is the number of edges in the graph generating the Laplacian~\cite{st04,kmp10,kmp11,kosz13,ls13,ckmpprx14,klpss16,ks16}. This algorithm is very efficient when the graph is sparse. However, in our setting where the $\k$ graph is dense but succinctly describable by only $n$ points in $\R^d$, we aim for much faster algorithms. 

\paragraph*{Algorithms for Kernel Density Function Approximation} A recent line of work by Charikar et al.~\cite{cs17,bcis18} also studies the algorithmic KDE problem. They show, among other things, that kernel density functions for ``smooth'' kernels $\k$ can be estimated in time which depends only polynomially on the dimension $d$, but which depends polynomially on the error $\eps$. We are unfortunately unable to use their algorithms in our setting, where we need to solve $\KAdjE$ with $\eps = n^{-\Omega(1)}$, and the algorithms of Charikar et al. do not run in subquadratic time. We instead design and make use of algorithms whose running times have only polylogarithmic dependences on $\eps$, but often have exponential dependences on $d$.

\paragraph{Kernel Functions}
Kernel functions are useful functions in data analysis, with applications in physics, machine learning, and computational biology~\cite{s10}. There are many kernels studied and applied in the literature; we list here most of the popular examples.

The following kernels are of the form $\k(x,y) = f(\|x-y\|_2^2)$, which we study in this paper: the Gaussian kernel \cite{njw02,rr08}, exponential kernel, Laplace kernel \cite{rr08}, rational quadratic kernel, multiquadric kernel \cite{bg97}, inverse multiquadric kernel \cite{m84,m12}, circular kernel \cite{btfb05}, spherical kernel, power kernel \cite{fs03}, log kernel \cite{bg97,m12}, Cauchy kernel \cite{rr08}, and generalized T-Student kernel \cite{btf04}.

For these next kernels, it is straightforward that their corresponding graphs have low-rank adjacency matrices, and so efficient linear algebra is possible using the Woodbury Identity \ifdefined\isfocs(see the full version)\else(see Section~\ref{sec:woodbury_identity} below)\fi: the linear kernel \cite{ssa98,mssmsr99,h07,s14} and the polynomial kernel \cite{cv95,ge08,bossr08,chcrl10}.

Finally, the following relatively popular kernels are not of the form we directly study in this paper, and we leave extending our results to them as an important open problem: the Hyperbolic tangent (Sigmoid) kernel \cite{hs97,btb05,jkl09,ksh12,zsjbd17,zsd17,ssbcv17,zsjd19}, spline kernel \cite{g98,u99}, B-spline kernel \cite{h04,m10}, Chi-Square kernel \cite{vz12}, and the histogram intersection kernel and generalized histogram intersection \cite{btb05a}. More interestingly, our result also can be applied to Neural Tangent Kernel \cite{jgh18}, which plays a crucial role in the recent work about convergence of neural network training \cite{ll18,dzps19,als18,als19,sy19,bpsw20,lsswy20,jmsz20}. For more details, we refer the readers to
\ifdefined\isfocs
the full version.
\else
Section~\ref{sec:ntk}.
\fi

\paragraph{Acknowledgements}

The authors would like to thank Lijie Chen for helpful suggestions in the hardness section and explanation of his papers. The authors would like to thank Sanjeev Arora, Simon Du, and Jason Lee for useful discussions about the neural tangent kernel.



\ifdefined\isfocs

\else
\section{Summary of Low Dimensional Results}\label{sec:low-dim-summary}

In the results we've discussed so far, we show that in high-dimensional settings, the curse of dimensionality applies to a wide variety of functions that are relevant in applications, including the Gaussian kernel and inverse polynomial kernels. Luckily, in many settings, the points supplied as input are very low-dimensional. In the classic $n$-body problem, for example, the input points are 3-dimensional. In this subsection, we discuss our results pertaining to whether algorithms with runtimes exponential in $d$ exist; such algorithms can still be efficient in low dimensions $d = o(\log n)$.

\subsection{Multiplication}

The prior work on the fast multipole method ~\cite{gr87, gr88, gr89} yields algorithms with runtime $(\log(n/\epsilon))^{O(d)} n^{1+o(1)}$ for $\epsilon$-approximate adjacency matrix-vector multiplication for a number of functions $\k$, including when $\k(u,v) = \frac{1}{\|u - v\|_2^c}$ for a constant $c$ and when $\k(u,v) = e^{-\|u-v\|_2^2}$. In order to explain what functions $\k$ the fast multipole methods work well for, and to clarify dependencies on $d$ in the literature, we give a complete exposition of how the fast multipole method of~\cite{gs91} works on the Gaussian kernel:

\begin{theorem}[fast Gaussian transform~\cite{gs91}, exposition in Section~\ref{sec:fastmm}]\label{thm:fast_gaussian_transform_intro}
Let $\k(x,y) = \exp (- \| x - y \|_2^2)$. Given a set of points $P \subset \R^d$ with $|P|= n$. Let $G$ denote the $\k$-graph. For any vector $u \in \R^d$, for accuracy parameter $\epsilon$, 
there is an algorithm that runs in 
$
n \log^{O(d)} ( \| u \|_1 / \epsilon ) 
$
time to approximate $A_G \cdot u$ within $\epsilon$ additive error.
\end{theorem}

The fast multipole method is fairly general, and so similar algorithms also exist for a number of other functions $\k$; see Section~\ref{sec:fastmm_general} for further discussion. 
Unlike in the high-dimensional case, we do not have a characterization of the functions for which almost-linear time algorithms exist in near-constant dimension. We leave this as an open problem. Nonetheless, we are able to show lower bounds, even in barely super-constant dimension $d = \exp(\log^*(n))$\footnote{Here, $\log^*(n)$ denotes the \emph{very} slowly growing iterated logarithm of $n$.}, on adjacency matrix-vector multiplication for kernels that are not multiplicatively Lipschitz:

\begin{theorem}[Informal version of Proposition \ref{prop:low-mult-hard}]
For some constant $c>1$, any function $f$ that is not $(C,L)$-multiplicatively Lipschitz for any constants $C > 1, L > 1$  does not have an $n^{1+o(1)}$ time adjacency matrix-vector multiplication (up to $2^{-\poly(\log n)}$ additive error) algorithm in $c^{\log^* n}$ dimensions assuming $\SETH$ when $\k(u,v) = f(\| u - v \|_2^2)$.
\end{theorem}

This includes threshold functions, but does not include piecewise exponential functions. Piecewise exponential functions do have efficient adjacency multiplication algorithms by Theorem \ref{thm:fast_gaussian_transform}.

To illustrate the complexity of the adjacency matrix-vector multiplication problem in low dimensions, we are also able to show hardness for the function $\k(u,v) = |\langle u,v\rangle|$ in nearly constant dimensions. By comparison, we are able to sparsify for this function $\k$, even in very high $d = n^{o(1)}$ dimensions (in Theorem~\ref{thm:sparsabsinnerproduct} above).
\begin{theorem}[Informal version of Corollary \ref{cor:lowdimmultabsinnerproduct}] \label{thm:introlowdimmultabsinnerproduct}
For some constant $c>1$, assuming $\SETH$, adjacency matrix-vector multiplication (up to $2^{-\poly(\log n)}$ additive error) in $c^{\log^* n}$ dimensions cannot be done in subquadratic time in dimension $d = \exp(\log^*(n))$ when $\k(u,v) = |\langle u,v\rangle|$.
\end{theorem}

\subsection{Sparsification}

We are able to give a characterization of the decreasing functions for which sparsification is possible in near-constant dimension. We show that a polynomial dependence on the multiplicative Lipschitz constant is allowed, unlike in the high-dimensional setting:
\begin{theorem}[Informal version of Theorem \ref{thm:low-sparsify-lipschitz}]\label{thm:informal-low-sparsify-lipschitz}
Let $f$ be a $(1 + 1/L,L)$-multiplicatively Lipschitz function and let $\k(u,v) = f(\|u-v\|_2^2)$. Then an $(1\pm\epsilon)$-spectral sparsifier for the $\k$-graph on $n$ points can be found in $ n^{1+o(1)} L^{O(d)} (\log \alpha) /\epsilon^2$ time. 
\end{theorem}

Thus, geometric graphs for piecewise exponential functions with $L = n^{o(1)}$ can be sparsified in almost-linear time when $d$ is constant, unlike in the case when $d = \Omega(\log n)$. In particular, spectral clustering can be done in $O(kn^{1+o(1)})$ time for $k$ clusters in low dimensions. Unfortunately, not all geometric graphs can be sparsified, even in nearly constant dimensions:

\begin{theorem}[Informal version of Theorem \ref{thm:low-spars-hard}]\label{thm:informal-low-spars-hard}
There are constants $c'\in (0,1),c > 1$ and a value $C_L$ given $L > 1$ for which any decreasing function $f$ that is not $(C_L,L)$-multiplicatively Lipschitz does not have an $O(n L^{c'})$ time sparsification algorithm for $\k$-graphs on $c^{\log^*n}$ dimensional points, where $\k(u,v) = f(\|u-v\|_2^2)$.
\end{theorem}

This theorem shows, in particular, that geometric graphs of threshold functions are not sparsifiable in subquadratic time even for low-dimensional pointsets. These two theorems together nearly classify the decreasing functions for which efficient sparsification is possible, up to the exponent on $L$.

\subsection{Laplacian solving}

As in the case of multiplication, we are unable to characterize the functions for which solving Laplacian systems can be done in almost-linear time in low dimensions. That said, we still have results for many functions $\k$, including most kernel functions of interest in applications. We prove most of these using the aforementioned connection from Section~\ref{sec:equivalence}: if a $\k$ graph can be efficiently sparsified, then there is an efficient Laplacian multiplier for $\k$ graphs if and only if there is an efficient Laplacian system solver for $\k$ graphs.

For the kernels $\k(u,v) = 1/\|u-v\|_2^c$ for constants $c$ and the piecewise exponential kernel, we have almost-linear time algorithms in low dimensions by Theorems \ref{thm:sparsify-lipschitz} and \ref{thm:low-sparsify-lipschitz} respectively. Furthermore, the fast multipole method yields almost-linear time algorithms for multiplication. Therefore, there are almost-linear time algorithm for solving Laplacian systems in geometric graphs for these kernels.

A similar approach also yields hardness results. Theorem~\ref{thm:sparsabsinnerproduct} above implies that an almost-linear time algorithm for solving Laplacian systems on $\k$-graphs for $\k(u,v) = |\langle u,v\rangle |$ yields an almost-linear time algorithm for $\k$-adjacency multiplication. However, no such algorithm exists assuming $\SETH$ by Theorem~\ref{thm:introlowdimmultabsinnerproduct} above. Therefore, $\SETH$ implies that no almost-linear time algorithm for solving Laplacian systems in this kernel can exist.

We directly (i.e. without using a sparsifier algorithm) show an additional hardness result for solving Laplacian systems for kernels that are not multiplicatively Lipschitz, like threshold functions of $\ell_2$-distance:

\begin{theorem}[Informal version of Theorem \ref{thm:low-lsolve-hard}]
Consider an $L > 1$. There is some sufficiently large value $C_L > 1$ depending on $L$ such that for any decreasing function $f:\mathbb{R}_{\ge 0}\rightarrow \mathbb{R}_{\ge 0}$ that is not $(C_L,L)$-multiplicatively Lipschitz, no $O(n L^{c'} \log \alpha)$-time algorithm exists for solving Laplacian systems $2^{-\poly(\log n)}$ approximately in the $\k$-graph of a set of $n$ points in $c^{\log^* n}$ dimensions for some constants $c > 1, c'\in (0,1)$ assuming $\SETH$, where $\k(u,v) = f(\|u-v\|_2^2)$.
\end{theorem} 
\section{Preliminaries}\label{sec:preli}

Our results build off of algorithms and hardness results from many different areas of theoretical computer science. We begin by defining the relevant notation and describing the important past work.

\subsection{Notation}

For an $n\in \mathbb{N}_{+}$, let $[n]$ denote the set $\{1,2,\cdots,n\}$.

For any function $f$, we write $\wt{O}(f)$ to denote $f\cdot \log^{O(1)}(f)$. In addition to $O(\cdot)$ notation, for two functions $f,g$, we use the shorthand $f\lesssim g$ (resp. $\gtrsim$) to indicate that $f\leq C g$ (resp. $\geq$) for an absolute constant $C$. We use $f\eqsim g$ to mean $cf\leq g\leq Cf$ for constants $c,C$.

For a matrix $A$, we use $\|A\|_2$ to denote the spectral norm of $A$. Let $A^\top$ denote the transpose of $A$. Let $A^\dagger$ denote the Moore-Penrose pseudoinverse of $A$. Let $A^{-1}$ denote the inverse of a full rank square matrix. 

We say matrix $A$ is positive semi-definite (PSD) if $A = A^\top$ and $x^\top A x \geq 0$ for all $x \in \R^n$. We use $\preceq$, $\succeq$ to denote the semidefinite ordering, e.g. $A \succeq 0$ denotes that $A$ is PSD, and $A \succeq B $ means $A- B \succeq 0$. We say matrix $A$ is positive definite (PD) if $A = A^\top$ and $x^\top A x > 0$ for all $x \in \R^n -\{0\}$. $A \succ B$ means $A-B$ is PD.

For a vector $v$, we denote $\| v \|_p$ as the standard $\ell_p$ norm. For a vector $v$ and PSD matrix $A$, we let $\| v \|_A = (v^\top A v)^{1/2}$.

The iterated logarithm $\log^* : \R \to \Z$ is given by \begin{align*}
    \log^*(n) = \begin{cases}
    0, &\text{ if } n \leq 1 ; \\
    1 + \log^*(\log n), &\text{ otherwise.}
    \end{cases}
\end{align*}

We use $G_{\text{grav}}$ to denote the Gravitational constant.

We define $\alpha$ slightly differently in different sections. Note that both are less than the value of $\alpha$ used in Theorem \ref{thm:informal-sparsify-lipschitz}: 
\begin{table}[!h]\caption{}
\centering
\begin{tabular}{|l|l|l|}
    \hline
    Notation & Meaning & Location \\ \hline
    $\alpha$ & $\frac{ \max_{i,j} f(\| x_i - x_j \|_2^2) }{ \min_{i,j} f( \| x_i - x_j \|_2^2 ) }$  & Section~\ref{sec:equivalence}  \\ \hline
    $\alpha$ & $\frac{\max_{i,j} \| x_i - x_j \|_2 }{\min_{i,j} \| x_i - x_j \|_2}$ & Section~\ref{sec:sparsify-lipschitz} \\ \hline 
\end{tabular}
\end{table}

\subsection{Graph and Laplacian Notation}\label{sec:graph_Laplacian_notation}

Let $G = (V,E,w)$ be a connected weighted undirected graph with $n$ vertices and $m$ edges and edge weights $w_e > 0$. We say $r_e = 1/w_e$ is the resistance of edge $e$.
If we give a direction to the edges of $G$ arbitrarily, we can write its Laplacian as $L_G = B^\top W B$, where $W \in \R^{m \times m}$ is the diagonal matrix $W(e,e) = w_e$ and $B \in \R^{m \times n}$ is the signed edge-vertex incidence matrix and can be defined in the following way
\begin{align}\label{eq:def_Laplacian_B}
B(e,v) = \begin{cases}
1,  & \text{~if~$v$~is~$e$'s~head}; \\
-1, & \text{~if~$v$~is~$e$'s~tail};\\
0,  & \text{~otherwise.}
\end{cases}
\end{align}

A useful notion related to Laplacian matrices is the effective resistance of a pair of nodes:
\begin{definition}[Effective resistance]\label{def:effective_resistance}
The effective resistance of a pair of vertices $u,v \in V_G$ is defined as 
\begin{align*}
\Reff_G(u,v) = b_{u,v}^\top L^\dagger b_{u,v}
\end{align*}
where $b_{u,v} \in \R^{|V_G|}$ is an all zero vector except for entries of $1$ at $u$ and $-1$ at $v$.
\end{definition}

Using effective resistance, we can define leverage score
\begin{definition}[Leverage score]\label{def:leverage_score}
The leverage score of an edge $e = (u,v) \in E_G$ is defined as 
\begin{align*}
l_e = w_e \cdot \Reff_G(u,v).
\end{align*}
\end{definition}

We define a useful notation called electrical flow
\begin{definition}[Electrical flow]\label{def:electrical_flow}
Let $B \in \R^{m \times n}$ be defined as Eq.~\eqref{eq:def_Laplacian_B}, for a given demand vector $d \in \R^n$, we define electrical flow $f \in \R^m$ as follows:
\begin{align*}
f = \arg\min_{ f : B^\top f = d } \sum_{ e \in E_G } f_e^2 /w_e .
\end{align*}
\end{definition}

We let $d(i)$ denote the degree of vertex $i$. For any set $S \subseteq V$, we define volume of $S$: $\mu(S) = \sum_{i \in S} d(i)$. It is obvious that $\mu(V) = 2|E|$. For any two sets $S, T \subseteq V$, let $E(S,T)$ be the set of edges connecting a vertex in $S$ with a vertex in $T$. We call $\Phi(S)$ to be the conductance of a set of vertices $S$, and can be formally defined as
\begin{align*}
\Phi(S) = \frac{ | E(S, V \setminus S) | }{ \min ( \mu (S) , \mu ( V \setminus S ) ) }.
\end{align*}
We define the notation conductance, which is standard in the literature of graph partitioning and graph clustering \cite{st04,kvv04,acl06,ap09,lrtv11,zlm13,lz14}.
\begin{definition}[Conductance]\label{def:conductance}
The conductance of a graph $G$ is defined as follows:
\begin{align*}
\Phi_G = \min_{ S \subset V } \Phi(S).
\end{align*}
\end{definition}

\begin{lemma}[\cite{st04,aalg17}]\label{lem:conductance_cut}
A graph $G$ with minimum conductance $\Phi_G$ has the property that for every pair of vertices $u,v$,
\begin{align*}
\Reff_G(u,v)\le O \left( \Big( \frac{1}{c_u} + \frac{1}{c_v} \Big) \cdot \frac{1}{\Phi_G^2} \right)
\end{align*}
where $c_u$ is the sum of the weights of edges incident with $u$. Furthermore, for every pair of vertices $u,v$,
$$\Reff_G(u,v)\ge \max(1/c_u,1/c_v)$$
\end{lemma}

For a function $\k:\mathbb{R}^d\times \mathbb{R}^d\rightarrow \mathbb{R}$, the $\k$-graph on a set of points $X\subseteq \mathbb{R}^d$ is the graph with vertex set $X$ and edge weights $\k(u,v)$ for $u,v\in X$. For a function $f\mathbb{R}_{\ge 0}\rightarrow \mathbb{R}_{\ge 0}$, the $f$-graph on a set of points $X$ is defined to be the $\k$ graph on $X$ for $\k(u,v) = f(\|u-v\|_2)$.

\subsection{Spectral Sparsification via Random Sampling}

Here, we state some well known results on spectral sparsification via random sampling, from previous works. The theorems below are essential for our results on sparsifying geometric graphs quickly.

\begin{theorem}[Oversampling \cite{kmp11}]\label{thm:oversampling}
Consider a graph $G = (V,E)$ with edge weights $w_e > 0$ and probabilities $p_e\in (0,1]$ assigned to each edge and parameters $\delta \in (0,1),\epsilon\in (0,1)$. Generate a reweighted subgraph $H$ of $G$ with $q$ edges, with each edge $e$ sampled with probability $p_e/t$ and added to $H$ with weight $w_e t/(p_e q)$, where $t = \sum_{e \in E} p_e$. If
\begin{enumerate}
    \item $q\ge C \cdot \epsilon^{-2} \cdot t\log t \cdot \log (1/\delta)$, where $C > 1$ is a sufficiently large constant
    \item $p_e\ge w_e \cdot \Reff_G(u,v)$ for all edges $e = \{u,v\}$ in $G$
\end{enumerate}
then $(1 - \epsilon) L_G \preceq L_H \preceq (1 + \epsilon) L_G$ with probability at least $1 - \delta$.
\end{theorem}

\begin{algorithm}
\begin{algorithmic}[1]\caption{}
\Procedure{\textsc{Oversampling}}{$G,w,p,\epsilon,\delta$} \Comment{Theorem~\ref{thm:oversampling}}
    \State $t \leftarrow \sum_{e \in E} p_e$
    \State $q \leftarrow C \cdot \epsilon^{-2} \cdot t \log t \cdot \log(1/\delta)$
    \State Initialize $H$ to be an empty graph
    \For{$i = 1 \to q$}
        \State Sample one $e \in E$ with probability $p_e/t$
        \State Add that edge with weight $w_e t / (p_e q)$ to graph $H$
    \EndFor
    \State \Return $H$
\EndProcedure
\end{algorithmic}
\end{algorithm}

\begin{theorem}[\cite{ss11} effective resistance data structure]\label{thm:ss11}
There is a $\tilde{O}(m(\log \alpha)/\eps^2)$ time algorithm which on input $\eps > 0$ and $G = (V,E,w)$ with $\alpha = w_{\max}/w_{\min}$ computes a $(24 \log n/\eps^2)\times n$ matrix $\tilde{Z}$ such that with probability at least $1 - 1/n$,

$$(1 - \eps) \Reff_G(u,v) \le \|\tilde{Z}b_{uv}\|_2^2 \le (1 + \eps) \Reff_G(u,v)$$
for every pair of vertices $u,v\in V$.
\end{theorem}

The following is an immediate corollary of Theorems \ref{thm:oversampling} and \ref{thm:ss11}:

\begin{corollary}[\cite{ss11}]\label{cor:ss11}
There is a $\tilde{O}(m(\log \alpha)/\eps^2)$ time algorithm which on input $\eps > 0$ and $G = (V,E,w)$ with $\alpha = w_{\max}/w_{\min}$, produces an $(1 \pm \eps)$-approximate sparsifier for $G$.
\end{corollary}

\subsection{Woodbury Identity}\label{sec:woodbury_identity}

\begin{proposition}[\cite{w49,w50}]\label{prop:woodbury}
The Woodbury matrix identity is
\begin{align*}
(A+ U C V)^{-1} = A^{-1} - A^{-1} U (C^{-1} + V A^{-1} U)^{-1} V A^{-1}
\end{align*}
where $A, U, C$ and $V$ all denote matrices of the correct (conformable) sizes: For integers $n$ and $k$, $A$ is $n\times n$, $U$ is $n\times k$, $C$ is $k \times k$ and $V$ is $k \times n$.
\end{proposition}

The Woodbury identity is useful for solving linear systems in a matrix $M$ which can be written as the sum of a diagonal matrix $A$ and a low-rank matrix $UV$ for $k\ll n$ (setting $C=I$).

\subsection{Tail Bounds}

We will use several well-known tail bounds from probability theory.
\begin{theorem}[Chernoff Bounds \cite{c52}]\label{thm:chernoff}
Let $X = \sum_{i=1}^n X_i$, where $X_i=1$ with probability $p_i$ and $X_i = 0$ with probability $1-p_i$, and all $X_i$ are independent. Let $\mu = \E[X] = \sum_{i=1}^n p_i$. Then \\
1. $ \Pr[ X \geq (1+\delta) \mu ] \leq \exp ( - \delta^2 \mu / 3 ) $, $\forall \delta > 0$ ; \\
2. $ \Pr[ X \leq (1-\delta) \mu ] \leq \exp ( - \delta^2 \mu / 2 ) $, $\forall 0 < \delta < 1$. 
\end{theorem}

\begin{theorem}[Hoeffding bound \cite{h63}]\label{thm:hoeffding}
Let $X_1, \cdots, X_n$ denote $n$ independent bounded variables in $[a_i,b_i]$. Let $X= \sum_{i=1}^n X_i$, then we have
\begin{align*}
\Pr[ | X - \E[X] | \geq t ] \leq 2\exp \left( - \frac{2t^2}{ \sum_{i=1}^n (b_i - a_i)^2 } \right)
\end{align*}
\end{theorem}

\subsection{Fine-Grained Hypotheses}\label{sec:fine_grained}
\paragraph*{Strong Exponential Time Hypothesis}

Impagliazzo and Paturi \cite{ip01} introduced the Strong Exponential Time Hypothesis ($\SETH$) to address the complexity of CNF-SAT. Although it was originally stated only for deterministic algorithms, it is now common to extend $\SETH$ to randomized algorithms as well.

\begin{hypothesis}[Strong Exponential Time Hypothesis ($\SETH$)]
For every $\epsilon > 0$ there exists an integer $k \geq 3$ such that CNF-SAT on formulas with clause size at most $k$ (the so called $k$-SAT problem) and $n$ variables cannot be solved in $O(2^{(1-\epsilon)n})$ time even by a randomized algorithm.
\end{hypothesis}

\paragraph*{Orthogonal Vectors Conjecture}

The Orthogonal Vectors (OV) problem asks: given $n$ vectors $x_1, \cdots, x_n \in \{0,1\}^d$, are there $i,j$ such that $\langle v_i, v_j \rangle = 0$ (where the inner product is taken over $\Z$)? It is easy to see that $O(n^2 d)$ time suffices for solving OV, and slightly subquadratic-time algorithms are known in the case of small $d$ \cite{awy15,cw16}. It is conjectured that there is no OV algorithm running in $n^{1.99}$ time when $d = \omega (\log n)$.

\begin{conjecture}[Orthogonal Vectors Conjecture (OVC) \cite{w05,avw14}]
For every $\epsilon > 0$, there is a $c \geq 1$ such that OV cannot be solved in $n^{2-\epsilon}$ time on instances with $d = c \log n$.
\end{conjecture}

In particular, it is known that $\SETH$ implies OVC~\cite{w05}.
$\SETH$ and OVC are the most common hardness assumption in fine-grained complexity theory, and they are known to imply tight lower bounds for a number of algorithmic problems throughout computer science. See, for instance, the survey~\cite{williams2018some} for more background.

\subsection{Dimensionality Reduction}

We make use of the following binary version of the Johnson-Lindenstrauss lemma due to Achlioptas \cite{a03}:

\begin{theorem}[\cite{jl84,a03}]\label{thm:jl}
Given fixed vectors $v_1,\hdots,v_n\in \mathbb{R}^d$ and $\epsilon > 0$, let $Q \in \mathbb{R}^{k\times d}$ be a random $\pm 1/\sqrt{k}$ matrix (i.e. independent Bernoulli entries) with $k\ge 24 (\log n)/\epsilon^2$. Then with probability at least $1 - 1/n$,
\begin{align*}
(1 - \epsilon)\|v_i - v_j\|^2\le \|Qv_i - Qv_j\|^2\le (1 + \epsilon)\|v_i - v_j\|^2
\end{align*}
for all pairs $i,j\in [n]$.
\end{theorem}

We will also use the following variant of Johnson Lindenstrauss for Euclidean space, for random projections onto $o(\log n)$ dimensions:

\begin{lemma}[Ultralow Dimensional Projection \cite{jl84,dg03}, see Theorem 8.2 in \cite{s16} for example]
\label{lem:low-dim-jl}  
For $k = o(\log n)$, with high probability the maximum distortion in pairwise distance obtained from projecting $n$ points into $k$ dimensions (with appropriate scaling) is at most $n^{O(1/k)}$.
\end{lemma}

\subsection{Nearest Neighbor Search}\label{sec:ann}

Our results will make use of a number of prior results, both algorithms and lower bounds, for nearest neighbor search problems.

\paragraph*{Nearest Neighbor Search Data Structures}

\begin{problem}[\cite{a09,r17} data-structure $\ANN$]
Given an $n$-point dataset $P$ in $\R^d$ with $d = n^{o(1)}$, the goal is to preprocess it to answer the following queries. Given a query point $q \in \R^d$ such that there exists a data point within $\ell_p$ distance $r$ from $q$, return a data point within $\ell_p$ distance $cr$ from $q$.
\end{problem}

\begin{theorem}[\cite{ai06}]\label{thm:l2-ann}
There exists a data structure that returns a $2c$-approximation to the nearest neighbor distance in $\ell_2^d$ with preprocessing time and space $O_c(n^{1+1/c^2+o_c(1)} + nd)$ and query time $O_c(dn^{1/c^2+o_c(1)})$
\end{theorem}

\paragraph*{Hardness for Approximate Hamming Nearest Neighbor Search}

We provide the definition of the Approximate Nearest Neighbor search problem

\begin{problem}[monochromatic $\ANN$]\label{pro:single_ANN}
The monochromatic Approximate Nearest Neighbor ($\ANN$) problem is defined as : given a set of $n$ points $x_1, \cdots, x_n \in \R^d$ with $n^{o(1)}$, the goal is to compute $\alpha$-approximation of $\min_{i\neq j} \dist(x_i, x_j)$.
\end{problem}

\begin{theorem}[\cite{sm19}]\label{thm:sm19}
Let $\dist(x,y)$ be $\ell_p$ distance. Assuming ${\sf SETH}$, for every $\delta > 0$, there is a $\epsilon > 0$ such that the monochromatic $(1+\epsilon)$-${\sf ANN}$ problem for dimension $d = \Omega(\log n)$ requires time $n^{1.5-\delta}$.
\end{theorem}

\begin{problem}[bichromatic $\ANN$]\label{pro:hamming_ANN}
Let $\dist(\cdot, \cdot)$ denote the some distance function. Let $\alpha > 1$ denote some approximation factor. 
The bichromatic Approximate Nearest Neighbor ($\ANN$) problem is defined as: given two sets $A,B$ of vectors $\R^d$, the goal is to compute $\alpha$-approximation of $\min_{a\in A,b\in B} \dist( a , b )$.

\end{problem}

\begin{theorem}[\cite{r18}]\label{thm:r18}
Let $\dist(x,y)$ be any of Euclidean, Manhattan, Hamming($\| x - y \|_0$), and edit distance. Assuming ${\sf SETH}$, for every $\delta > 0$, there is a $\epsilon > 0$ such that the bichromatic $(1+\epsilon)$-${\sf ANN}$ problem for dimension $d = \Omega(\log n)$ requires time $n^{2-\delta}$.
\end{theorem}

By comparison, the best known algorithm for $d = \Omega(\log n)$ for each of these distance measures other than edit distance runs in time about $dn + n^{2 - \Omega(\eps^{1/3}/\log(1/\eps))}$ \cite{acw16}.

\paragraph*{Hardness for $\Z$-Max-IP}

\begin{problem}\label{pro:Z_maxip}
For $n,d \in \mathbb{N}$, the $\mathbb{Z}$-${\sf MaxIP}$ problem for dimension $d$ asks: given two sets $A,B$ of vectors from $\mathbb{Z}^d$, compute
\begin{align*}
\max_{ a \in A, b \in B } \langle a , b \rangle.
\end{align*}
\end{problem}

$\mathbb{Z}$-${\sf MaxIP}$ is known to be hard even when the dimension $d$ is barely superconstant:

\begin{theorem}[Theorem 1.14 in \cite{c18}] \label{thm:maxiphard}
Assuming ${\sf SETH}$ (or {\sf OVC}), there is a constant $c$ such that any exact algorithm for $\Z$-${\sf MaxIP}$ for $d = c^{\log^* n}$ dimensions requires $n^{2-o(1)}$ time, with vectors of $O(\log n)$-bit entries.
\end{theorem}

It is believed that $\Z$-${\sf MaxIP}$ cannot be solved in truly subquadratic time even in constant dimension~\cite{c18}. Even for $d=3$, the best known algorithm runs in $O(n^{4/3})$ time and has not been improved for decades:

\begin{theorem}[\cite{mat92,aesw91,y82}]\label{thm:maxiplowdim}
$\Z$-${\sf MaxIP}$ for $d=3$ can be solved in time $O(n^{4/3})$. For general $d$, it can be solved in $n^{2- \Theta(1/d)}$.
\end{theorem}

The closely related problem of $\ell_2$-nearest neighbor search is also hard in barely superconstant dimension:

\begin{theorem}[Theorem 1.16 in \cite{c18}]\label{thm:lownnhard}
Assuming ${\sf SETH}$ (or {\sf OVC}), there is a constant $c$ such that any exact algorithm for bichromatic $\ell_2$-closest pair for $d = c^{\log^* n}$ dimensions requires $n^{2-o(1)}$ time, with vectors of $c_0\log n$-bit entries for some constants $c > 1$ and $c_0 > 1$.
\end{theorem}

\subsection{Geometric Laplacian System}

Building off of a long line of work on Laplacian system solving~\cite{st04,kmp10,kmp11,kosz13,ckmpprx14}, we study the problem of solving geometric Laplacian systems:
\begin{problem}\label{pro:geometric_Laplacian_system}
Let $\k : \R^d \times \R^d \rightarrow \R$. Given a set of points $x_1, \cdots, x_n \in \R^d$, a vector $ b\in \R^n$ and accuracy parameter $\epsilon$. Let graph $G$ denote the graph that has $n$ vertices and each edge$(i,j)$'s weight is $\k(x_i,x_j)$. Let $L_G$ denote the Laplacian matrix of graph $G$. The goal is to output a vector $u \in \R^n$ such that
\begin{align*}
\| u - L_G^\dagger b \|_{L_G} \leq \epsilon \| L_G^\dagger b \|_{L_G} 
\end{align*}
where $L_G^\dagger$ denotes the pseudo-inverse of $L_G$ and matrix norm is defined as $\| c \|_A = \sqrt{ c^\top A c }$.
\end{problem}

\section{Equivalence of Matrix-Vector Multiplication and Solving Linear Systems} \label{sec:equivalence}

In this section, we show that for linear systems with sparse preconditioners, approximately solving them is equivalent to approximate matrix multiplication. We begin by formalizing our notion of approximation.

\begin{definition} \label{def:approxmult}
Given a matrix $M$ and a vector $x$, we say that $b$ is an $\epsilon$-\emph{approximate multiplication of $Mx$} if
\begin{align*}
(b - Mx)^\top M^{\dag}(b - Mx)\le \epsilon \cdot x^\top Mx.
\end{align*}
Given a vector $d$, we say that a vector $y$ is an $\epsilon$-\emph{approximate solution to $My = d$} if $y$ is an $\epsilon$-approximate multiplication of $M^{\dag}d$.
\end{definition}

Before stating the desired reductions, we state a folklore fact about the Laplacian norm:
\begin{proposition}[property of Laplacian norm]\label{prop:lapl-apx}
Consider a $w$-weighted $n$-vertex connected graph $G$ and let $w_{\min} = \min_{u,v\in G} w_{uv}$, $w_{\max} = \max_{u,v\in G} w_{uv}$, and $\alpha = w_{\max}/w_{\min}$. Then, for any vector $x\in \mathbb{R}^n$ that is orthogonal to the all ones vector,
$$ 
\frac{w_{\min}}{2n^4 \alpha^2} \|x\|_{\infty}^2 \le \|x\|_{L_G}^2 \le n^2 w_{\max} \|x\|_{\infty}^2
$$
and for any vector $b\in \mathbb{R}^n$ orthogonal to all ones,
$$
\frac{1}{n^2 w_{\max}} \|b\|_{\infty}^2 \le \|b\|_{L_G^{\dag}}^2 \le  \frac{2n^4 \alpha^2}{w_{\min}} \|b\|_{\infty}^2
$$
\end{proposition}

\begin{proof}
\textbf{Lower Bound for $L_G$}: Let $\lambda_{\min}$ and $\lambda_{\max}$ denote the minimum nonzero and maximum eigenvalues of $D_G^{-1/2}L_GD_G^{-1/2}$ respectively, where $D_G$ is the diagonal matrix of vertex degrees. Since $G$ is connected, all cuts have conductance at least $w_{\min}/(n^2 w_{\max}) = 1/(n^2\alpha)$. Therefore, by Cheeger's Inequality \cite{c70}, $\lambda_{\min} \ge 1/(2n^4 \alpha^2)$. It follows that,
\begin{align*}
\|x\|_{L_G}^2 &= x^{\top} L_G x\\
&\ge (x^{\top} D_G x)/(2n^4 \alpha^2)\\
&\ge \|x\|_{\infty}^2 w_{\min}/(2n^4 \alpha^2)
\end{align*}
as desired.

\textbf{Upper bound for $L_G$}: $\lambda_{\max}\le 1$. Therefore,
\begin{align*}
\|x\|_{L_G}^2 
\le x^{\top} D_G x 
\le n^2 w_{\max} \|x\|_{\infty}^2
\end{align*}
as desired.

\textbf{Lower bound for $L_G^{\dag}$}:
\begin{align*}
\|b\|_{L_G^{\dag}}^2 
\ge (1/\lambda_{\max}) \|b\|_{D_G^{-1}}^2
\ge 1/(n^2 w_{\max}) \|b\|_{\infty}^2
\end{align*}

\textbf{Upper bound for $L_G^{\dag}$}:
\begin{align*}
\|b\|_{L_G^{\dag}}^2 \le (1/\lambda_{\min}) \|b\|_{D_G^{-1}}^2
\le (2n^4\alpha^2/w_{\min}) \|b\|_{\infty}^2
\end{align*}
\end{proof}

Proposition \ref{prop:lapl-apx} implies the following equivalent definition of $\epsilon$-approximate multiplication:

\begin{corollary}\label{cor:simpleapproxmult}
Let $G$ be a connected $n$-vertex graph with edge weights $\{w_e\}_{e\in E(G)}$ with $w_{\min} = \min_{e\in E(G)} w_e$, $w_{\max} = \max_{e\in E(G)} w_e$, and $\alpha = w_{\max}/w_{\min}$ and consider any vectors $b,x\in \mathbb{R}^n$. If

$$\|b - L_Gx\|_{\infty} \le \epsilon w_{\max} \|x\|_{\infty}$$
,$b^\top {\bf 1} = 0$, and $x^\top {\bf 1} = 0$, then $b$ is an $2n^3\alpha^2\epsilon$-approximate multiplication of $L_Gx$.
\end{corollary}

\begin{proof}
Since $b^\top {\bf 1} = 0$, $(b - L_Gx)^\top {\bf 1} = 0$ also. By the upper bound for $L_G^{\dagger}$-norms in Proposition \ref{prop:lapl-apx},

$$\|b - L_Gx\|_{L_G^{\dagger}}^2\le \frac{2n^4\alpha^2}{w_{\min}} \|b - L_Gx\|_{\infty}^2\le 2n^4\alpha^4\epsilon^2w_{\min} \|x\|_{\infty}^2$$

Since $x^\top {\bf 1} = 0$, $x$ has both nonnegative and nonpositive coordinates. Therefore, since $G$ is connected, there exists vertices $a,b$ in $G$ for which $\{a,b\}$ is an edge and for which $|x_a - x_b| \ge \|x\|_{\infty}/n$. Therefore,

$$x^\top L_G x \ge w_{ab}(x_a - x_b)^2\ge (w_{\min}/n^2) \|x\|_{\infty}^2$$
Substitution shows that

$$\|b - L_G x\|_{L_G^{\dagger}}^2\le 2n^4\alpha^4\epsilon^2(n^2 x^\top L_G x)$$
This is the desired result by definition of $\epsilon$-approximate multiplication.
\end{proof}

\begin{corollary}
Let $G$ be a connected $n$-vertex graph with edge weights $\{w_e\}_{e\in E(G)}$ with $w_{\min} = \min_{e\in E(G)} w_e$, $w_{\max} = \max_{e\in E(G)} w_e$, and $\alpha = w_{\max}/w_{\min}$ and consider any vectors $b,x\in \mathbb{R}^n$. If $b$ is an $\epsilon/(2n^3\alpha^2)$-approximate multiplication of $L_Gx$ and $b^\top {\bf1} = 0$, then

$$\|b - L_Gx\|_{\infty} \le \epsilon w_{\min}\|x\|_{\infty}$$
\end{corollary}

\begin{proof}
Since $b^\top {\bf 1} = 0$, $(b - L_Gx)^\top {\bf 1} = 0$ as well. By the lower bound for $L_G^{\dagger}$-norms in Proposition \ref{prop:lapl-apx} and the fact that $b$ is an approximate multiplication for $L_Gx$,

\begin{align*}
\|b - L_G x\|_{\infty}^2\le n^2 \alpha w_{\min} \|b - L_G x\|_{L_G^{\dagger}}^2\le \frac{\epsilon^2 w_{\min}}{4n^4\alpha^3} x^\top L_G x .
\end{align*}
Notice that
\begin{align*}
x^\top L_Gx = \sum_{\{a,b\}\in E(G)} w_{ab}(x_a - x_b)^2\le n^2 w_{\max} (4\|x\|_{\infty}^2) .
\end{align*}
Therefore, by substitution,
\begin{align*}
\|b - L_G x\|_{\infty}^2\le (\frac{\epsilon^2w_{\min}}{4n^4\alpha^3})(n^2 \alpha w_{\min} (4\|x\|_{\infty}^2))\le w_{\min}^2 \epsilon^2 \|x\|_{\infty}^2 .
\end{align*}
Taking square roots gives the desired result.
\end{proof}

\subsection{Solving Linear Systems Implies Matrix-Vector Multiplication}
\begin{lemma}\label{lem:mul-given-solve}
Consider an $n$-vertex $w$-weighted graph $G$, let $w_{\min} = \min_{e\in G} w_e $, $w_{\max} = \max_{e\in G} w_e$, $\alpha = w_{\max} / w_{\min}$, and $H$ be a known graph for which \begin{align*}
    (1 - 1/900)L_G
    \preceq 
    L_H
    \preceq 
    (1 + 1/900)L_G.
\end{align*} Suppose that $H$ has at most $Z$ edges and suppose that there is a $\T(n,\delta)$-time algorithm $\textsc{SolveG}(b,\delta)$ that, when given a vector $b\in \mathbb{R}^n$ and $\delta\in (0,1)$, returns a vector $x\in \mathbb{R}^n$ with
\begin{align*}
\|x - L_G^{\dag} b\|_{L_G}\le \delta \cdot \|L_G^{\dag}b\|_{L_G}.
\end{align*}
Then, given a vector $x \in \mathbb{R}^n$ and an $\epsilon \in (0,1)$, there is a 
\begin{align*}
\tilde{O}((Z + \T(n,\epsilon/(n^2\alpha)))\log(Zn\alpha/\epsilon))
\end{align*}
-time algorithm $\textsc{MultiplyG}(x,\epsilon)$ (Algorithm~\ref{alg:multiplyG}) that returns a vector $b\in \mathbb{R}^n$ for which
\begin{align*}
\|b - L_G x\|_{\infty}\le \epsilon \cdot w_{\min} \cdot \|x\|_{\infty}.
\end{align*} 
\end{lemma}

The algorithm $\textsc{MultiplyG}$ (Algorithm~\ref{alg:multiplyG}) uses standard preconditioned iterative refinement. It is applied in the opposite from the usual way. Instead of using iterative refinement to solve a linear system given matrix-vector multiplication, we use iterative refinement to do matrix-vector multiplication given access to a linear system solver.

\begin{algorithm}\caption{\textsc{MultiplyG} and \textsc{MultiplyGAdditive}}\label{alg:multiplyG}
\begin{algorithmic}[1]
\Procedure{\textsc{MultiplyG}}{$x,\epsilon$} \Comment{Lemma~\ref{lem:mul-given-solve}, Theorem~\ref{lem:solve-given-mul}}

    \State \textbf{Given}: $x\in \mathbb{R}^n$, $\epsilon\in (0,1)$, the sparsifier $H$ for $G$, and a system solver for $G$
    
    \State \textbf{Returns}: an approximation $b$ to $L_G x$
    
    \State \Return $\textsc{MultiplyGAdditive}(x,\epsilon \|x\|_{\infty}/( \log^2(\alpha n/\epsilon)) )$

\EndProcedure
\Procedure{\textsc{MultiplyGAdditive}}{$x,\tau$}

    \If{$\|x\|_{\infty}\le \tau/(n^{10}\alpha^5)$}
    
        \State \Return $0$
    
    \EndIf
    
    \State $x_{\text{main}}\gets \textsc{SolveG}(L_H x,\tau \sqrt{w_{\min}}/(n^{10}\alpha^5))$
    
    \State $x_{\text{res}}\gets x - x_{\text{main}}$
    
    \State $b_{\text{res}}\gets \textsc{MultiplyGAdditive}(x_{\text{res}},\tau)$
    
    \State \Return $L_H x + b_{\text{res}}$ 

\EndProcedure
\end{algorithmic}
\end{algorithm}

\begin{proof}
In this proof, assume that ${\bf 1}^{\top}x = 0$. If this is not the case, then shifting $x$ so that it is orthogonal to ${\bf 1}$ only decreases its $\ell_2$-norm, which means that the $\ell_{\infty}$ norm only increases by a factor of $\sqrt{n}$, so the error only increases by $O(\log n)$ additional solves.

{\bf Reduction in residual and iteration bound}: First, we show that
\begin{align*}
\|x_{\text{res}}\|_{L_G}\le (1/14) \|x\|_{L_G}
\end{align*}
Since $(I - L_H L_G^{\dag})L_G(I - L_G^{\dag}L_H)\preceq 3(1/900) L_G$,
\begin{align*}
\|x_{\text{res}}\|_{L_G} &= \|x - x_{\text{main}}\|_{L_G}\\
&\le \|x - L_G^{\dag}L_H x\|_{L_G} + \|L_G^{\dag}L_H x - x_{\text{main}}\|_{L_G}\\
&\le (1/15) \|x\|_{L_G} + (1/10000) \|L_G^{\dag}L_H x\|_{L_G}\\
&\le (1/14) \|x\|_{L_G}\\ 
\end{align*}
Let $x_{\text{final}}$ be the lowest element of the call stack and let $k$ be the number of recursive calls to $\textsc{MultiplyGAdditive}$ (Algorithm~\ref{alg:multiplyG}). By Proposition \ref{prop:lapl-apx},
\begin{align*}
    \|x\|_{L_G}\le \|x\|_{\infty} \sqrt{w_{\max}} \cdot n
    \text{~~~and~~~} \|x_{\text{final}}\|_{L_G}\ge \|x_{\text{final}}\|_{\infty} \sqrt{w_{\min}}/(2n^2\alpha).
\end{align*}

By definition of $x_{\text{final}}$, 
\begin{align*}\|x_{\text{final}}\|_{\infty}
\ge \frac{\tau}{ (14 n^{10}\alpha^5) } = \frac{\epsilon \sqrt{w_{\min}} \|x\|_{\infty}}{ 28 n^{12}\alpha^5 \log(\alpha n/\epsilon)}.
\end{align*}
Therefore,
\begin{align*}
k\le \log_{14}(\|x\|_{L_G}/\|x_{\text{final}}\|_{L_G}) \le \log^2 (\alpha n/\epsilon)
\end{align*}
as desired.

\textbf{Error}: We start by bounding error in the $L_G^{\dag}$ norm. Let $b$ be the output of $\textsc{MultiplyGAdditive}(x,\tau)$ (Algorithm~\ref{alg:multiplyG}). We bound the desired error recursively:

\begin{align*}
\|b - L_G x\|_{L_G^{\dag}} &= \|L_H x + b_{\text{res}} - L_G x\|_{L_G^{\dag}}\\
&= \|b_{\text{res}} - L_G (x - L_G^{\dag} L_H x) \|_{L_G^{\dag}}\\
&= \|b_{\text{res}} - L_G (x - x_{\text{main}}) - L_G(x_{\text{main}} - L_G^{\dag} L_H x)\|_{L_G^{\dag}}\\
&\le \|b_{\text{res}} - L_G (x - x_{\text{main}})\|_{L_G^{\dag}} + \|L_G(x_{\text{main}} - L_G^{\dag} L_H x)\|_{L_G^{\dag}}\\
&= \|b_{\text{res}} - L_G x_{\text{res}}\|_{L_G^{\dag}} + \|x_{\text{main}} - L_G^{\dag} L_H x\|_{L_G}\\
&\le \|b_{\text{res}} - L_G x_{\text{res}}\|_{L_G^{\dag}} + \sqrt{w_{\min}}\tau/(n^{10}\alpha^5)
\end{align*}
Because $0 = \textsc{MultiplyGAdditive}(x_{\text{final}},\tau)$ (Algorithm~\ref{alg:multiplyG}),

\begin{align*}
\|b - L_G x\|_{L_G^{\dag}} &\le \|x_{\text{final}}\|_{L_G} + k\sqrt{w_{\min}}\tau/(n^{10}\alpha^5)\\
&\le n \sqrt{w_{\max}}\|x_{\text{final}}\|_{\infty} + k\sqrt{w_{\min}}\tau/(n^{10}\alpha^5)\\
&\le \sqrt{w_{\min}} \tau/(n^8 \alpha^4)\\
&\le \epsilon \sqrt{w_{\min}} \|x\|_{\infty}/(n^8 \alpha^4) & \text{~by~} \tau \leq \eps \| x \|_{\infty}
\end{align*}
By Proposition \ref{prop:lapl-apx} applied to $L_G^{\dag}$, $\|b - L_Gx\|_{L_G^{\dag}} \ge \|b - L_Gx\|_{\infty}/(n\sqrt{w_{\max}} )$. 
Therefore,
\begin{align*}\|b - L_Gx\|_{\infty} 
\leq & ~ n \sqrt{w_{\max}} \| b - L_G x \|_{L_G^\dag} \\
\leq & ~ n \sqrt{w_{\max}} \epsilon \sqrt{w_{\min}} \frac{1}{n^8 \alpha^4} \| x \|_{\infty} \\
\leq & ~ \epsilon w_{\min} \|x\|_{\infty} \cdot \frac{1}{n^7 \alpha^{3.5}} & \text{~by~} \alpha = w_{\max} / w_{\min} \\
\leq & ~ \epsilon w_{\min}   \|x\|_{\infty}
\end{align*}
as desired.

\textbf{Runtime}: There is one call to $\textsc{SolveG}$ and one multiplication by $L_H$ per call to $\textsc{MultiplyGAdditive}$ (Algorithm~\ref{alg:multiplyG}). Each multiplication by $L_H$ takes $O(Z)$ time. As we have shown, there are only $k\le O(\log (n\alpha/\epsilon))$ calls to $\textsc{MultiplyGAdditive}$ (Algorithm~\ref{alg:multiplyG}). Therefore, we are done.
\end{proof}

\subsection{Matrix-Vector Multiplication Implies Solving Linear Systems}

The converse is well-known to be true \cite{st04,kmp11}:

\begin{lemma}[\cite{st04}]\label{lem:solve-given-mul} 
Consider an $n$-vertex $w$-weighted graph $G$, let $w_{\min} = \min_{e\in G} w_e $, $w_{\max} = \max_{e\in G} w_e$, $\alpha =  w_{\max} / w_{\min}$, and $H$ be a known graph with at most $Z$ edges for which \begin{align*}
    (1 - 1/900)L_G
    \preceq 
    L_H
    \preceq 
    (1 + 1/900)L_G.
\end{align*} 
Suppose that, given an $\epsilon \in (0,1)$, there is a $ \T(n,\epsilon)$-time algorithm $\textsc{MultiplyG}(x,\epsilon)$ (Algorithm~\ref{alg:multiplyG}) that, given a vector $x\in \mathbb{R}^n$, returns a vector $b\in \mathbb{R}^n$ for which
\begin{align*}
\|b - L_G x\|_{\infty}\le \epsilon \cdot w_{\min} \cdot \|x\|_{\infty}.
\end{align*} Then, there is an algorithm $\textsc{SolveG}(b,\delta)$  that, when given a vector $b\in \mathbb{R}^n$ and $\delta\in (0,1)$, returns a vector $x\in \mathbb{R}^n$ with
\begin{align*}
\|x - L_G^{\dag} b\|_{L_G}\le \delta \cdot \|L_G^{\dag}b\|_{L_G}.
\end{align*}
in 
\begin{align*}
\tilde{O}( Z + \T(n, \delta / ( n^4 \alpha^2 ) ) ) \log(Zn\alpha/\delta)
\end{align*}
time.
\end{lemma}

\subsection{Lower bound for high-dimensional linear system solving}

We have shown in this section that if a $\k$ graph can be efficiently sparsified, then there is an efficient Laplacian multiplier for $\k$ graphs if and only if ther is an efficient Laplacian system solver for $\k$ graphs. Here we give one example of how this connection can be used to prove \emph{lower bounds} for Laplacian system solving:

\begin{corollary}[Restatement of Corollary~\ref{cor:introsystemhardhighdim}]
Consider a function $f$ that is $(2,o(\log n))$-multiplicatively Lipschitz for which $f$ cannot be $\epsilon$-approximated by a polynomial of degree at most $o(\log n)$. Then, assuming $\SETH$, there is no $\poly(d\log(\alpha n/\epsilon))n^{1+o(1)}$-time algorithm for $\epsilon$-approximately solving Laplacian systems in the $\k$-graph on $n$ points, where $\k(u,v) = f(\|u-v\|_2^2)$.
\end{corollary}

\begin{proof}
By Theorem \ref{thm:sparsify-lipschitz}, there is a $\poly(d\log(\alpha))n^{1+o(1)}$-time algorithm for sparsifying the $\k$-graph on $n$ points. Since sparsification is efficient, Lemma \ref{lem:mul-given-solve} implies that the existence of a $\poly(d\log(\alpha n/\epsilon))n^{1+o(1)}$-time Laplacian solver yields access to a $\poly(d\log(\alpha n/\epsilon))n^{1+o(1)}$-time Laplacian multiplier. The existence of this multiplier contradicts Theorem \ref{thm:hardnessapprox} assuming $\SETH$, as desired.
\end{proof} 
\section{Matrix-Vector Multiplication}

Recall the adjacency and Laplacian matrices of a $\k$ graph: For any function $\k : \R^d \times \R^d \to \R$, and any set $P = \{x_1, \ldots, x_n\} \subseteq \R^d$ of $n$ points, define the matrix $A_{\k,P} \in \R^{n \times n}$ by 
\begin{align*}
A_{\k,P}[i,j] = 
\begin{cases} 
\k(x_i, x_j), & \text{if } i \neq j ; \\
 0, &\text{if } i=j.
\end{cases}
\end{align*}
Similarly, define the matrix $L_{\k,P} \in \R^{n \times n}$ by 
\begin{align*}
L_{\k,P}[i,j] = 
\begin{cases} 
- \k(x_i, x_j), & \text{if } i \neq j ; \\ 
\sum_{a \in [n] \setminus \{i\}} \k(x_i, x_a), &\text{if } i=j.
\end{cases}
\end{align*}
$A_{\k,P}$ and $L_{\k,P}$ are the adjacency matrix and Laplacian matrix, respectively, of the complete weighted graph on $n$ nodes where the weight between node $i$ and node $j$ is $\k(x_i, x_j)$. 

In this section, we study the algorithmic problem of computing the linear transformations defined by these matrices:

\begin{problem}[$\k$ Adjacency Evaluation]\label{pro:KAdjE}
For a given function $\k : \R^d \times \R^d \to \R$, the $\k$ Adjacency Evaluation ($\KAdjE$) problem asks: Given as input a set $P =\{x_1, \ldots, x_n\} \subseteq \R^d$ with $|P|=n$ and a vector $y \in \R^n$, compute a vector $b \in \R^n$ such that $\|b - A_{\k,P} \cdot y \|_\infty \leq \eps \cdot w_{\max} \cdot \|y\|_\infty$.
\end{problem}

\begin{problem}[$\k$ Laplacian Evaluation]\label{pro:KLapE}
For a given function $\k : \R^d \times \R^d \to \R$, the $\k$ Laplacian Evaluation ($\KLapE$) problem asks: Given as input a set $P =\{x_1, \ldots, x_n\} \subseteq \R^d$ with $|P|=n$ and a vector $y \in \R^n$, compute a vector $b \in \R^n$ such that $\|b - L_{\k,P} \cdot y \|_\infty \leq \eps \cdot w_{\max} \cdot \|y\|_\infty$.
\end{problem}

We make a few important notes about these problems:
\begin{itemize}
    \item In both of the above, problems, $w_{\max} := \max_{u,v \in P} |\k(u,v)|$.
    \item We assume $\eps = 2^{-\polylog n}$ when it is omitted in the above problems. As discussed in Section~\ref{sec:equivalence}, this is small enough error so that we can apply such an algorithm for $\k$ Laplacian Evaluation to solve Laplacian systems, and furthermore, if we can prove hardness for any such $\eps$, it implies hardness for solving Laplacian systems.
    \item Note, by Corollary~\ref{cor:simpleapproxmult}, that when $\eps = 2^{-\polylog n}$, the result of $\KAdjE$ is an $\eps$-approximate multiplication of $A_{\k,P} \cdot y$ (see Definition~\ref{def:approxmult}), and the result of $\KLapE$ is an $\eps$-approximate multiplication of $L_{\k,P} \cdot y$.
    \item We will also sometimes discuss the $f$ $\KAdjE$ and $f$ $\KLapE$ problems for a single-input function $f : \R \to \R$. In this case, we implicitly pick $\k(u,v) = f(\|u-v\|_2^2)$.
\end{itemize}

Suppose the function $\k$ can be evaluated in time $T$ (in this paper we've been assuming $T = \tilde{O}(1)$). Then, both the $\KAdjE$  and $\KLapE$ problems can be solved in $O(T n^2)$ time, by computing all $n^2$ entries of the matrix and then doing a straightforward matrix-vector multiplication. However, since the input size to the problem is only $O(nd)$ real numbers, we can hope for much faster algorithms when $d = o(n)$. In particular, we will aim for $n^{1 + o(1)}$ time algorithms when $d = n^{o(1)}$.

For some functions $\k$, like $\k(x,y) = \| x - y \|_2^2$, we will show that a running time of $n^{1+o(1)}$ is possible for all $d = n^{o(1)}$. For others, like $\k(x,y) = 1/\|x - y\|_2^2$ and $\k(x,y) = \exp(-\|x - y\|_2^2)$, we will show that such an algorithm is only possible when $d \ll \log(n)$. More precisely, for these $\k$:
\begin{enumerate}
    \item When $d = O(\log (n) / \log \log (n))$, we give an algorithm running in time $n^{1 + o(1)}$, and
    \item For $d = \Omega(\log n)$, we prove a conditional lower bound showing that $n^{2 - o(1)}$ time is necessary. 
\end{enumerate}
Finally, for some functions like $\k(x,y) = |\langle x,y \rangle|$, we will show a conditional lower bound showing that $\Omega(n^{2-\delta})$ time is required even when $d = 2^{\Omega(\log^* n)}$ is just barely super-constant.

In fact, assuming $\SETH$, we will characterize the functions $f$ for which the $\KAdjE$ and $\KLapE$ problems can be efficiently solved in high dimensions $d = \Omega(\log n)$ in terms of the \emph{approximate degree} of $f$ (see subsection~\ref{subsec:approxdegree} below). The answer is more complicated in low dimensions $d = o(\log n)$, and for some functions $f$ we make use of the Fast Multipole Method to design efficient algorithms (in fact, we will see that the Fast Multipole Method solves a problem equivalent to our $\KAdjE$ problem).

\subsection{Equivalence between Adjacency and Laplacian Evaluation}

Although our goal in this section is to study the $\k$ Laplacian Evaluation problem, it will make the details easier to instead look at the $\k$ Adjacency Evaluation problem. Here we show that any running time achievable for one of the two problems can also be achieved for the other (up to a $\log n$ factor), and so it will be sufficient in the rest of this section to only give algorithms and lower bounds for the $\k$ Adjacency Evaluation problem.

\begin{proposition} \label{prop:adjtolap}
Suppose the $\KAdjE$ (Problem~\ref{pro:KAdjE}) can be solved in $\T(n,d,\eps)$ time. Then, the $\KLapE$ (Problem~\ref{pro:KLapE}) can be solved in $O(\T(n,d,\eps/2))$ time.
\end{proposition}

\begin{proof}
We use the $\k$ Adjacency Evaluation algorithm with error $\eps/2$ twice, to compute $s := A_{\k, P} \cdot y$ and $g := A_{\k,P} \cdot \vec{1}$, where $\vec{1}$ is the all-1s vector of length $n$. We then output the vector $z \in \R^n$ given by $z_i = g_i \cdot y_i - s_i$, which can be computed in $O(n) = O(\T(n,d,\eps/2))$ time.
\end{proof}

\begin{proposition} \label{prop:laptoadj}
Suppose the $\KLapE$ (Problem~\ref{pro:KLapE}) can be solved in $\T(n,d,\eps)$ time, and that $\T$ satisfies $\T(n_1 + n_2,d,\eps) \geq \T(n_1,d,\eps)+\T(n_2,d,\eps)$ for all $n_1, n_2, d, \eps$. Then, the $\KAdjE$ (Problem~\ref{pro:KAdjE}) can be solved in $ O( \T(n \log n,d,0.5 \eps/\log n))$ time.
\end{proposition}

\begin{proof}
We will show that the $\k$ Adjacency Evaluation problem can be solved in 
\begin{align*}
\sum_{i=0}^{\log n} O(2^i \cdot \T(n/2^i,d,0.5 \eps/\log n))
\end{align*}
time, and then apply the superadditive identity for $T$ to get the final running time.
For a fixed $d$, we proceed by strong induction on $n$, and assume the $\k$ Adjacency Evaluation problem can be solved in this running time for all smaller values of $n$.

Let $a', a'' \in \R^n$ be the vectors given by $a'_i = y_i$ and $a''_i = 0$ when $i \in [n/2]$, and $a'_i = 0$ and $a''_i = y_i$ when $i > n/2$. 

We first compute $z' \in \R^n$ and $z'' \in \R^n$ as follows:
\begin{align*}
 z' := L_{\k,P} \cdot a' \text{~~~and~~~}z'' := L_{\k,P} \cdot a''
\end{align*} 
in $O( \T(n,d,0.5 \eps/\log n))$ time. 

Next, let $y', y'' \in \R^{n/2}$ be the vectors given by $y'_i = y_i$ and $y''_i = y_{n/2 + i}$ for all $i \in [n/2]$, and let $P', P'' \subseteq \R^d$ be given by $P' = \{ x_1, \ldots, x_{n/2}\}$ and $P'' = \{x_{n/2 + 1}, \ldots, x_n \}$. 

We recursively compute $r', r'' \in \R^{n/2}$ given by 
\begin{align*}
r' = L_{\k, P'} \cdot y' \text{~~~and~~~} r'' = L_{\k, P''} \cdot y''.
\end{align*}
Finally, we can output the vector $z \in \R^n$ given by $z_i = z''_i + r'_i$ and $z_{n/2 + i} = z'_{n/2 + i} + r''_i$ for all $i \in [n/2]$. Each of the two recursive calls took time 
\begin{align*}
\sum_{i=1}^{\log_2(n)} O(2^{i-1} \cdot \T(n/2^i,d,0.5 \eps/\log n)),
\end{align*}
and our two initial calls took time $O( \T(n,d,0.5 \eps/\log n))$, leading to the desired running time. Each output entry is ultimately the sum of at most $2 \log n$ terms from calls to the given algorithm, and hence has error $\eps$ (since we perform all recursive calls with error $0.5 \eps / \log n$ and the additive error guarantees in recursive calls can only be more stringent).
\end{proof}

\begin{remark} \label{rem:equivzeroone}
In both Proposition~\ref{prop:adjtolap} and Proposition~\ref{prop:laptoadj}, if the input to the $\KLapE$ (resp. $\KAdjE$) problem is a $\{0,1\}$ vector, then we only apply the given $\KAdjE$ ($\KLapE$) algorithm on $\{0,1\}$ vectors. Hence, the two problems are equivalent even in the special case where the input vector $y$ must be a $\{0,1\}$ vector.
\end{remark}

\subsection{Approximate Degree} \label{subsec:approxdegree}

We will see in this section that the key property of a function $f : \R \to \R$ for determining whether $f$ $\KAdjE$ is easy or hard is how well it can be approximated by a low-degree polynomial.

\begin{definition} \label{def:closetopoly}
For a positive integer $k$ and a positive real number $\eps>0$, we say a function $f : [0,1] \to \R$ is \emph{$\eps$-close to a polynomial of degree $k$} if there is a polynomial $p : [0,1] \to \R$ of degree at most $k$ such that, for every $x \in [0,1]$, we have $|f(x) - p(x)| \leq \eps$.
\end{definition}

The Stone-Weierstrass theorem says that, for any $\eps>0$, and any continuous function $f$ which is bounded on $[0,1]$, there is a positive integer $k$ such that $f$ is $\eps$-close to a polynomial of degree $k$. That said, $k$ can be quite large for some natural and important continuous functions $f$. For some examples:

\begin{example}
For the function $f(x) = 1/(1+x)$, we have $f(x) = \sum_{\ell = 0}^\infty (-1)^\ell x^\ell $ for all $x \in [0,1)$. Truncating this series to degree $O(\log(1/\eps) )$ gives a $\eps$ approximation on the interval $[0,1/2]$. The following proposition shows that this is optimal up to constant factors.
\end{example}

\begin{proposition}
Any polynomial $p(x)$ such that $|p(x) - 1/(1+x)| \leq \eps$ for all $x \in [0,1/2]$ has degree at least $\Omega(\log(1/\eps))$.
\end{proposition}

\begin{proof}
For such a polynomial $p(x)$, define $q(x) := 1 - x \cdot p(x-1)$. Thus, the polynomial $q$ has the two properties that $|q(x)| <  \eps$ for all $x \in [1,3/2]$, and $q(0)=1$. By standard properties of the Chebyshev polynomials (see e.g.~\cite[Proposition~2.4]{sachdeva2014faster}), the polynomial $q$ with those two properties of minimum degree is an appropriately scaled and shifted Chebyshev polynomial, which requires degree $\Omega(\log(1/\eps))$.
\end{proof}

\begin{example}
For the function $f(x) = e^{-x}$, we have $f(x) = \sum_{\ell = 0}^\infty (-1)^\ell x^\ell / \ell!$ for all $x \in \R_+$. Truncating this series to degree $O(\log(1/\eps) / \log\log(1/\eps))$ gives a $\eps$ approximation on any interval $[0,a]$ for constant $a>0$. Such a dependence is believed to be optimal, and is known to be optimal if we must approximate $f(x)$ on the slightly larger interval $[0,\log^2 (1/\eps) / \log^2 \log(1/\eps)]$~\cite[Section~5]{sachdeva2014faster}.
\end{example}

In both of the above settings, for error $\eps = n^{- \Omega(\log^4 n)}$, the function $f$ is only $\eps$-close to a polynomial of degree $\omega(\log n)$. We will see in Theorem~\ref{thm:hardnessapprox} below that this implies that, for each of these functions $f$, the $\eps$-approximate $f$ $\KAdjE$ problem in dimension $d = \Omega(\log n)$ requires time $n^{2 - o(1)}$ assuming $\SETH$.

\subsection{`Kernel Method' Algorithms}

\begin{lemma} \label{lem:low-rank-mmult}
For any integer $q \geq 0$, let $\k(u,v) = (\|u-v\|_2^2)^q$. The $\KAdjE$ problem (Problem~\ref{pro:KAdjE}) can be solved exactly (with $0$ error) in time $\tilde{O}( n \cdot \binom{2d+2q-1}{2q} )$.
\end{lemma}

\begin{proof}
The function 
\begin{align*}
\k(u,v) = \left( \sum_{i=1}^d ( u_i - v_i )^2 \right)^q
\end{align*}
is a homogeneous polynomial of degree $2q$ in the variables $u_1, \ldots, u_d, v_1, \ldots, v_d$. Let 
\begin{align*}
V = \{u_1, \ldots, u_d, v_1, \ldots, v_d\},
\end{align*}
and let $T$ be the set of functions $t : V \to \{0,1,2,\ldots\, 2q\}$  such that $\sum_{v \in V} t(v) = 2q$. 

We can count that 
\begin{align*}
|T| = \binom{2d+2q-1}{2q}.
\end{align*}
Hence, there are coefficients $c_t \in \R$ for each $t \in T$ such that
\begin{align} \label{eqn:polyexpansion}\k(u,v) = \sum_{t \in T} c_t \cdot \prod_{v \in V} v^{t(v)}.\end{align}
Let $V_u = \{u_1, \ldots, u_d\}$ and $V_v = V \setminus V_u$. Define $\phi_u : \R^d \to \R^{|T|}$ by, for $t \in T$, 
\begin{align*}
\phi_u(u_1, \ldots, u_d)_t = c_t \cdot \prod_{u_i \in V_u} {u_i}^{t(u_i)}.
\end{align*}
Similarly define $\phi_v : \R^d \to \R^{|T|}$ by, for $t \in T$, 
\begin{align*}
\phi_v(v_1, \ldots, v_d)_t = \prod_{v_i \in V_v} {v_i}^{t(v_i)}.
\end{align*}
It follows from (\ref{eqn:polyexpansion}) that, for all $u,v \in \R^d$, we have $\k(u,v) = \langle \phi_u(u), \phi_v(v) \rangle$.

Our algorithm thus constructs the matrix $M_u \in \R^{n \times |T|}$ whose rows are the vectors $\phi_u(x_i)$ for $i \in [n]$, and the matrix $M_v \in \R^{|T| \times n}$ whose columns are the vectors $\phi_v(x_i)$ for $i \in [n]$. Then, on input $y \in \R^n$, it computes $y' := M_v \cdot y \in \R^{|T|}$ in $\tilde{O}(n \cdot |T|)$ time, then $z := M_u \cdot y' \in \R^n$, again in $\tilde{O}(n \cdot |T|)$ time, and outputs $z$. Since $M_u \cdot M_v = A_{\k,\{x_1, \ldots, x_n\}}$, it follows that the vector we output is the desired $z = A_{\k,\{x_1, \ldots, x_n\}} \cdot y$.
\end{proof}

\begin{remark}
The running time in Lemma~\ref{lem:low-rank-mmult} can be improved to $\tilde{O}( n \cdot \binom{d+q-1}{q} )$ with more careful work, by noting that each monomial has either `$x$-degree' or `$y$-degree' at most $d$, but we omit this here since the difference is negligible for our parameters of interest.
\end{remark}

\begin{corollary} \label{cor:exactmonomial}
Let $q,d$ be positive integers which may be functions of $n$, such that $\binom{2(d+q)}{2q} < n^{o(1)}$. For example:
\begin{itemize}
    \item when $d = o(\log n / \log \log n)$ and $q \leq \poly (\log n)$, or
    \item when $d = o(\log n)$ and $q \leq O(\log n)$, or
    \item when $d = \Theta(\log n)$ and $q < o(\log n)$.
\end{itemize}
If $f : \R \to \R$ is a polynomial of degree at most $q$, and we define $\k(u,v) := f(\|u-v\|_2^2)$, then the $\KAdjE$ problem in dimension $d$ can be solved exactly in $n^{1 + o(1)}$ time.
\end{corollary}

\begin{proof}
This follows by applying Lemma~\ref{lem:low-rank-mmult} separately to each monomial of $f$, and summing the results.

When $d = o(\log n)$ and $q \leq O(\log n)$, then we can write $d = \frac{1}{a(n)} \log n$ for some $a(n) = \omega(1)$, and $q = b(n) \cdot \log n$ for some $b(n) = O(1)$. It follows that \begin{align*}
\binom{2(d+q)}{2q} 
= & ~ \binom{2(d+q)}{2d} \\
= & ~ \binom{O(q)}{O(d)} & \text{~by~} d = O(q) \\
\leq & ~ O(q/d)^{O(d)} \\
= & ~ 2^{O(d \log(q/d))} \\
= & ~ 2^{O(\frac{\log(ab)}{a}) \cdot \log n} & \text{~by~} d = \frac{\log n}{a}, q = b \log n \\
\leq & ~ 2^{O(\frac{\log(a)}{a}) \cdot \log n} & \text{~by~} b = O(1) \\
< & ~ n^{o(1)}. & \text{~by~} a = \omega(1)
\end{align*} The other cases are similar.
\end{proof}

\begin{corollary} \label{cor:kernelalg}
Suppose $f : \R \to \R$ is $\eps/n$-close to a polynomial of degree $q$ (Definition~\ref{def:closetopoly}), where $q,d$ are positive integers such that $\binom{2(d+q)}{2q} < n^{o(1)}$ (such as the parameter setting examples in Corollary~\ref{cor:exactmonomial}), and define $\k(u,v) := f(\|u-v\|_2^2)$. Then, the $\KAdjE$ problem in dimension $d$ can be solved $\eps$-approximately in $n^{1 + o(1)}$ time. \Josh{fix this to include that one can compute the polynomial efficiently}
\end{corollary}

\begin{proof}
Apply Corollary~\ref{cor:exactmonomial} for the degree $q$ approximation of $f$.
\end{proof}

\subsection{Lower Bound in High Dimensions} \label{sec:lbhighdim}

We now prove that in the high dimensional setting, where $d = \Theta(\log n)$, the algorithm from Corollary~\ref{cor:kernelalg} is essentially tight. In that algorithm, we showed that (recalling Definition~\ref{def:closetopoly}) functions $f$ which are $\eps$-close to a polynomial of degree $o(\log n)$ have efficient algorithms; here we show a lower bound if $f$ is not $\eps$-close to a polynomial of degree $O(\log n)$.

\begin{theorem} \label{thm:hardnessapprox}
Let $f : [0,1] \to \R$ be an analytic function on $[0,1]$ and let $\kappa: \N \to [0,1]$ be a nonincreasing function. Suppose that, for infinitely many positive integers $k$, $f$ is not $\kappa(k)$-close to a polynomial of degree $k$.

Then, assuming $\SETH$, the $\KAdjE$ problem for $\k(x,y) = f(\|x-y\|_2^2)$ in dimension $d$ and error $(\kappa(d+1))^{O(d^4)}$ on $n=1.01^d$ points requires time $n^{2 - o(1)}$.
\end{theorem}

This theorem will be a corollary of another result, which is simpler to use in proving lower bounds:

\begin{theorem} \label{thm:hardnessapprox-easy}
Let $f : [0,1] \to \R$ be an analytic function on $[0,1]$ and let $\kappa: \N \to [0,1]$ be a nonincreasing function. Suppose that, for infinitely many positive integers $k$, there exists an $x_k\in [0,1]$ for which $|f^{(k+1)}(x_k)| > \kappa(k)$.

Then, assuming $\SETH$, the $\KAdjE$ problem for $\k(x,y) = f(\|x-y\|_2^2)$ in dimension $d$ and error $(\kappa(d+1))^{O(d^4)}$ on $n=1.01^d$ points requires time $n^{2 - o(1)}$.
\end{theorem}

To better understand Theorem \ref{thm:hardnessapprox} in the context of our dichotomy, think about the following example:

\begin{example}
Consider a function $f$ that is exactly $\kappa(k) = 2^{-k^3}$-far from the closest polynomial of degree $k$ for every $k\in \N$. By Corollary \ref{cor:kernelalg}, there is an $d^{\log^{1/3} n} n < n^{1+o(1)}$-time algorithm for $1/\text{poly}(n)$-approximate adjacency matrix multiplication on an $n$-vertex $f$-graph when $d = \Theta(log n)$. In fact, there is an algorithm even for $2^{-o(\log^3 n)}$-error that takes $n^{1+o(1)}$. However, by Theorem \ref{thm:hardnessapprox}, there is no $n^{2-o(1)}$-time algorithm for $2^{-d^3\cdot d^4} = 2^{-\Theta(\log^7 n)}$-approximate multiplication. 
\end{example}

We now give a more concrete version of the proof outline described in the introduction. To prove Theorem \ref{thm:hardnessapprox} given Theorem \ref{thm:hardnessapprox-easy}, it suffices to show that for any function that is far from a degree $k$ polynomial, there exists a point with high $(k+1)$-th derivative (Lemma \ref{lem:step1}). To prove Theorem \ref{thm:hardnessapprox-easy}, we start by showing that there exists an interval (not just a single point) with high $(k+1)$-th derivative (Lemma \ref{lem:step2}). This is done by integrating over the $(k+2)$-nd derivative, exploiting the fact that it is bounded for analytic functions (Proposition \ref{prop:derivs}). Then, we further improve this derivative lower bound by showing that there is an interval on which \emph{all} $i$-th derivatives for $i\le k+1$ are bounded from below (Lemma \ref{lem:step3}). This is done by induction, deriving a bound for $i$-th derivatives by integrating over the $(i+1)$-th derivative. The lower bound on the $(i+1)$-th derivative ensures that it can only be close to 0 at a small interval around one point, so picking an interval far from that point suffices for the inductive step.

Up to this point, we have argued that there must be an interval $I\subset [0,1]$ on which all of $f$'s $\le (k+1)$-derivatives are large in absolute value (Lemma \ref{lem:step3}). We exploit this property to solve an exact bichromatic nearest neighbors problem in Hamming distance in $d = \Theta(\log n)$ dimensions (Lemma \ref{lem:step6}). Since even approximate nearest neighbors cannot be solved in $n^{2 - \delta}$-time for $\delta > 0$ assuming $\SETH$ (Theorem \ref{thm:r18}), this suffices. To solve Hamming nearest neighbors a a pair of sets $S$ and $T$ with $|S| = |T| = n$, we set up $d+1$ different $f$-graph adjacency matrix multiplication problems. In problem $i$, we scale the points in $S$ and $T$ by a factor of $\zeta i$ for some $\zeta  0$ and translate them by $c \in [0,1]$ so that they are in the interval $I$. Then, with one adjacency multiplication, one can evaluate an expression $Z_i$, where $Z_i = \sum_{x\in S, y\in T} f(c + i^2\zeta^2 \|x - y\|_2^2)$. For each distance $i\in [d]$, let $u_j = |\{x\in S, y\in T: \|x - y\|_2^2 = j\}|$. To solve bichromatic nearest neighbors, it suffices to compute all of the $u_j$s. This can be done by setting up a linear system in the $u_j$s, where there is one equation for each $Z_i$. The matrix for this linear system has high determinant because $f$ has high $\le (k+1)$-th derivatives on $I$ (Lemma \ref{lem:step4}). Cramer's Rule can be used to bound the error in our estimate of the $u_j$s that comes from the error in the multiplication oracle (Lemma \ref{lem:step5}). Therefore, $O(d)$ calls to a multiplication oracle suffices for computing the number of pairs of vertices in $S\times T$ that are at each distance value. Returning the minimum distance $i$ for which $u_i > 0$ solves bichromatic nearest neighbors, as desired.

In Lemma \ref{lem:step4}, we will make use of the Cauchy-Binet formula:

\begin{lemma}[Cauchy-Binet formula for infinite matrices]\label{lem:cauchy-binet}
Let $k$ be a positive integer, and for functions $A : [k] \times \N \to \R$ and $B : \N \times [k] \to \R$, define the matrix $C \in \R^{k \times k}$ by, for $i,j \in [k]$, $$C_{ij} := \sum_{\ell=0}^\infty A_{i\ell} \cdot B_{\ell j},$$
and suppose that the sum defining $C_{ij}$ converges absolutely for all $i,j$. Then,
$$\det(C) = \sum_{1 \leq \ell_1 < \ell_2 < \cdots < \ell_k} \det(A[\ell_1, \ell_2, \cdots, \ell_k]) \cdot \det(B[\ell_1, \ell_2, \cdots, \ell_k]),$$
where $A[\ell_1, \ell_2, \cdots, \ell_k] \in \R^{k \times k}$ denotes the matrix whose $i,j$ entry is given by $A_{i\ell_j}$ and $B[\ell_1, \ell_2, \cdots, \ell_k]$ denotes the matrix whose $i,j$ entry is given by $B_{\ell_i j}$.
\end{lemma}

In this section, we exploit the following property of analytic functions:

\begin{proposition}\label{prop:derivs}
Consider a function $f:[0,1]\rightarrow \mathbb{R}$ that is analytic. Then, there is a constant $B > 0$ depending on $f$ such that for all $k \ge 0$ and all $x\in [0,1]$, $|f^{(k)}(x)| < B 4^k k!$
\end{proposition}

\begin{proof}
We first show that, for any $x\in [0,1]$, $|f^{(k)}(x)| < B_x 2^k k!$ for some constant $B_x$ depending on $x$. Write $f$'s Taylor expansion around $x$:
\begin{align*}
f(y) = \sum_{i=0}^{\infty} f^{(i)}(x) (y - x)^i / i! .
\end{align*}
Let $y_0 = \arg \max_{a\in \{0,1\}} |a - x|$. Note that $|y_0 - x| \ge 1/2$. Since $f(y_0)$ is a convergent series, there exists a constant $N_x$ dependent on $x$ such that for all $i > N_x$, the absolute value of the $i$-th term of the series for $f(y_0)$ is at most 1/2. Therefore, for all $i > N_x$, $|f^{(i)}(x)| < 2(2^i) i!$. For all $i \le N_x$, $f^{(i)}(x)$ is a constant depending on $x$, so we are done with this part.

Next, we show that $|f^{(k)}(x)| < B 4^k k!$ for all $x\in [0,1]$ and some constant $B$ depending only on $f$. Let $x_0\in \{i/8\}_{i=0}^8$ be the point that minimizes $|x - x_0|$. Note that $|x - x_0| < 1/16$. Taylor expand $f$ around $x_0$:
\begin{align*}
f(x) = \sum_{i = 0}^{\infty} f^{(i)}(x_0) (x - x_0)^i / i! .
\end{align*}
Take derivatives for some $k \ge 0$ and use the triangle inequality:
\begin{align*}
|f^{(k)}(x)| \le \sum_{i = 0}^{\infty} |f^{(i+k)}(x_0)| |x - x_0|^i / i!
\end{align*}
By the first part, $|f^{(i+k)}(x_0)| \le B_{x_0} 2^{i+k} (i+k)!$, so
\begin{align*}
|f^{(k)}(x)| \le \sum_{i=0}^{\infty} B_{x_0} 2^{i+k} ((i+k)!/i!) (1/16)^{(i+k)}
\end{align*}
Note that $(i + k)!/i! \le (2i)^k$ for $i > k$ (we are done for $i \le k$). There is some constant $C$ for which $\sum_{i=0}^{\infty} i^k 4^{-i} = C$, so letting $B = B_{x_0} C$ suffices, as desired.
\end{proof}

We now move on to proving the main results of this section, which consists of several steps.

\begin{lemma}[Step 1: high $(k+1)$-derivative]\label{lem:step1}
Let $f : [0,1] \to \R$ be an analytic function on $[0,1]$ and let $\kappa: \N \to [0,1]$ be a nonincreasing function. Suppose that, for some positive integer $k$, $f$ is not $\kappa(k)$-close to a polynomial of degree $k$.
Then, there exists some $x\in [0,1]$ for which $|f^{(k+1)}(x)| > \kappa(k)$.
\end{lemma}

\begin{proof}
Define the function $g : [0,1] \to \R$ by $g(x) = f(x) - \sum_{\ell=0}^k \frac{f^{(\ell)}(0) \cdot x^\ell}{\ell!}$. We claim that, for each $i \in \{0,1,\ldots,k+1\}$, there is an $x_i \in [0,1]$ such that $|g^{(i)}(x_i)| > \kappa(k)$. Since $g^{(k+1)}(x) = f^{(k+1)}(x)$, plugging in $i=k+1$ into this will imply our desired result. We prove this by induction on $i$.

For the base case $i=0$: note that $g$ is the difference between $f$ and a polynomial of degree $k$, and so by our assumption that $f$ is not $\kappa(k)$-close to a polynomial of degree $k$, there must be an $x_0 \in [0,1]$ such that $|g(x_0)| > \kappa(k)$.

For the inductive step, consider an integer $i\in [k+1]$. Notice that $g^{(i-1)}(0) = 0$ by definition of $g$ since $i \leq k+1$. By the inductive hypothesis, there is an $x_{i-1} \in [0,1]$ with $|g^{(i-1)}(x_{i-1})| > \kappa(k)$. Hence, by the mean value theorem, there must be an $x_i \in [0,x_{i-1}]$ such that 
\begin{align*}
    |g^{(i)}(x_i)| \geq |g^{(i-1)}(x_{i-1}) - g^{(i-1)}(0)|/x_{i-1} \geq |g^{(i-1)}(x_{i-1})| > \kappa(k),
\end{align*}
as desired.
\end{proof}

\begin{proof}[Proof of Theorem \ref{thm:hardnessapprox} given Theorem \ref{thm:hardnessapprox-easy}]
For each of the infinitely many $k$s for which $f$ is $\kappa(k)$-far from a degree $k$ polynomial, Lemma \ref{lem:step1} implies the existance of an $x_k$ for which $|f^{(k+1)}(x_k)| > \kappa(x_k)$. Thus, $f$ satisfies the input condition of Theorem \ref{thm:hardnessapprox-easy}, so applying Theorem \ref{thm:hardnessapprox-easy} proves Theorem \ref{thm:hardnessapprox} as desired.
\end{proof}

Now, we focus on Theorem \ref{thm:hardnessapprox-easy}:

\begin{lemma}[Step 2: high $(k+1)$-derivative on interval] \label{lem:step2}
Let $f : [0,1] \to \R$ be an analytic function on $[0,1]$ and let $\kappa: \N \to [0,1]$ be a nonincreasing function. There is a constant $B>0$ depending only on $f$ such that the following holds.

Suppose that, for some sufficiently large positive integer $k$, there is an $x\in [0,1]$ for which $|f^{(k+1)}(x)| > \kappa(k)$.
Then, there exists an interval $[a,b]\subseteq [0,1]$ with the property that both $b - a \geq \kappa(k) / (32 Bk4^k \cdot k!)$ and, for all $y\in [a,b]$, $|f^{(k+1)}(y)| > \kappa(k)/2$.
\end{lemma}

\begin{proof}
Since $f$ is analytic on $[0,1]$, Proposition \ref{prop:derivs} applies and it follows that there is a constant $B > 0$ dependent on $f$ such that for every $y \in [0,1]$ and every nonnegative integer $m$, we have $|f^{(m)}(y)| \leq Bm4^m \cdot m!$. In particular, for all $y \in [0,1]$, we have $|f^{(k+2)}(y)| \leq 16 Bk4^k \cdot k!$. Let $\delta = \kappa(k)/(32 Bk4^k \cdot k!)$, then let $a = \max\{0, x-\delta\}$, and $b = \min\{1,x+\delta\}$. We have $b-a \geq \delta = \kappa(k)/(32 Bk4^k \cdot k!)$, since when $k$ is large enough, we get that $\delta<1/2$, so we cannot have both $a=0$ and $b=1$. Meanwhile, for any $y \in [a,b]$, we have as desired that
\begin{align*}
|f^{(k+1)}(y)| 
\geq & ~ |f^{(k+1)}(x)| - |x-y| \cdot \sup_{y' \in [a,b]} |f^{(k+2)}(y')| \\
> & ~ \kappa(k) - \delta \cdot 16 Bk4^k \cdot k! \\
= & ~ \kappa(k) / 2.
\end{align*}
\end{proof}

\begin{lemma}[Step 3: high lower derivatives on subintervals] \label{lem:step3}
Let $f : [0,1] \to \R$ be an analytic function on $[0,1]$ and let $\kappa: \N \to [0,1]$ be a nonincreasing function. There is a constant $B>1$ depending only on $f$ such that the following holds.

Suppose that, for some sufficiently large positive integer $k$, there exists an $x\in [0,1]$ for which $|f^{(k+1)}(x)| > \kappa(k)$.
Then, there exists an interval $[c,d]\subseteq [0,1]$ with the property that both $d - c > \kappa(k) / (128 Bk4^{2k+1} \cdot k!)$ and, for all $y\in [c,d]$ and all $i\le k+1$, $|f^{(i)}(y)| > \left[ \kappa(k)^2 / (64 Bk4^{2k+1} \cdot k!) \right]^{k+2-i}$.
\end{lemma}

\begin{proof}
We will prove that, for all integers  $0 \leq i \leq k+1$, there is an interval $[c_i,d_i] \subseteq [0,1]$ such that $d_i - c_i > \kappa(k) / (32 Bk4^{2k+1 - i} \cdot k!)$, and for all integers $i \leq i' \leq k+1$ and all $y\in [c_i,d_i]$ we have $|f^{(i')}(y)| > \left[ \kappa(k)^2 / (64 Bk4^{2k+1} \cdot k!) \right]^{k+2-i'}$. Plugging in $i=0$ gives the desired statement. We will prove this by induction on $i$, from $i=k+1$ to $i=0$. The base case $i=k+1$ is given (with slightly better parameters) by Lemma~\ref{lem:step2}.

For the inductive step, suppose the statement is true for $i+1$. We will pick $[c_i,d_i]$ to be a subinterval of $[c_{i+1},d_{i+1}]$, so the inductive hypothesis says that for every $y \in [c_i,d_i]$ and every integer $i < i' \leq k+1$ we have $|f^{(i')}(y)| > \left[ \kappa(k)^2 / (64 Bk4^{2k+1} \cdot k!) \right]^{k+2-i'}$. It thus remains to show that we can further pick $c_i$ and $d_i$ such that $d_i-c_i \geq \frac14 (d_{i+1} - c_{i+1})$ and $|f^{(i)}(y)| > \left[ \kappa(k)^2 / (64 Bk4^{2k+1} \cdot k!) \right]^{k+2-i}$ for all $y \in [c_i,d_i]$. 

Recall that $|f^{(i+1)}(y)| > \left[ \kappa(k)^2 / (64 Bk4^{2k+1} \cdot k!) \right]^{k+1-i}$ for all $y \in [c_{i+1},d_{i+1}]$. Since $f^{(i+1)}$ is continuous, we must have
\begin{itemize}
    \item either $f^{(i+1)}(y) > \left[ \kappa(k)^2 / (64 Bk4^{2k+1} \cdot k!) \right]^{k+1-i}$ for all such $y$,
    \item or ~~~$-f^{(i+1)}(y) > \left[ \kappa(k)^2 / (64 Bk4^{2k+1} \cdot k!) \right]^{k+1-i}$ for all such $y$.
\end{itemize}  
Let us assume we are in the first case; the second case is nearly identical. Let $\delta = (d_{i+1} - c_{i+1})/4$, and consider the four subintervals
\begin{align*}
    [c_{i+1}, c_{i+1} + \delta], ~~~ [c_{i+1} + \delta, c_{i+1} + 2\delta], ~~~ [c_{i+1} + 2\delta, c_{i+1} + 3\delta], \text{~~~and~~~} [c_{i+1} + 3\delta, c_{i+1} + 4\delta].
\end{align*} 
Since $f^{(i+1)}(y) > \left[ \kappa(k)^2 / (64 Bk4^{2k+1} \cdot k!) \right]^{k+1-i}$ for all $y$ in each of those intervals, we know that for each of the intervals, letting $c'$ denote its left endpoint and $d'$ denote its right endpoint, we have 
\begin{align*}
f^{(i)}(d') - f^{(i)}(c') 
\geq & ~ \delta \cdot \left[ \kappa(k)^2 / (64 Bk4^{2k+1} \cdot k!) \right]^{k+1-i} \\
\geq & ~ \left[ \kappa(k) / (32 Bk4^{2k+1 - i} \cdot k!) \right] \cdot \left[ \kappa(k)^2 / (64 Bk4^{2k+1} \cdot k!) \right]^{k+1-i} / 4 \\
\geq & ~ \left[ \kappa(k)^2 / (64 Bk4^{2k+1} \cdot k!) \right]^{k+2-i}.
\end{align*}

In particular, $f^{(i)}$ is increasing on the interval $[c_{i+1}, d_{i+1}]$, and if we look at the five points $y = c_{i+1} + a \cdot \delta$ for $a \in \{0,1,2,3,4\}$ which form the endpoints of our four subintervals, $f^{(i)}$ increases by more than $\left[ \kappa(k)^2 / (64 Bk4^{2k+1} \cdot k!) \right]^{k+2-i}$ from each to the next. It follows by a simple case analysis (on where in our interval $f^{(i)}$ has a root) that there must be one of our four subintervals with $|f^{(i)}(y)| > \left[ \kappa(k)^2 / (64 Bk4^{2k+1} \cdot k!) \right]^{k+2-i}$ for all $y$ in the subinterval. We can pick that subinterval as desired.
\end{proof}

To simplify notation in the rest of the proof, we will let $\rho(k) = \left[\kappa(k)^2/(64Bk4^{2k+1}\cdot k!)\right]^{k+2}$. We now use these properties of $f$ to reason about a certain matrix connected to $f$ that can be used to count the number of pairs of points at each distance. 

\begin{definition}[Counting matrix]\label{def:counting_matrix}
For a function $f : \R \rightarrow \R$, an integer $k \geq 1$, and a function $\rho: \N\rightarrow \mathbb{R}_{>0}$, let $[c,d]$ be the interval from Lemma~\ref{lem:step3}. Define the \emph{counting matrix} be the $k\times k$ matrix $M$ for which 
\begin{align*}
M_{ij} = f(c + (\rho(k)/(B(200k)^k))^{10k} \cdot i \cdot j / k^2).
\end{align*}
\end{definition}

\begin{lemma}[Step 4: determinant lower bound for functions $f$ that are far from polynomials using Cauchy-Binet] \label{lem:step4}
Let $f : [0,1] \to \R$ be an analytic function on $[0,1]$, let $\kappa: \N \to [0,1]$ be a nonincreasing function, and let $\rho(\ell) = \left[\kappa(\ell)^2/(64B\ell4^{2\ell+1}\cdot \ell!)\right]^{\ell+2}$ (as discussed before). Let $k$ be an integer for which there exists an $x\in [0,1]$ with $|f^{(k+1)}(x)| > \kappa(k)$.

Let $M$ be the counting matrix (Definition~\ref{def:counting_matrix}) for $f$, $k$, and $\rho$. Then \begin{align*}
|\det(M)| > (\rho(k))^k(\rho(k)/(Bk^2(200k)^k))^{10k^3}.
\end{align*}
\end{lemma}

\begin{proof}
Since $f$ is analytic on $[c,d]$, we can Taylor expand it around $c$:

$$f(x) = \sum_{\ell=0}^{\infty} \frac{f^{(\ell)}(c)}{\ell!} (x - c)^\ell$$
Let $\delta = (\rho(k)/(B(200k)^k))^{10k}/k^2$. Note that for all values of $i,j\in [k]$, the input to $f$ in $M_{ij}$ is in the interval $[c,d]$ by the lower bound on $d - c$ in Lemma \ref{lem:step3}. In particular, for all $i,j\in [k]$,

$$M_{ij} = \sum_{\ell=0}^{\infty} \frac{f^{(\ell)}(c)}{\ell!} \delta^{\ell} i^{\ell} j^{\ell}$$
Define two infinite matrices $A: [k]\times \mathbb Z_{\ge 0} \rightarrow \R$ and $C: \mathbb Z_{\ge 0}\times [k] \rightarrow \R$ as follows:

$$A_{i\ell} = \frac{f^{(\ell)}(c)}{\ell!} \delta^{\ell} i^{\ell}$$

$$C_{\ell j} = j^{\ell}$$
for all $i,j\in [k]$ and $\ell\in \mathbb Z_{\ge 0}$. Then

$$M_{ij} = \sum_{\ell=0}^{\infty} A_{i\ell}C_{\ell j}$$
for all $i,j\in [k]$ and converges, so we may apply Lemma \ref{lem:cauchy-binet}. By Lemma \ref{lem:cauchy-binet},

$$\det(M) = \sum_{0\le \ell_1 < \ell_2 < \hdots < \ell_k} \det(A[\ell_1, \ell_2, \hdots, \ell_k]) \det(C[\ell_1, \ell_2, \hdots, \ell_k])$$
To lower bound $|\det(M)|$, we 
\begin{enumerate}[(1)]
    \item lower bound the contribution of the term for the tuple $(\ell_1, \ell_2, \hdots, \ell_k) = (0, 1, \hdots, k-1)$,
    \item upper bound the contribution of every other term,
    \item  show that the lower bound dominates the sum.
\end{enumerate} 
We start with part (1). Let $D$ and $P$ be $k\times k$ matrices, with $D$ diagonal, $D_{\ell\ell} = \frac{f^{(\ell)}(c)\delta^{\ell}}{\ell!}$, and $P_{i\ell} = i^{\ell}$ for all $\ell\in \{0,1,\hdots,k\}$ and $i\in [k]$. Then $A[0,1,\hdots,k-1] = P D$, which means that
\begin{align*}
\det(A[0,1,\hdots,k-1]) = \det(P) \cdot \det(D)
\end{align*}
$P$ and $C[0,1,\hdots,k-1]$ are Vandermonde matrices, so their determinants has a closed form and, in particular, have the property that $|\det(P)| \ge 1$ and $|\det(C[0,1,\hdots,k-1])| \ge 1$. By Lemma \ref{lem:step3}, $|D_{\ell\ell}| > \delta^\ell \rho(\ell) \ge \delta^{\ell}\rho(k)$ for all $\ell\in \{0,1\hdots,k-1\}$. Therefore,
\begin{align*}
|\det(A[0,1,\hdots,k-1])| \cdot |\det(C[0,1,\hdots,k-1])| 
> & ~ \delta^{1 + 2 + \hdots + (k-1)} \rho(k)^k \\
= & ~ \delta^{\binom{k}{2}} \rho(k)^k.
\end{align*}
This completes part (1). Next, we do part (2). Consider a $k$-tuple $0\le \ell_1 < \ell_2 < \hdots < \ell_k$ and a permutation $\sigma: [k]\rightarrow [k]$. By Proposition \ref{prop:derivs}, there is some constant $B > 0$ depending on $f$ for which $|f^{(\ell)}(c)|\le B 10^{\ell}(\ell!)\le B(10\ell)^{\ell}$ for all $\ell$. Therefore,

$$\left|\prod_{i=1}^k A_{i\ell_{\sigma(i)}} \right|\le B^k (10\delta k)^{\sum_{i=1}^k \ell_i}$$
We also get that

$$\left|\prod_{j=1}^k C_{\ell_{\sigma(j)}j}\right|\le k^{\sum_{j=1}^k \ell_j}$$
Summing over all $k!$ permutations $\sigma$ yields an upper bound on the determinants of the blocks of $A$ and $C$, excluding the top block:

\begin{align*}
 & ~ \sum_{0\le \ell_1 < \ell_2 < \hdots < \ell_k, \ell_k\ne k-1} |\det(A[\ell_1,\ell_2,\hdots,\ell_k])| |\det(C[\ell_1,\ell_2,\hdots,\ell_k])| \\
\le & ~ \sum_{0\le \ell_1 < \ell_2 < \hdots < \ell_k, \ell_k\ne k-1} (k!)^2 B^k (10\delta k^2)^{\sum_{i=1}^k \ell_i}\\
\le & ~ \sum_{\tau = 1 + 2 + \hdots + (k-3) + (k-2) + k}^{\infty} \tau^k (k!)^2 B^k (10\delta k^2)^{\tau}\\
\le & ~ 2\tau_0^k (k!)^2 B^k (10\delta k^2)^{\tau_0},
\end{align*}
where $\tau_0 = \binom{k}{2} + 1$. This completes part (2). Now, we do part (3). By Lemma \ref{lem:cauchy-binet},

\begin{align*}
|\det(M)| &\ge |\det(A[0,1,\hdots,k-1])| |\det(C[0,1,\hdots,k-1])|\\
&- \sum_{0\le \ell_1<\ell_2<\hdots<\ell_k, \ell_k\ne k-1} |\det(A[\ell_1,\hdots,\ell_k])| |\det(C[\ell_1,\hdots,\ell_k])|
\end{align*}
Plugging in the part (1) lower bound and the part (2) upper bound yields

\begin{align*}
|\det(M)| &\ge \delta^{\tau_0-1}\rho(k)^k - 2\tau_0^k (k!)^2 B^k (10\delta k^2)^{\tau_0}\\
&= \delta^{\tau_0-1}\left(\rho(k)^k - 2\delta\tau_0^k (k!)^2 B^k (10k^2)^{\tau_0}\right)\\
&> \delta^{\tau_0-1} \rho(k)^k/2\\
&> \rho(k)^k(\rho(k)/(Bk^2(200k)^k))^{10k^3}
\end{align*}
as desired.
\end{proof}

\begin{lemma}[Step 5: Cramer's rule-based bound on error in linear system] \label{lem:step5}
Let $M$ be an invertible $k$ by $k$ matrix with $|M_{ij}|\le B$ for all $i,j\in [k]$. Let $b$ be a $k$-dimensional vector for which $|b_i|\le \epsilon$ for all $i\in [k]$. Then, $\|M^{-1}b\|_{\infty}\le \epsilon k! B^k / |\det(M)|$.
\label{lemma:step5}
\end{lemma} 

\begin{proof}
Cramer's rule says that, for each $i \in [k]$, the entry $i$ of the vector $M^{-1}b$ is given by $$(M^{-1}b)_i = \frac{\det(M_i)}{\det(M)},$$ where $M_i$ is the matrix which one gets by replacing column $i$ of $M$ by $b$. Let us upper-bound $|\det(M_i)|$. We are given that each entry of $M_i$ in column $i$ has magnitude at most $\eps$, and each entry in every other column has magnitude at most $B$. Hence, for any permutation $\sigma \in S_k$ on $[k]$, we have $$\left| \prod_{j=1}^k (M_i)_{j,\sigma(j)} \right| \leq \eps \cdot B^{k-1},$$
and so $$|\det(M_i)| \leq  \sum_{\sigma \in S_k} \left| \prod_{j=1}^k (M_i)_{j,\sigma(j)} \right| \leq \eps \cdot B^{k-1} \cdot k!.$$ It follows from Cramer's rule that $|(M^{-1}b)_i| \leq \eps \cdot B^{k-1} \cdot k! / |\det(M)|$, as desired.
\end{proof}

\begin{lemma}[Step 6: reduction] \label{lem:step6}
Let $f : [0,1] \to \R$ be an analytic function on $[0,1]$, let $\kappa: \N \to [0,1]$ be a nonincreasing function, and let $\rho(\ell) = \left[\kappa(\ell)^2/(64B\ell 4^{2\ell+1}\cdot \ell!)\right]^{\ell+2}$ for any $\ell\in \N$ (as discussed before). Suppose that, for infinitely many positive integers $k$, there exists an $x_k$ for which $|f^{(k+1)}(x)| > \kappa(k)$.

Suppose that there is an algorithm for $\epsilon$-approximate matrix-vector multiplication by an $n\times n$ $f$-matrix for points in $[0,1]^d$ in $T(n,d,\epsilon)$ time. Then, there is a 
\begin{align*}
n \cdot \poly(d + \log(1/\kappa(d+1)))  + O(d \cdot T(2n,d+1,(\kappa(d+1))^{O(d^4)}))) \end{align*}
time algorithm for exact bichromatic Hamming nearest neighbors on $n$-point sets in dimension $d$.
\end{lemma}

\begin{proof}
Let $x_1, \ldots, x_n \in \{0,1\}^d$ be the input to the Hamming nearest neighbors problem, so our goal is to compute $\min_{1 \leq i < j \leq n} \|x_i - x_j\|$. Let $t \in \Z_{\geq 0}^{d+1}$ be the 0-indexed vector of nonnegative integers, where $t_\ell := | \{ 1 \leq i < j \leq n \mid \|x_i - x_j\|_2^2 = \ell \} |$ counts the number of pairs of input points with distance $\ell$. Our goal will be to recover the vector $t$, from which we can recover the answer to the Hamming nearest neighbors problem by returning the smallest index where $t$ is nonzero.

Let $k=d+1$ and let $\eps = (\kappa(k))^{\alpha \cdot k^4}$ for a constant $\alpha>0$ to be picked later. In order to recover $t$, we will make $d+1$ calls to our given algorithm. For $\ell \in \{0,1,\ldots,d\}$, the goal of call $\ell$ is to compute a value $u_\ell$ which is an approximation, with additive error at most $\eps$, of entry $\ell$ of the vector $M t$, where $M$ is the counting matrix defined above.

Let us explain why this is sufficient to recover $t$. Suppose we have computed this vector $u$. We claim that if we compute $M^{-1}u$, and round each entry to the nearest integer, the result is the vector $t$. Indeed, by Lemma~\ref{lemma:step5}, each entry of $M^{-1}u$ differs from the corresponding entry of $t$ by at most an additive $\eps \cdot B^{k-1} \cdot k! / |\det(M)|$, where the constant $B$ is from Proposition \ref{prop:derivs} (since $|f(z)|\le B$ for all $z\in [0,1]$). Substituting our lower bound on $|\det(M)|$ from Lemma~\ref{lem:step4}, we see this additive error is at most $1/3$ as long as we've picked a sufficiently large constant $\alpha>0$. Thus, rounding each entry to the nearest integer recovers $t$, as desired.

It remains to show how to compute $u_\ell$, an approximation with additive error at most $\eps$ of entry $\ell$ of the vector $M t$. In other words, we need to approximate the sum $$\sum_{p=0}^d M_{\ell,p} \cdot t_p = \sum_{1 \leq i < j \leq k} f(c + (\rho(k)/(B(200k)^k))^{10k} \cdot \ell \cdot \|x_i - x_j\|_2^2 / k^2).$$ To do this, we will pick points $y_1, \ldots, y_n, z_1, \ldots, z_n \in [0,1]^{d+1}$ such that 
\begin{align*}
\|y_i - z_j\|_2^2 = c + (\rho(k)/(B(200k)^k))^{10k} \cdot \ell \cdot \|x_i - x_j\|_2^2 / k^2
\end{align*}
for all $i,j \in [n]$, and apply our given algorithm with error $\eps$ to these points. For $i \in [n]$, let $x_i' = (\rho(k)/(B(200k)^k))^{5k} \cdot \sqrt{\ell} \cdot x_i / k$ be a rescaling of $x_i$. We pick $y_i$ to equal $x'_i$ in the first $d$ entries and $0$ in the last entry, and $z_i$  to equal $x'_i$ in the first $d$ entries and $\sqrt{c}$ in the last entry. These points have the desired distances, completing the proof.
\end{proof}

\begin{proof}[Step 7: proof of Theorem~\ref{thm:hardnessapprox-easy}]
If the $\KAdjE$ problem for $\k(x,y) = f(\|x-y\|_2^2)$ in dimension $d$ and error $(\kappa(d+1))^{O(d^4)}$ on $n=1.01^d$ points could be solved in time time $n^{2-\delta}$ for any constant $\delta>0$, then one could immediately substitute this into Lemma~\ref{lem:step6} to refute $\SETH$ in light of Theorem~\ref{thm:r18}.
\end{proof}

\subsection{Lower Bounds in Low Dimensions}\label{sec:inner_product_kernel_matrix_vector_hardness}

The landscape of algorithms available in low dimensions $d = o(\log n)$ is a fair bit more complex. The Fast Multipole Method allows us to solve $\KAdjE$ for a number of functions $f$, including $\k (x,y) = \exp ( -  \| x - y \|_2^2 ) $ and $\k (x,y) = 1 / \| x - y \|_2^2$, for which we have a lower bound in high dimensions. Classifying when these multipole methods apply to a function $f$ seems quite difficult, as researchers have introduced more and more tools to expand the class of applicable functions. See Section~\ref{sec:fastmm}, below, in which we give a much more detailed overview of these methods.

That said, in this subsection, we prove lower bounds for a number of functions $\k$ of interest. We show that for a number of functions $\k$ with applications to geometry and statistics, the $\KAdjE$ problem seems to become hard even in dimension $d=3$ (see the end of this subsection for a list of such $\k$).

We begin with the function $\k(x,y) = |\langle x , y \rangle|$. Here, the $\k$ Adjacency Evaluation problem becomes hard even for very small $d$. We give an $n^{1 + o(1)}$ time algorithm only for $d \leq 2$. For $d=3$ we show that such an algorithm would lead to a breakthrough in algorithms for the $\Z$-{\sf MaxIP} problem, and for the only slightly super-constant $d = 2^{\Omega(\log^* n)}$, we show that a $n^{2-\eps}$ time algorithm would refute {\sf SETH}.

\begin{lemma}
For the function $\k(x,y) = |\langle x , y \rangle|$, the $\KAdjE$ problem (Problem~\ref{pro:KAdjE}) can be solved exactly in time $n^{1 + o(1)}$ when $d=2$.
\end{lemma}

\begin{proof}
Given as input $x_1, \ldots, x_n \in \R^2$ and $y \in \R^n$, our goal is to compute $z \in \R^n$ given by $z_i := \sum_{j\neq i} |\langle x_i, x_j\rangle| \cdot y_j$. We will first compute $z'_i := \sum_{j} |\langle x_i, x_j\rangle| \cdot y_j$, and then subtract $|\langle x_i, x_i\rangle| \cdot y_i$ for each $i$ to get $z_i$.

We first sort the input vectors by their polar coordinate angle, and relabel so that $x_1, \ldots, x_n$ are in sorted order. Let $\phi_i$ be the polar coordinate angle of $x_i$. We will maintain two vectors $x^+, x^- \in \R^2$.  We will `sweep' an angle $\theta$ from $0$ to $2\pi$, and maintain that $x^+$ is the sum of the $y_i \cdot x_i$ with $\langle x_i, \theta \rangle > 0$, and $x^-$ is the sum of the other $y_i \cdot x_i$. Initially let $\theta = 0$ and let $x^+$ be the sum of the $y_i \cdot x_i$ with $\phi_i \in [-\pi/2,\pi/2)$, and $x^-$ be the sum of the remaining $y_i \cdot x_i$. As we sweep, whenever $\theta$ is in the direction of an $x_i$ we can set $z'_i = \langle x_i, x^+ - x^-\rangle$. Whenever $\theta$ is orthogonal to an $x_i$, we swap $y_i \cdot x_i$ from one of $x^+$ or $x^-$ to the other. Over the whole sweep, each point is swapped at most twice, so the total running time is indeed $n^{1 + o(1)}$.
\end{proof}

\begin{lemma} \label{lem:solvemaxip}
For the function $\k(x,y) = |\langle x , y \rangle|$, if the $\KAdjE$ problem (Problem~\ref{pro:KAdjE}) with error $1/n^{\omega(1)}$ can be solved in time $\T(n,d)$, and $n>d+1$, then $\Z$-{\sf MaxIP} (Problem~\ref{pro:Z_maxip}) with $d$-dimensional vectors of integer entries of bit length $O(\log n)$ can be solved in time $O(\T(n,d) \log^2 n + nd)$.
\end{lemma}

\begin{proof}
Let $x_1, \ldots, x_n \in \Z^d$ be the input vectors. Let $M \leq n^{O(1)}$ be the maximum magnitude of any entry of any input vector. Thus, $\max_{i \neq j} \langle x_i, x_j \rangle$ is an integer in the range $[-dM^2, dM^2]$. We will binary search for the answer in this interval. The total number of binary search steps will be $O(\log(dM^2)) \leq O(\log n)$.

We now show how to do each binary search step. Suppose we are testing whether the answer is $\leq a$, i.e. testing whether $\langle x_i, x_j \rangle \leq a$ for all $i \neq j$. Let $S_1, \ldots, S_{\log n} \subseteq \{ 1, \ldots, n \}$ be subsets such that for each $i \neq j$, there is a $k$ such that $|S_k \cap \{i,j\}| = 1$. For each $k \in \{1, \ldots, \log n\}$ we will show how to test whether there are $i,j$ with $|S_k \cap \{i,j\}| = 1$ such that $\langle x_i, x_j \rangle \leq a$, which will complete the binary search step.

Define the vectors $x'_1, \ldots, x'_n \in \Z^{d+1}$ by $(x'_i)_{j} = (x_i)_{j}$ for $j \leq d$, and $(x'_i)_{d+1} = a-1$ if $i \in P_k$, and $(x'_i)_{d+1} = -1$ if $i \notin P_k$. Hence, for  $i,j$ with $|S_k \cap \{i,j\}| = 1$ we have $\langle x'_i, x'_j \rangle = \langle x_i, x_j \rangle - a+1$, and so our goal is to test whether there are any such $i,j$ with $\langle x'_i, x'_j \rangle > 0$.

Let $v_k \in \{0,1\}^n$ be the vector with $(v_k)_{i} = 1$ when $i \in P_k$ and $(v_k)_{i} = 0$ when $i \notin P_k$. Use the given algorithm to vector $v_k$, we can compute a $(a \pm n^{-\omega(1)})$ approximation to 
\begin{align*}
s_1 := \sum_{i \in P_k} \sum_{j \notin P_k} |\langle x'_i, x'_j \rangle|
\end{align*}
in time $O(\T(n,d))$. Similarly, using the fact that the corresponding matrix has rank $d$ by definition, we can exactly compute 
\begin{align*}
s_2 := \sum_{i \in P_k} \sum_{j \notin P_k} \langle x'_i, x'_j \rangle
\end{align*}
in time $O(nd)$. Our goal is to determine whether $s_1 = -s_2$. Since each $s_1$ and $s_2$ is a polynomially-bounded integer, and we have a superpolynomially low error approximation to each, we can determine this as desired.
\end{proof}

Combining Lemma~\ref{lem:solvemaxip} with Theorem~\ref{thm:maxiphard} we get:

\begin{corollary} \label{cor:lowdimmultabsinnerproduct}
Assuming {\sf SETH}, there is a constant $c$ such that for the function $\k(x,y) = |\langle x , y \rangle|$, the $\KAdjE$ problem (Problem~\ref{pro:KAdjE}) with error $1/n^{\omega(1)}$ and dimension $d = c^{\log^* n}$ vectors of $O(\log n)$ bit entries requires time $n^{2 - o(1)}$.
\end{corollary}

Similarly, combining with Theorem~\ref{thm:maxiplowdim} we get:

\begin{corollary}
For the function $\k(x,y) = |\langle x , y \rangle|$, if the $\KAdjE$ problem (Problem~\ref{pro:KAdjE}) with error $1/n^{\omega(1)}$ and dimension $d = 3$ vectors of $O(\log n)$ bit entries can be solved in time $n^{4/3 - O(1)}$, then we would get a faster-than-known algorithm for $\Z$-${\sf MaxIP}$ (Problem~\ref{pro:Z_maxip}) in dimension $d=3$.
\end{corollary}

The same proof, but using Theorem~\ref{thm:lownnhard} instead of Theorem~\ref{thm:maxiphard}, can also show hardness of thresholds of distance functions:

\begin{corollary} \label{cor:lowdimmultabsthresholddist}
Corollary~\ref{cor:lowdimmultabsinnerproduct} also holds for the function $\k(x,y) = \mathrm{TH}(\|x-y\|_2^2)$, where $\mathrm{TH}$ is any threshold function (i.e. $\mathrm{TH}(z) = 1$ for $z \geq \theta$ and $\mathrm{TH}(z) = 0$ otherwise, for some $\theta \in \R_{>0}$).
\end{corollary}

Using essentially the same proof as for Lemma~\ref{lem:spars-reduc} in Section~\ref{sec:hardnessnonlip}, we can further extend Corollary~\ref{cor:lowdimmultabsthresholddist} to any non-Lipschitz functions $f$:

\begin{proposition}\label{prop:low-mult-hard}
Suppose $f : \R_+ \to \R$ is any function which is not $(C,L)$-multiplicatively Lipschitz for any constants $C,L \geq 1$. Then, assuming {\sf SETH}, the $\KAdjE$ problem (Problem~\ref{pro:KAdjE}) with error $1/n^{\omega(1)}$ and dimension $d = c^{\log^* n}$ requires time $n^{2 - o(1)}$.
\end{proposition}

\subsection{Hardness of the \texorpdfstring{$n$}{}-Body Problem}

We now prove Corollary~\ref{cor:nbodyhardness} from the Introduction, showing that our hardness results for the $\KAdjE$ problem (Problem~\ref{pro:KAdjE}) also imply hardness for the $n$-body problem.

\begin{corollary}[Restatement of Corollary~\ref{cor:nbodyhardness}]
Assuming $\SETH$, there is no 
\begin{align*}
\poly( d,\log(\alpha) ) \cdot n^{1+o(1)}
\end{align*}-time algorithm for one step of the $n$-body problem. 
\end{corollary}

\begin{proof}
We reduce from the $\k$ graph Laplacian multiplication problem, where $\k(u,v) = f(\|u - v\|_2^2)$ and $f(z) = \frac{1}{(1 + z)^{3/2}}$. A $\k$ graph Laplacian multiplication instance consists of the $\k$ graph $G$ on a set of $n$ points $X\subseteq \mathbb{R}^d$ and a vector $y\in \{0,1\}^n$ for which we wish to compute $L_Gy$. Think of $y$ as vector with coordinates in the set $X$. Compute this multiplication using an $n$-body problem as follows:

\begin{enumerate}
    \item For each $b\in \{0,1\}$, let $X_b = \{x\in X ~|~ y_x = b\}$. Let $Z\subseteq \mathbb{R}^{d+1}$ be the set of all $(x,0)$ for $x\in X_0$ and $(x,1)$ for $x \in X_1$.
    \item Solve the one-step $n$-body problem on $Z$ with unit masses. Let $z \in \mathbb{R}^n$ be the vector of the $(d+1)$-th coordinate of these forces, with forces negated for coordinates $x \in X_0$.
    \item Return $-z/G_{\text{grav}}$ (Note that $G_{\text{grav}}$ is the Gravitational constant)
\end{enumerate}

We now show that $z = L_G y$. For $x\in X_b$, $(L_G y)_{x} = (-1)^{1-b} \sum_{x'\in X_{1-b}} \k(x,x')$. We now check that $z_x$ is equal to this by going through pairs $\{x,x'\}$ individually. Note that the $(d+1)$-th coordinate of the force between $(x,b)$ and $(x',b)$ is 0. The $(d+1)$-th coordinate of the force exerted by $(x',1)$ on $(x,0)$ is 
\begin{align*}
\frac{ G_{\text{grav}} }{\|(x,0)-(x',1)\|_2^2} \cdot \frac{(1-0)}{\|(x,0)-(x',1)\|_2} = G_{\text{grav}} \cdot \k(x,x')
\end{align*}
Negating this gives the force exerted by $(x',0)$ on $(x,1)$. All of these contributions agree with the corresponding contributions to the sum $(L_Gy)_x$, so $-z/G_{\text{grav}} = L_G \cdot y$ as desired.

The runtime of this reduction is $O(n)$ plus the runtime of the $n$-body problem. However, Theorem \ref{thm:informal-high} shows that no almost-linear time algorithm exists for $\k$-Laplacian multiplication, since $f$ is not approximable by a polynomial with degree less than $\Theta(\log n)$. Therefore, no almost-linear time algorithm exists for $n$-body either assuming $\SETH$, as desired.
\end{proof}

\subsection{Hardness of Kernel PCA} \label{subsec:kernelPCA}

For any function $\k : \R^d \times \R^d \to \R$, and any set $P = \{x_1, \ldots, x_n\} \subseteq \R^d$ of $n$ points, define the matrix $K_{\k,P} \in \R^{n \times n}$ by 
\begin{align*}
K_{\k,P}[i,j] = \k(x_i, x_j)
\end{align*}

$A_{\k,P}$ and $K_{\k,P}$ differ only on their diagonal entries, so a $n^{1 + o(1)}$ time algorithm for multiplying by one can be easily converted into such an algorithm for the other. Kernel PCA studies 

\begin{problem}[$\k$ PCA]\label{pro:KPCA}
For a given function $\k : \R^d \times \R^d \to \R$, the $\k$ PCA problem asks: Given as input a set $P =\{x_1, \ldots, x_n\} \subseteq \R^d$ with $|P|=n$, output the $n$ eigenvalues of the matrix $(I_n-J_n) \times K_{\k,P} \times (I_n - J_n)$, where $J_n$ is the $n \times n$ matrix whose entries are all $1/n$. In $\eps$-approximate $\k$ PCA, we want to return a $(1 \pm \eps)$-multiplicative approximation to each eigenvalue.
\end{problem}

We can now show a general hardness result for $\k$ PCA:

\begin{theorem}[Approximate] \label{thm:hardnessapproxpca}
For every function $f:\mathbb{R}\rightarrow \mathbb{R}$ which is equal to a Taylor expansion $f(x) = \sum_{i=0}^{\infty} c_i x^i$ on an interval $(0,1)$, if $f$ is not $\eps$-approximated by a polynomial of degree $O(\log n)$ (Definition~\ref{def:closetopoly}) on an interval $(0,1)$ for $\eps = 2^{-\log^4 n}$, then, assuming {\sf SETH}, the $\eps$-approximate $\k$ PCA problem (Problem~\ref{pro:KAdjE}) in dimension $d = O(\log n)$ requires time $n^{2-o(1)}$.
\end{theorem}

\begin{proof}[Proof Sketch]
Theorem~\ref{thm:hardnessapproxpca} follows almost directly from Theorem~\ref{thm:hardnessapprox} when combined with the reduction from \cite[{Section~5}]{bcis18}. The idea is as follows: Suppose we are able to estimate the $n$ eigenvalues of $(I_n-J_n) \times K_{\k,P} \times (I_n - J_n)$. Then, in particular, we can estimate their sum, which is equal to:
$$\tr((I_n-J_n) \times K_{\k,P} \times (I_n - J_n)) = \tr(K_{\k,P} \times (I_n - J_n)^2)  = \tr(K_{\k,P} \times (I_n - J_n)) = \tr(K_{\k,P}) - S(K_{\k,P})/n,$$
where $S(K_{\k,P})$ denotes the sum of the entries of $K_{\k,P}$. We can compute $\tr(K_{\k,P})$ exactly in time $n^{1+o(1)}$, so we are able to get an approximation to $S(K_{\k,P})$. However, in the proof of Theorem~\ref{thm:hardnessapprox}, we showed hardness for approximating $S(K_{\k,P})$, which concludes our proof sketch.
\end{proof} 
\section{Sparsifying Multiplicatively Lipschitz Functions in Almost Linear Time}\label{sec:sparsify-lipschitz}

Given $n$ points, let $\alpha$ denote 
\begin{align*}
\alpha = \frac{\max_{u, v \in P} (\|u-v\|_2^2)}{\min_{u, v \in P} (\|u-v\|_2^2)}.
\end{align*}

In this section, we give an algorithm to compute sparsifiers for a large class of kernels $\k$ in almost linear time in $nd$, with logarithmic dependency on $\alpha$ and $1/\eps^2$ dependence on $\eps$. When $d = \log n$, our algorithm runs in almost linear time in $n$.  
To formally state our main theorem, we define multiplicatively Lipschitz functions:
\begin{definition}\label{def:mult-lip}
Let $C \geq 1$ and $L \geq 1$. A function $f: \mathbb{R}_{\geq 0} \rightarrow \mathbb{R}_{\geq 0}$ is $(C, L)$-multiplicatively Lipschitz iff for all $c \in [1/C, C]$, we have:
  \begin{align*}
  \frac{1}{C^L} < \frac{f(cx)}{f(x)} < C^L, \forall x \in \mathbb{R}_{\geq 0}.
  \end{align*}
\end{definition}

\textbf{Examples:} Any polynomial with non-negative coefficients and maximum degree $q$ is $(1+\eps, q)$ multiplicatively Lipschitz for any $\eps > 0$. The function $f(x) = 1$ when $x < 1$ and $f(x) = 2$ when $x \geq 1$ is $(2, 1)$ multiplicatively Lipschitz.

The following lemma is a simple consequence of our definition of multiplicatively Lipschitz functions:
\begin{lemma}\label{lem:lip} 
Let $C \geq 1$ and $L > 0$. Any function $f : \R_{\geq 0} \rightarrow \R_{\geq 0}$ that is $(C, L)$-multiplicatively Lipschitz satisfies for all $c \in (0,1/C) \cup [C, +\infty)$ :
\begin{align*}
\frac{1}{c^{2L}} < \frac{f(cx)}{f(x)} < c^{2L} .
\end{align*}
\end{lemma}

We now state the core theorem of this section: 
\begin{theorem}\label{thm:sparsify-lipschitz}
Let $C\geq 1$ and $L \geq 1$. 
  Consider any $(C, L)$-multiplicatively
  Lipschitz function $f : \R \rightarrow \R$, and let $\k(x,y)=f(\|x-y\|_2^2)$. Let $P$ be a set of $n$ points in $\mathbb{R}^d$. For any $k \in [\Omega(1) , O(\log n)]$  
  such that $C = n^{O(1/k)}$, there exists an algorithm ~\textsc{sparsify-$\k$-graph}$(P,n,d, \k, k, L)$ (Algorithm~\ref{alg:sparsify-lipschitz}) that runs in time: 
  \begin{align*}
  O( ndk )  + 
  \eps^{-2} \cdot n^{1+O(L/k)} 2^{O(k)} \log n \cdot \log \alpha 
  \end{align*}
  and outputs an $\eps$-spectral sparsifier $H$ of the $\k$ graph with $|E_H| = O(n \log n /\eps^2)$.
\end{theorem}
 We give a corollary of this theorem.
\begin{corollary}\label{cor:poly-sparse}
Consider a $\k$ graph where $\k(x,y) = f(\|x-y\|_2^2)$, and $f$ is $(2, L)$-multiplicatively Lipschitz. Let $G$ denote the $\k$ graph from $n$ points in $\mathbb{R}^d$. There is an algorithm that takes in time
\begin{align*}
O(n d \sqrt{L \log n} ) + \eps^{-2} \cdot n \cdot 2^{O(\sqrt{L \log n} \log\log n )} \cdot \log \alpha  
\end{align*}
and outputs an $\eps$-spectral sparsifier $H$ of $G$ with $|E_H| = O( n \log n / \eps^2)$. 
If $L = o( \log n/ (\log \log n)^2)$, this runs in time 
\begin{align*}
o(nd \log n) + \eps^{-2}  n^{1+o(1)} \log \alpha.
\end{align*}
\end{corollary}
\begin{proof} 
Set $k = \sqrt{L \log n}$, and the corollary follows from Theorem~\ref{thm:sparsify-lipschitz}.
\end{proof}

This implies that if $f$ is a polynomial with non-negative coefficients, then sparsifiers of the corresponding $\k$-graph can be found in almost linear time. The same result holds if $f$ is the reciprocal of a polynomial with non-negative coefficients.

We will need a few geometric preliminaries in order to present our core algorithm of this section,~\textsc{Sparsify-$\k$-graph}.

\begin{definition}[$\epsilon$-well separated pair]\label{def:wsp}
Given two sets of points $A$ and $B$. We say $A,B$ is an $\epsilon$-well separated pair if the diameter of $A_i$ and $B_i$ are at most $\epsilon$ times the distance between $A_i$ and $B_i$.
\end{definition}

\begin{definition}[$\epsilon$-well separated pair decomposition \cite{ck95}]\label{def:wspd}
An $\eps$-well separated pair decomposition ($\eps$-$\WSPD$) of a given point set $P$ is a family of pairs $\mathcal{F} = \{(A_1, B_1), \ldots (A_s, B_s)\}$ with $A_i, B_i \subset P$ such that:
\begin{itemize}
    \item $\forall i \in [s]$, $A_i$, $B_i$ are $\eps$-well separated pair (Definition~\ref{def:wsp}) 
    \item For any pair $p, q \in P$, there is a unique $i \in [s]$ such that $p \in A_i$ and $q \in B_i$
\end{itemize}
\end{definition}
A famous theorem of Callahan and Kosaraju \cite{ck95} states:
\begin{theorem}[Callahan and Kosaraju \cite{ck95}]\label{thm:wspd}
Given any point set $P \subset \mathbb{R}^d$ and $0 \leq \eps \leq 9/10$, an $\eps$-$\WSPD$ (Definition~\ref{def:wspd}) of size $O(n /\eps^d)$ can be found in $2^{O(d)} \cdot ( n \log n + n / \eps^d ) $ time. Moreover, each vertex participates in at most $2^{O(d)} \cdot \log \alpha$ $\epsilon$-well separated pairs.
\end{theorem}
Well-separated pairs can be interpreted as complete bipartite graphs on the vertex set, or \textbf{bicliques}. The biclique associated with a well-separated pair is the bipartite graph connecting all vertices on one side of the pair to another. 

This concludes our definitions on well-separated pairs. We now give names to some algorithms in past work, which will be used in our algorithm \textsc{sparsify-$k$-graph}. We define the algorithm $\textsc{GenerateWSPD}(P, \eps)$ to output an $\eps$-WSPD (Definition~\ref{def:wspd}) of $P$. We define the algorithm $\textsc{RandomProject}(P, k)$ to be a random projection of $P$ onto $k$ dimensions. 

Let $\textsc{Biclique}(\k, P, A, B)$ be the complete biclique on the $\k$-graph of $P$ with one side of the biclique having verticese corresponding to points in $A$, and the other side having vertices corresponding to points in $B$. We store this biclique implicitly as $(A, B)$ rather than as a collection of edges. 

Let $\textsc{RandSample}(G, s)$ be an algorithm uniformly at random sampling $O(s)$ edges from $G$, where the big $O$ is the same constant as the big $O$ in the $n^{O(1/k)}$ from Lemma~\ref{lem:low-dim-jl}.

Let $\textsc{SpectralSparsify}(G, \eps)$ be any nearly linear time spectral sparsification algorithm that outputs a $(1+\eps)$ spectral sparsifier with $O(n \log n / \eps^2)$ edges, such as that in Theorem~\ref{thm:ss11} from $\cite{ss11}$.
\begin{algorithm}[!ht]\caption{}\label{alg:sparsify-lipschitz}
\begin{algorithmic}[1]
\Procedure{\textsc{Sparsify-$\k$-Graph}}{$P, n, d, \k, k, L, \eps$} \Comment{Theorem~\ref{thm:sparsify-lipschitz}}

    \State \textbf{Input}: A point set $P$ with $n$ points in dimension $d$, a kernel function $\k(x,y) = f(\|x-y\|_2^2)$, an integer variable $\Omega(1) \leq k \leq O(\log n)$, and a variable $L$, and error $\eps$.
    
    \State \textbf{Output:} A candidate sparsifier of the $\k$-graph on $P$.
    
    \State
    $P' \gets \textsc{RandomProject}(P, k)$
    
    \State $H \gets $ empty graph with $n$ vertices.
    \State $\{(A_1', B_1'), \ldots (A_t', B_t')\} \gets \textsc{GenerateWSPD}(P', n , d, 1/2)$ \Comment{$t = n \cdot 2^d$}
    \For{$ i = 1 \to t$}
        \State Find $(A_i, B_i)$ corresponding to $(A_i', B_i')$, where $A_i, B_i \subset P$.
        \State $Q \gets \textsc{BiClique}(\k, P, A_i, B_i)$.
        \State $s \leftarrow \eps^{-2} n^{O(L/k)} (|A_i|+|B_i|) \log (|A_i|+|B_i|)$
        \State $\overline{Q} \gets \textsc{RandSample}(Q, s)$
        \State $\overline{Q} \gets \overline{Q}$ with each edge scaled by $ |A_i||B_i| / s $.
        \State $H \gets H + \overline{Q}$
    \EndFor
    \State $H \gets \textsc{SpectralSparsify}( H , \epsilon )$ \Comment{Corollary \ref{cor:ss11}}
    \State Return $H$
\EndProcedure
\Procedure{\textsc{GenerateWSPD}}{$P,n,d,\epsilon$} \Comment{Theorem~\ref{thm:wspd}}
    \State ... \Comment{See details in \cite{ck95}}
\EndProcedure
\Procedure{\textsc{RandSample}}{$G,s$}
    \State Sample $O(s)$ edges from $G$ and generate a new graph $\ov{G}$
    \State \Return $\ov{G}$
\EndProcedure
\Procedure{\text{RandomProject}}{$P,d,k$}
    \State $P' \leftarrow \emptyset$
    \State Choose a JL matrix $S \in \R^{k \times d}$
    \For{$x \in P$}
        \State $x'\leftarrow S \cdot x $ 
        \State $P' \leftarrow P' \cup x'$
    \EndFor
    \State \Return $P'$
\EndProcedure
\end{algorithmic}
\end{algorithm}

The rest of this section is devoted to proving Theorem~\ref{thm:sparsify-lipschitz}.
\subsection{High Dimensional Sparsification}

We are nearly ready to prove Theorem~\ref{thm:sparsify-lipschitz}. We start with a Lemma:

\begin{lemma}[$(C, L)$ multiplicative Lipschitz functions don't distort a graph's edge weights much]\label{lem:lipschitz-distortion} 
Consider a complete graph $G$, and a complete graph $G'$, where vertices of $G$ are identified with vertices of $G'$ (which induces an identification between edges). Let $K \geq 1$. Suppose each edge in $G$ satisfies:
\begin{align*} 
\frac{1}{K} \cdot w_{G'}(e) \leq w_G(e) \leq K \cdot w_{G'}(e) 
\end{align*}
If $f$ is a $(C, L)$ multiplicative Lipschitz function, and $f(G)$ refers to the graph $G$ where $f$ is applied to each edge length, and $C < K$ then: 
\begin{align*}
\frac{1}{K^{2L}} \cdot w_{G'}(e) \leq w_{f(G)}(e) \leq K^{2L} \cdot w_{G'}(e).
\end{align*}
\end{lemma}
\begin{proof} 
This follows from Lemma~\ref{lem:lip}.
\end{proof}

\begin{proof} (of Theorem~\ref{thm:sparsify-lipschitz}):
The Algorithm~\textsc{sparsify-$\k$-graph} starts by performing a random projection of point set $P$ into $k$ dimensions. Call the new point set $P'$. This runs in time $O(ndk)$, and incurs distortion $n^{O(1/k)}$, as seen in Lemma~\ref{lem:low-dim-jl}. 
Next, our algorithm performs a $1/2$-$\WSPD$ on $P'$. As seen in Theorem~\ref{thm:wspd}, this runs in time 
    \begin{align*}
    O(n \log n + n2^{O(k)} )
    \end{align*}  
    We view each well-separated pair $(A_i', B_i')$ on $P'$ as a biclique, where the edge length between any two points in $P'$ corresponds to the edge length between those two points in the original $\k$-graph.
    By the guarantees of Theorem~\ref{thm:wspd},
    the longest edge divided by the shortest edge between two sides of a well-separated pair in $P'$ is at most $2$. Thus, the longest edge divided by the shortest edge within any induced bipartite graph on the $\k$-graph is $2 \cdot n^{O(1/k)}$, by Lemma~\ref{lem:lipschitz-distortion}.

    For each such biclique, the leverage score for each edge is overestimated by 
  \begin{align*}
  n^{O(L/k)} \cdot (|A_i'|+|B_i'|) / (|A_i'||B_i'|).
  \end{align*} 
  This comes from first applying Lemma~\ref{lem:lipschitz-distortion} to upper bound the ratio of the longest edge in a biclique divided by the shortest edge.  This ratio comes out to be $2 n^{O(L/k)}$.  Now recall the definition of leverage score on graphs as $w_e R_e$, where $R_e$ is the effective resistance assuming conductances of $w_e$ on the graph, and $w_e$ is the edge weight. Here, $w_e$ is upper bounded by the longest edge length, and $R_e$ is upper bounded by the leverage score of a biclique supported on the same edges, where all edges lengths are equal to the shortest edge length (this is an underestimate of effective resistance due to Rayleigh monotonicity, see~\cite{c97} for details). Therefore, a leverage score overestimate of the graph can be obtained by $n^{O(L/k)} \cdot (|A_i'|+|B_i'|) / (|A_i'||B_i'|)$, as claimed. The union of these graphs is a spectral sparsifier of our $\k$-graph.
  
  Finally, our algorithm samples
  $n^{O(L/k)}(|A_i'|+|B_i'|)\log(|A_i'|+|B_i'|)$ edges uniformly at random from each biclique, scaling each sampled edge's weight so that the expected value of the sampled graph is equal to the original biclique. Each vertex participates in at most $\log \alpha 2^{O(k)}$ bicliques (see Theorem~~\ref{thm:wspd}). Thus, this uniform sampling procedures' run time is bounded above by 
  \[ 
  n^{1+O(L/k)}2^{O(k)} \log n \cdot \log \alpha.
  \]
  
  Finally, our algorithm runs a sparsification algorithm on our graph after uniform sampling, which gets the edge count of the final graph down to $O(n \log n / \eps^2)$.
This completes our proof of Theorem~\ref{thm:sparsify-lipschitz}.
\end{proof}

\subsection{Low Dimensional Sparsification}

We now present a result on sparsification in low dimensions, when $d$ is assumed to be small or constant. 

\begin{theorem}\label{thm:low-sparsify-lipschitz}
Let $L \geq 1$. Consider a $\k$-graph with $n$ vertices arising from a point set in $d$ dimensions, and let $\alpha$ be the ratio of the maximum Euclidean distance to the minimum Euclidean distance in the point set . Let $f$ be a $(1+1/L, L)$ multiplicatively Lipschitz function. Then an $\eps$ spectral sparsifier of the $\k$-graph can be found in time
\begin{align*}
    \eps^{-2} \cdot n  \cdot \log n \cdot \log \alpha   \cdot (2L)^{O(d)} 
\end{align*}
\end{theorem}
\begin{proof} We roughly follow the proof of Theorem~\ref{thm:sparsify-lipschitz}, except without the projection onto low dimensions. Now, on the $d$ dimensional data, we create a $1/L$-$\WSPD$. This takes time 
\begin{align*}
O( n \log n ) + n \cdot (2L)^{O(d)}.
\end{align*}
Since $f$ is $(1+1/L,L)$-multiplicatively Lipschitz, it follows that within each biclique of the $\k$-graph induced by the $\WSPD$ (Definition~\ref{def:wspd}), the maximum edge length divided by the minimum edge length is bounded above by $(1+1/L)^{O(L)} = O(1)$. Therefore, performing scaled and reweighted uniform sampling on each biclique takes $O(s \log s)$ time if there are $s$ vertices in the biclique, and gives a sparsified biclique with $O(s \log s)$ edges. 

Taking the union of this number over all bicliques gives an algorithm that runs in time
\begin{align*}
 \eps^{-2} \cdot n \cdot \log n \cdot \log \alpha \cdot (2L)^{O(d)} 
\end{align*}
as desired.
\end{proof}

\section{Sparsifiers for \texorpdfstring{$| \langle x,y \rangle|$}{}}

In this section, we construct sparsifiers for Kernels of the form $| \langle x,y \rangle |$.

\begin{lemma}[sparsification algorithm for inner product kernel]\label{lem:inner-sparsify}
Given a set of vectors $X\subseteq \mathbb{R}^d$ with $|X|=n$ and accuracy parameter $\epsilon \in (0,1/2)$, there is an algorithm that runs in $\epsilon^{-2} n \cdot \poly(d, \log n)$ time, and outputs a graph $H$ that satisfies both of the following properties with probability at least $1 - 1/\poly(n)$:
\begin{enumerate}
    \item $(1 - \epsilon)L_G\preceq L_H\preceq (1 + \epsilon)L_G$;
    \item $|E(H)|\le  \epsilon^{-2} \cdot n \cdot \poly(d,\log n)$.
\end{enumerate}
where $G$ is the $\k$-graph on $X$, where $\k(x,y) = |\langle x,y\rangle|$.
\end{lemma}

Throughout this section, we specify various values $C_i$. For each subscript $i$, it is the case that $1\le C_i\le \text{poly}(d,\log n)$.

\subsection{Existence of large expanders in inner product graphs}\label{sec:inner_product_kernel_sparsifier_algorithm}

We start by showing that certain graphs that are related to unweighted versions of $\k$-weighted graphs contain large expanders:

\begin{definition}[$k$-dependent graphs]
For a positive integer $k > 1$, call an unweighted graph $G$ $k$-\emph{dependent} if no independent set with size at least $k+1$ exists in $G$.
\end{definition}

We start by observing that inner product graphs are $(d+1)$-dependent.

\begin{definition}[inner product graphs]
For a set of points $X\subseteq \mathbb{R}^d$, the \emph{unweighted inner product graph for } $X$ is a graph $G$ with vertex set $X$ and unweighted edges $\{u,v\}\in E(G)$ if and only if $|\langle u,v\rangle| \ge \frac{1}{d+1} \|u\|_2 \|v\|_2$. The \emph{weighted inner product graph for } $X$ is a complete graph $G$ with edge weights $w_e$ for which $w_{uv} = |\langle u,v\rangle|$.
\end{definition}

We now show that these graphs are $k$-dependent:

\begin{lemma}\label{lem:indep-rank}
Suppose $M \in \R^{n \times n}$ such that $M[i,i] = 1$ for all $i \in [n]$, and $|M[i,j]| < 1/n$ for all $i \neq j$. Then, $M$ has full rank.
\end{lemma}

\begin{proof}
Assume to the contrary that $M$ does not have full rank. Thus, there are values $c_1, \ldots, c_{n-1} \in \R$ such that for all $i \in [n]$ we have $\sum_{j=1}^{n-1} c_j M[i,j] = M[i,n]$.

First, note that there must be a $j$ with $|c_j| \geq 1 + 1/n$. Otherwise, we would have
$$1 = |M[n,n]| = \left| \sum_{j=1}^{n-1} c_j M[n,j] \right| < \frac{1}{n}  \sum_{j=1}^{n-1} \left| c_j \right|  < \frac{n-1}{n} (1 + 1/n) < 1.$$
Assume without loss of generality that $|c_1| \geq |c_j|$ for all $j \in \{2,3,\ldots,n-1\}$, so in particular $|c_1| \geq 1 + 1/n$. Letting $c_n = -1$, this means that $\sum_{j=1}^n c_j M[1,j] = 0$, and so $M[1,1] = -\sum_{j=2}^n (c_j / c_1) M[1,j]$. Thus,
\begin{align*}
 1 = |M[1,1]| = \left| \sum_{j=2}^n \frac{c_j }{ c_1} M[1,j] \right| \leq  \sum_{j=2}^n \left|\frac{c_j }{ c_1} M[1,j] \right| <  \sum_{j=2}^n \left| M[1,j] \right| < (n-1)\cdot \frac{1}{n} < 1,
\end{align*}
a contradiction as desired.
\end{proof}

\begin{proposition}\label{prop:inner-product-dep}
The unweighted inner product graph for $X\subseteq \mathbb{R}^d$ is $(d+1)$-dependent.
\end{proposition}

\begin{proof}
For an independent set $S$ in the unweighted inner product graph $G$ for $X$, define an $S\times S$ matrix $M$ with $M[i,j] = \langle s_i,s_j\rangle$ where $S = \{s_1,s_2,\hdots,s_{|S|}\}$. Then Lemma \ref{lem:indep-rank} coupled with the definition for edge presence in $G$ shows that $M$ is full rank. However, $M$ is a rank $d$ matrix because it is the matrix of inner products for dimension $d$ vectors. Therefore, $d\ge |S|$, so no independent set has size greater than $d$.
\end{proof}

Next, we show that $k$-dependent graphs are dense:

\begin{proposition}[dependent graph is dense]\label{prop:k-dep-dense}
Any $k$-dependent graph $G$ has at least $n^2/(2k^2)$ edges.
\end{proposition}

\begin{proof}
Consider any $k+1$-tuple of vertices in $G$. There are $\binom{n}{k+1}$ such $k+1$-tuples. By definition of $k$-dependence, there must be some edge with endpoints in any $k+1$-tuple. The number of $k$-tuples that any given edge can be a part of is at most $\binom{n-2}{k-1}$. Therefore, the number of edges in the graph is at least
\begin{align*}
\frac{ \binom{n}{k+1} }{ \binom{n-2}{k-1} } \ge \frac{ n(n-1)}{ (k+1)k } \ge \frac{ n^2 }{ 2k^2 }
\end{align*}
as desired.
\end{proof}

Next, we argue that unweighted inner product graphs have large expanders:

\begin{proposition}[every dense graph has a large expander]\label{prop:dense-has-expander}
Consider an unweighted graph $G$ with $n$ vertices and at least $n^2/c$ edges for some $c > 1$. Then, there exists a set $S\subseteq V(G)$ with the following properties:

\begin{enumerate}
    \item (Size) $|S|\ge n/(40c)$
    \item (Expander) $\Phi_{G[S]}\ge 1/(100c\log n)$
    \item (Degree) The degree of each vertex in $S$ within $G[S]$ is at least $n/(10000c)$
\end{enumerate}

\end{proposition}

\begin{proof}
We start by partitioning the graph as follows:

\begin{enumerate}
    \item Initialize $\mathcal F = \{V(G)\}$
    
    \item While there exists a set $U\in \mathcal F$ with (a) a partition $U = U_1\cup U_2$ with $U_1$ cut having conductance $\le 1/(100 c\log n)$ \textbf{or} (b) a vertex $u$ with degree less than $n/(10000c)$ in $G[U]$
    
    \begin{enumerate}
        \item If (a), replace $U$ in $\mathcal F$ with $U_1$ and $U_2$
        \item Else if (b), replace $U$ in $\mathcal F$ with $U\setminus \{u\}$ and $\{u\}$
    \end{enumerate}
\end{enumerate}
We now argue that when this procedure stops,

$$\sum_{U\in \mathcal F} |E(G[U])| \ge n^2/(2c)$$
To prove this, think of each splitting of $U$ into $U_1$ and $U_2$ as deleting the edges in $E(U_1,U_2)$ from $G$. Design a charging scheme that assigns deleted edges due to (a) steps to edges of $G$ as follows. Let $c_e$ denote the charge assigned to an edge $e$ and initialize each charge to 0. When $U$ is split into $U_1$ and $U_2$, let $U_1$ denote the set with $|E(U_1)| \le |E(U_2)|$. When $U$ is split, increase the charge $c_e$ for each $e\in E(U_1)\cup E(U_1,U_2)$ by $|E(U_1,U_2)|/|E(U_1)\cup E(U_1,U_2)|$.

We now bound the charge assigned to each edge at the end of the algorithm. By construction, $\sum_{e\in E(G)} c_e$ is the number of edges deleted over the course of the algorithm of type (a). Each edge is assigned charge at most $\log |E(G)|\le 2\log n$ times, because $|E(U_1)|\le \frac{|E(U_1)| + |E(U_2)|}{2}\le \frac{|E(U)|}{2}$ when charge is assigned to edges in $U_1$. Furthermore, the amount of charge assigned is the conductance of the cut deleted, which is at most $1/(100c\log n)$. Therefore, 
\begin{align*}
c_e\le \frac{ 2 \log n } { 100c\log n } = \frac{ 1 }{ 50c }
\end{align*}
for all edges in $G$, which means that the total number of edges deleted of type (a) was at most $n^2/(100c)$. Each type (b) deletion reduces the number of edges in $G$ by at most $n/(10000c)$, so the total number of type (b) edge deletions is at most $n^2/(10000c)$. Therefore, the total number of edges remaining is at least 
\begin{align*}
\frac{ n^2 }{ c } - \frac{ n^2 }{ 100c } - \frac{ n^2 }{ 10000c }  > \frac{ n^2 }{ 2c },
\end{align*}
as desired.

By the stopping condition of the algorithm, each connected component of $G$ after edge deletions is a graph with all cuts having conductance at least $1/(100c \log n)$ and all vertices having degree at least $n/(10000c)$. Next, we show that some set in $\mathcal F$ has at least $n/(40c)$ vertices. If this is not the case, then
\begin{align*}
\sum_{U\in \mathcal F} |E(G[U])|
\le & ~ \sum_{U\in \mathcal F} |U|^2\\
\le & ~ \sum_{U\in \mathcal F} \frac{ n}{40c } |U| & \text{~by~assuming~all~}|U|\leq n / (40c) \\
\le & ~ \frac{ n^2 }{ 40c } & \text{~by~} |U| \leq n \\
< & ~ \frac{n^2} {2c}
\end{align*}
which leads to a contradiction. Therefore, there must be some connected component with at least $n/(40c)$ vertices. Let $S$ be this component. By definition $S$ satisfies the \emph{Size} guarantee. By the stopping condition for $\mathcal F$, $S$ satisfies the other two guarantees as well, as desired.
\end{proof}

Proposition \ref{prop:dense-has-expander} does not immediately lead to an efficient algorithm for finding $S$. Instead, we give an algorithm for finding a weaker but sufficient object:

\Aaron{Replace constants}
\begin{proposition}[algorithm for finding sets with low effective resistance diameter]\label{prop:efficient-low-eff-res}
For any set of vectors $X\subseteq \mathbb{R}^d$ with $n = |X|$ and unweighted inner product graph $G$ for $X$, there is an algorithm $\LowDiamSet(X)$ that runs in time

$$\poly(d, \log n) n$$
and returns a set $Q$ that has the following properties with probability at least $1 - 1/\poly(n)$:

\begin{enumerate}
    \item (Size) $|Q| \ge n/C_1$, where $C_1 = 320d^2$
    \item (Low effective resistance diameter) For any pair of vertices $u,v\in Q$, $\Reff_G(u,v)\le \frac{C_2}{n}$, where $C_2 = (10d\log n)^10$
\end{enumerate}
\end{proposition}

We prove this proposition in Section \ref{subsec:eff-res-diam}.

\subsection{Efficient algorithm for finding sets with low effective resistance diameter} \label{subsec:eff-res-diam}

In this section, we prove Proposition \ref{prop:efficient-low-eff-res}. We implement $\LowDiamSet$ by picking a random vertex $v$ and checking to see whether or not it belongs to a set with low enough effective resistance diameter. We know that such a set exists by Proposition \ref{prop:dense-has-expander} and the fact that dense expander graphs have low effective resistance diameter. One could check that $v$ lies in this set in $\tilde{O}(n^2)$ time by using the effective resistance data structure of \cite{ss11}. Verifying that $v$ is in such a set in $\poly(d,\log n) n$ time is challenging given that $G$ is dense. We instead use the following data structure, which is implemented by uniformly sampling sparse subgraphs of $G$ and using \cite{ss11} on those subgraphs:

\begin{proposition}[Sampling-based effective resistance data structure]\label{prop:weak-sampling}
Consider a set $X\subseteq \mathbb{R}^d$, let $n = |X|$, let $G$ be the unweighted inner product graph for $X$, and let $S\subseteq X$ be a set with the following properties:

\begin{enumerate}
    \item (Size) $|S|\ge n/C_{3a}$, where $C_{3a} = 40\cdot 8d$
    \item (Expander) $\Phi_{G[S]}\ge 1/C_{3b}$, where $C_{3b} = 800d\log n$
    \item (Degree) The degree of each vertex in $S$ within $G[S]$ is at least $n/C_{3c}$, where $C_{3c} = 10000\cdot 8d$.
\end{enumerate}

There is a data structure that, when given a pair of query points $u,v\in X$, outputs a value $\ReffQuery(u,v)\in \mathbb{R}_{\ge 0}$ that satisfies the following properties with probability at least $1 - 1/\poly(n)$:

\begin{enumerate}
    \item (Upper bound) For any pair $u,v\in S$, $\ReffQuery(u,v) \le C_4\Reff_G(u,v)$, where $C_4 = 10$
    \item (Lower bound) For any pair $u,v\in X$, $\ReffQuery(u,v) \ge \Reff_G(u,v)/2$.
\end{enumerate}

The preprocessing method $\ReffPreproc(X)$ takes $\poly(d,\log n) n$ time and $\ReffQuery$ takes $\poly(d,\log n)$ time.
\end{proposition}

We now implement this data structure. $\ReffPreproc$ uniformly samples $O(\log n)$ subgraphs of $G$ and builds an effective resistance data structure for each one using \cite{ss11}. $\ReffQuery$ queries each data structure and returns the maximum:

\begin{algorithm}[!ht]
\begin{algorithmic}[1]
\Procedure{\ReffPreproc}{$X$}

    \State \textbf{Input}: $X\subseteq \mathbb{R}^d$ with unweighted inner product graph $G$
    
    \State $C_5\gets 1000\log n$
    
    \State $C_6\gets 80000CC_{3a}^2 C_{3b}^2 C_{3c}^2$, where $C$ is the constant in Theorem \ref{thm:oversampling}
    
    \For{$i$ from 1 through $C_5$}
    
        \State $H_i\gets $ uniformly random subgraph of $G$; sampled by picking $C_6 n$ uniformly random pairs $(u,v)\in X\times X$ and adding the $e = \{u,v\}$ edge to $H_i$ if and only if $\{u,v\}\in E(G)$ (that is when $|\langle u,v\rangle|\ge \frac{1}{d+1} \|u\|_2 \|v\|_2$).
        
        \State For each edge $e\in E(H_i)$, let $w_e = \frac{|E(G)|}{|E(H_i)|}$.
        
        \State $f_i\in \mathbb{R}^{X\times X}\gets$ approximation to $\Reff_{H_i}(u,v)$ given by the data structure of Theorem \ref{thm:ss11} for $\epsilon = 1/6$.
    
    \EndFor
    
\EndProcedure
\Procedure{\ReffQuery}{$u,v$}

    \State \textbf{Input}: A pair of points $u,v\in X$
    
    \State \textbf{Output}: An estimate for the $u$-$v$ effective resistance in $G$
    
    \State \Return $\max_{i=1}^{C_5} f_i(u,v)$
    
\EndProcedure
\end{algorithmic}
\end{algorithm}

Bounding the runtime of these two routines is fairly straightforward. We now outline how we prove the approximation guarantee for $\ReffQuery$. To obtain the upper bound, we use Theorem \ref{thm:oversampling} to show that $H_i$ contains a sparsifier for $G[S]$, so effective resistances are preserved within $S$. To obtain the lower bound, we use the following novel Markov-style bound on effective resistances:

\begin{lemma}\label{lem:sparse-lower-tail}
Let $G$ be a $w$-weighted graph with vertex set $X$ and assign numbers $p_e\in [0,1]$ to each edge. Sample a reweighted subgraph $H$ of $G$ by independently and identically selecting $q$ edges, with an edge chosen with probability proportional to $p_e$ and added to $H$ with weight $t w_e/(p_eq)$, where $t = \sum_{e \in G} p_e$. Fix a pair of vertices $u,v$. Then for any $\kappa > 1$,
\begin{align*}
    \Pr \left[\Reff_H(u,v)\le \Reff_G(u,v)/\kappa \right] \le 1/\kappa.
\end{align*}
\end{lemma}

\begin{proof}
For two vertices $u,v$, define the (folklore) \emph{effective conductance} between $u$ and $v$ in the graph $I$ to be
\begin{align*}
\Ceff_I(u,v) := \min_{q\in \mathbb{R}^n: q_u = 0,q_v = 1} \sum_{\text{ edges } \{x,y\} \in I} w_{xy}(q_x - q_y)^2 .
\end{align*}
It is well-known that $\Ceff_I(u,v) = 1/\Reff_I(u,v)$.

 Let 
\begin{align*}
 q^* = \arg\min_{q\in \mathbb{R}^n: q_u = 0, q_v = 1} \sum_{\{x,y\} \in E(G)} w_{xy}(q_x - q_y)^2.
\end{align*} 

Using $q^*$ as a feasible solution in the $\Ceff_H$ optimization problem shows that
\begin{align*}
\textbf{E}[\Ceff_H(u,v)] &\le \textbf{E} \left[\sum_{\{x,y\}\in E(G)} w_{xy}^H(q^*_x - q^*_y)^2 \right]\\
&= \sum_{\{x,y\}\in E(G)} \textbf{E}[w_{xy}^H](q^*_x - q^*_y)^2\\
&= \sum_{\{x,y\}\in E(G)} \frac{p_{xy}}{t}\sum_{i=1}^q (\frac{t w_{xy}^G}{q p_{xy}})(q^*_x - q^*_y)^2\\
&= \sum_{\{x,y\}\in E(G)} w_{xy}^G(q^*_x - q^*_y)^2\\
&= \Ceff_G(u,v)
\end{align*}
where $w^I_e$ denotes the weight of the edge $e$ in the graph $I$.

Finally, we have 
\begin{align*}
\Pr[\Reff_H(u,v) < \Reff_G(u,v)/\kappa]
= & ~ \Pr[\Ceff_H(u,v) > \kappa \cdot \Ceff_G(u,v)]\\
\le & ~ 1/\kappa,
\end{align*}
where the first step follows from $\Reff_G(u,v) = 1 / \Ceff_G(u,v)$, and the last step follows from Markov's inequality.
\end{proof}

We now prove Proposition \ref{prop:weak-sampling}:

\begin{proof}[Proof of Proposition \ref{prop:weak-sampling}]
\textbf{Runtime.} We start with preprocessing. Sampling $C_6 n$ pairs and checking if each pair satisfies the edge presence condition for $G$ takes $O(d C_6 n) = \poly(d,\log n) n$ time. Preprocessing for the function $f_i$ takes $\poly(d,\log n) n$ time by the preprocessing guarantee of Theorem \ref{thm:ss11}. Since there are $C_5 \le \poly(d,\log n)$ different $i$s, the total preprocessing time is $\poly(d,\log n) n$, as desired.

Next, we reason about query time. This follows immediately from the query time bound of Theorem \ref{thm:ss11}, along the the fact that $C_5 \le \poly(d,\log n)$.

\textbf{Upper bound.} Let $p_e^{\text{upper}} = C_7/n$ for each edge $e\in E(G[S])$, where $C_7 = 2C_{3a}C_{3b}^2$. We apply Theorem \ref{thm:oversampling} to argue that $H_i[S]$ is a sparsifier for $G[S]$ for each $i$. By choice of $C_7$ and the \emph{Size} condition on $|S|$, $p_{u,v}^{\text{upper}}\ge \frac{2}{\Phi_{G[S]}^2 |S|}$. By Lemma \ref{lem:conductance_cut} and the \emph{Expander} condition on $G[S]$, $\frac{2}{\Phi_{G[S]}^2 |S|}\ge \Reff_{G[S]}(u,v)$ for each pair $u,v\in S$. Therefore, the second condition of Theorem \ref{thm:oversampling} is satisfied by the probabilities $p_e$. Let $q = 10000C (n^2/(C_{3a}C_{3c})) \log(n^2/(C_{3a}C_{3c})) (C_7/n)$, where $C$ is the constant in Theorem \ref{thm:oversampling}. This value of $q$ satisfies the first condition of Theorem \ref{thm:oversampling}.

To apply Theorem \ref{thm:oversampling}, we just need to show that $H_i[S]$ has at least $q$ edges with probability at least $1 - 1/\poly(n)$. First, note that for uniformly random $u,v\in X$

\begin{align*}
    \Pr_{u,v\in X}[\{u,v\}\in E(G[S])] &= \Pr[\{u,v\}\in E(G[S]), u,v\in S]\\
    &= \Pr[\{u,v\}\in E(G[S]) | u,v\in S] \Pr[u\in S]\Pr[v\in S]\\
    &\ge \frac{1}{2C_{3c}} \frac{1}{C_{3a}^2}\\
\end{align*}
by the \emph{Degree} and \emph{Size} conditions on $S$. Since at least $C_6n$ pairs in $X\times X$ are chosen, $\textbf{E}[|E(H_i[S])|]\ge C_6 n (\frac{1}{2C_{3c}})(\frac{1}{C_{3a}^2})\ge 2q$. By Chernoff bounds, this means that $|E(H_i[S])|\ge q$ with probability at least $1 - 1/\poly(n)$. Therefore, Theorem \ref{thm:oversampling} applies and shows that $H_i[S]$ with edge weights $|E(G[S])|/|E(H_i[S])|$ is a $(1/6)$-sparsifier for $G[S]$ with probability at least $1 - 1/\poly(n)$. By Chernoff bounds, $|E(G[S])|/|E(H_i[S])| \ge |E(G)|/(2|E(H_i)|)$, so $\Reff_{H_i[S]}(u,v)\le \Reff_{G[S]}(u,v) (1 + 1/6)2 \le C_4 \Reff_{G[S]}(u,v)$ for all $u,v\in S$ by the sparsification accuracy guarantee. Since this holds for each $i$, the maximum over all $i$ satisfies the guarantee as well, as desired.

\textbf{Lower bound.} Let $p_e^{\text{lower}} = 1/|E(G)|$ for each edge $e$. Note that $t = 1$, so with $q = |E(H_i)|$, all edges in the sampled graph should have weight $tw_e/(p_eq) = |E(G)|/|E(H_i)|$. Therefore, by Lemma \ref{lem:sparse-lower-tail}, for a pair $u,v\in X$

$$\Pr[\Reff_{H_i}(u,v) \le 4\Reff_G(u,v)/5]\le \frac{4}{5}$$
for each $i$. Since the $H_i$s are chosen independently and identically,

$$\Pr[\max_i \Reff_{H_i}(u,v) \le 4\Reff_G(u,v)/5]\le \left(\frac{4}{5}\right)^{C_5} \le \frac{1}{n^{100}}$$
Union bounding over all pairs shows that $$\ReffQuery(u,v) > (6/7)(4/5)\Reff_G(u,v) > \Reff_G(u,v)/2$$ for all $u,v\in X$ with probability at least $1 - 1/n^{98} = 1 - 1/\poly(n)$, as desired.
\end{proof}

We now describe the algorithm $\LowDiamSet$. This algorithm simply picks random vertices $v$ and queries the effective resistance data structure to check that $v$ is in a set with the desired properties:

\begin{algorithm}[!h]
\begin{algorithmic}[1]
\Procedure{\LowDiamSet}{$X$}

    \State $\ReffPreproc(X)$

    \While{true}
    
        \State Pick a uniformly random vertex $v$ from $X$
        
        \State $Q_v\gets \{u\in X: \ReffQuery(u,v) \le C_2/(2n)\}$
        
        \State Return $Q_v$ if $|Q_v| \ge n/C_1$
    
    \EndWhile

\EndProcedure
\end{algorithmic}
\end{algorithm}

We now prove that this algorithm suffices:

\begin{proof}[Proof of Proposition \ref{prop:efficient-low-eff-res}]
\textbf{Size and Low effective resistance diameter.} Follows immediately from the return condition and the fact that for every $u\in Q_v$ for the returned $Q_v$,

$$\Reff_G(u,v)\le 2\ReffQuery(u,v)\le C_2/n$$
by the \emph{Lower bound} guarantee of Proposition \ref{prop:weak-sampling}.

\textbf{Runtime.} We start by showing that the while loop terminates after $\poly(d,\log n)$ iterations with probability at least $1 - 1/\poly(n)$. By Proposition \ref{prop:inner-product-dep}, $G$ is $(d+1)$-dependent. By Proposition \ref{prop:k-dep-dense}, $G$ has at least $n^2/(8d^2)$ edges. By Proposition \ref{prop:dense-has-expander}, $G$ has a set of vertices $S$ for which $|S|\ge n/(320d^2)$, $\Phi_{G[S]} \ge 1/(800d^2\log n)$, and for which the minimum degree of $G[S]$ is at least $n/(80000d^2)$. By the first condition on $S$,

$$\Pr[v\in S] \ge 1/(320d^2)$$
so the algorithm picks a $v$ in $S$ with probability at least $1 - 1/\poly(n)$ after at most $32000d^2\log n\le \poly(d,\log n)$ iterations. By Lemma \ref{lem:conductance_cut} applied to $S$,

$$\Reff_G(x,y)\le \Reff_{G[S]}(x,y)\le \left(\frac{1}{d_S(x)} + \frac{1}{d_S(y)}\right)\frac{1}{\Phi_{G[S]}^2}$$
for any $x,y\in S$, where $d_S(w)$ denotes the degree of the vertex $w$ in $G[S]$. By the third property of $S$, $d_S(x)\ge n/(80000d^2)$ and $d_S(y)\ge n/(80000d^2)$. By this and the second property of $S$,

$$\Reff_G(x,y)\le 102400000000 d^6(\log^2 n)/n\le C_2/(2C_4 n)$$
for any $x,y\in S$. $S$ satisfies the conditions required of Proposition \ref{prop:weak-sampling} by choice of the values $C_{3a}$, $C_{3b}$, and $C_{3c}$. Therefore, by the \emph{Upper bound} guarantee of Proposition \ref{prop:weak-sampling},

$$\ReffQuery(u,v)\le C_2/(2n)$$
for every $u\in S$ if $v\in S$. Since $|S|\ge n/C_1$ by choice of $C_1$, the return statement returns $Q_v$ with probability at least $1 - 1/\poly(n)$ when $v\in S$. Therefore, the algorithm returns a set $Q_v$ with probability at least $1 - 1/\poly(n)$ after at most $\poly(d,\log n)$ iterations.

Each iteration consists of $O(n)$ $\ReffQuery$ calls and $O(n)$ additional work. Therefore, the total work done by the while loop is $\poly(d,\log n) n$ with probability at least $1 - 1/\poly(n)$. $\ReffPreproc(X)$ takes $\poly(d,\log n) n$ by Proposition \ref{prop:weak-sampling}. Thus, the total runtime is $\poly(d,\log n) n$, as desired.
\end{proof}

\subsection{Using low-effective-resistance clusters to sparsify the unweighted IP graph}

\Aaron{Edit this section to reflect changes to previous subsection. For example, we need to add a proposition about unweighted inner product graph sparsification, as that is what we are proving here.}

In this section, we prove the following result:

\begin{proposition}[Unweighted inner product sparsification]\label{prop:unweighted-ip-sparsify}
There is a $\poly(d,\log n) n/\epsilon^2$-time algorithm for constructing an $\epsilon$-sparsifier with $O(n/\epsilon^2)$ for the unweighted inner product graph of a set of $n$ points $X\subseteq \mathbb{R}^d$.
\end{proposition}

To sparsify an unweighted inner product graph, it suffices to apply Proposition \ref{prop:efficient-low-eff-res} repeatedly to partition the graph into $\poly(d,\log n)$ clusters, each with low effective resistance diameter. We can use this structure to get a good bound on the leverage scores of edges between clusters:

\begin{proposition}[bound leverage scores of edges between clusters]\label{prop:lev-score-bound}
For a $w$-weighted graph $G$ with vertex set $X$ and $n = |X|$, let $S_1,S_2\subseteq X$ be two sets of vertices, let $R_1 = \max_{u,v\in S_1} \Reff_G(u,v)$, and let $R_2 = \max_{u,v\in S_2} \Reff_G(u,v)$. Then, for any $u\in S_1$ and $v\in S_2$,
\begin{align*}
\Reff_G(u,v)\le 3R_1 + 3R_2 + \frac{3}{\sum_{x\in S_1,y\in S_2} w_{xy}} .
\end{align*}
where $w_{xx} = \infty$ for all $x\in X$
\end{proposition}

\begin{proof}
Let $\chi\in \mathbb{R}^n$ be the vector with $\chi_u = 1$, $\chi_v = -1$, and $\chi_x = 0$ for all $x\in X$ with $x\ne u$ and $x\ne v$. For each vertex $x\in S_1$, let $s_x = \sum_{y\in S_2} w_{xy}$. For each vertex $y\in S_2$, let $s_y = \sum_{x\in S_1} w_{xy}$. Let $\tau = \sum_{x\in S_1} s_x = \sum_{y\in S_2} s_y = \sum_{x\in S_1,y\in S_2} w_{xy}$. Write $\chi$ as a sum of three vectors $d^{(1)},d^{(12)},d^{(2)}\in \mathbb{R}^n$ as follows:

$$d_x^{(1)} = \begin{cases} 
      1 - \frac{s_x}{\tau} & x = u \\
      -\frac{s_x}{\tau} & x\in S_1\setminus \{u\} \\
      0 & \text{ otherwise} \\
   \end{cases}$$
$$d_x^{(2)} = \begin{cases} 
      \frac{s_x}{\tau} - 1 & x = v \\
      \frac{s_x}{\tau} & x\in S_2\setminus \{v\} \\
      0 & \text{ otherwise} \\
   \end{cases}$$
$$d_x^{(12)} = \begin{cases}
      \frac{s_x}{\tau} & x\in S_1\\
      -\frac{s_x}{\tau} & x\in S_2\\
      0 & \text{ otherwise}\\
    \end{cases}$$
Notice that $d^{(1)} + d^{(2)} + d^{(12)} = \chi$. Furthermore, notice that $d^{(1)} = \sum_{x\in S_1} p_x\chi^{(ux)}$ and $d^{(2)} = \sum_{y\in S_2} q_y\chi^{(yv)}$, where $\chi^{(ab)}$ is the signed indicator vector of the edge from $a$ to $b$ and $\sum_{x\in S_1} p_x = 1$, $\sum_{y\in S_2} q_y = 1$, and $p_x\ge 0$ and $q_y\ge 0$ for all $x\in S_1,y\in S_2$. The function $f(d) = d^{\top} L^{\dag} d$ is convex, so by Jensen's Inequality,
\begin{align*}
(d^{(1)})^{\top} L_G^{\dag} (d^{(1)}) &\le \sum_{x\in S_1} p_x (\chi^{(ux)})^{\top} L_G^{\dag} \chi^{(ux)}\\
&\le \sum_{x\in S_1} p_x R_1\\
&\le R_1
\end{align*}
and $(d^{(2)})^{\top} L_G^{\dag} (d^{(2)}) \le R_2$. Let $f\in \mathbb{R}^{|E(G)|}$ be the vector with $f_{xy} = \frac{w_{xy}}{\tau}$ for all $x\in S_1,y\in S_2$ and let $f_e = 0$ for all other $e\in E(G)$. By definition of the $s_u$s, $f$ is a feasible flow for the electrical flow optimization problem for the demand vector $d^{(12)}$. Therefore,
\begin{align*}
(d^{(12)})^{\top} L_G^{\dag} d^{(12)} &\le \sum_{x\in S_1,y\in S_2} \frac{f_{xy}^2}{w_{xy}}\\
&= \sum_{x\in S_1,y\in S_2} \frac{w_{xy}}{\tau^2}\\
&= \frac{1}{\tau}
\end{align*}
so we can upper bound $\Reff_G(u,v)$ in the following way
\begin{align*}
\Reff_G(u,v) &= (d^{(1)} + d^{(2)} + d^{(12)})^{\top} L_G^{\dag} (d^{(1)} + d^{(2)} + d^{(12)})\\
&\le 3 (d^{(1)})^{\top} L_G^{\dag} d^{(1)} + 3 (d^{(2)})^{\top} L_G^{\dag} d^{(2)} + 3 (d^{(12)})^{\top} L_G^{\dag} d^{(12)}\\
&\le 3 R_1 + 3R_2 + 3/\tau
\end{align*}
as desired.
\end{proof}

\subsection{Sampling data structure}

In this section, we give a data structure for efficiently sampling pairs of points $(u,v)\in \mathbb{R}^d\times \mathbb{R}^d$ with probability proportional to a constant-factor approximation of $|\langle u,v\rangle|$:

\begin{lemma}\label{lem:simple-sampling-ds}
Given a pair of sets $S_1,S_2\subseteq \mathbb{R}^d$, there is a data structure that can be constructed in $\tilde{O}(d|S_1| + |S_2|)$ time that, in $\poly(d\log(|S_1| + |S_2|))$ time per sample, independently samples pairs $u\in S_1, v\in S_2$ with probability $p_{uv}$, where

$$\frac{1}{2} \frac{|\langle u,v\rangle|}{\sum_{a\in S_1,b\in S_2} |\langle a,b\rangle|}\le p_{uv}\le 2\frac{|\langle u,v\rangle|}{\sum_{a\in S_1,b\in S_2} |\langle a,b\rangle|}$$
Furthermore, it is possible to query the probability $p_{uv}$ in $\poly(d\log(|S_1| + |S_2|))$ time.
\end{lemma}

To produce this data structure, we use the following algorithm for sketching $\ell_1$-norms:

\Aaron{Move to preliminaries and fix citation.}
\begin{theorem}[Theorem 3 in \cite{i06}]\label{thm:l1-sketch}
An efficiently computable, $\poly(\log d, 1/\epsilon)$-space linear sketch exists for the $\ell_1$ norm. That is, given a $d\in \mathbb{Z}_{\ge 1}$, $\delta\in (0,1)$, and $\epsilon\in (0,1)$, there is a matrix $C = \textsc{SketchMatrix}(d,\delta,\epsilon)\in \mathbb{R}^{\ell\times d}$ and an algorithm $\textsc{RecoverNorm}(s,d,\delta,\epsilon)$ with the following properties:

\begin{enumerate}
    \item (Approximation) For any vector $v\in \mathbb{R}^d$, with probability at least $1 - \delta$ over the randomness of $\textsc{SketchMatrix}$, the value $r = \textsc{RecoverNorm}(Cv, d, \delta, \epsilon)$ is as follows:
    
    $$(1 - \epsilon) \|v\|_1\le r\le (1 + \epsilon)\|v\|_1$$
    \item $\ell = (c/\epsilon^2) \log(1/\delta)$ for some constant $c > 1$
    \item (Runtime) $\textsc{SketchMatrix}$ and $\textsc{RecoverNorm}$ take $\tilde{O}(\ell d)$ and $\poly(\ell)$ time respectively.
\end{enumerate}
\end{theorem}

We use this sketching algorithm to obtain the desired sampling algorithm in the following subroutine:

\begin{corollary}\label{cor:s-data-struct}
Given a set $S\subseteq \mathbb{R}^d$, an $\epsilon\in (0,1)$, and a $\delta\in (0,1)$, there exists a data structure which, when given a query point $u\in \mathbb{R}^d$, returns a $(1 \pm \epsilon)$-multiplicative approximation to $\sum_{v\in S} |\langle u,v\rangle|$ with probability at least $1 - \delta$. This data structure can be computed in $O(\ell d |S|)$ preprocessing time and takes $\poly(\ell d)$ time per query, where $\ell = O( \epsilon^{-2} \log(1/\delta))$.
\end{corollary}

\begin{proof}
Let $n = |S|$ and $C = \textsc{SketchMatrix}(n,\delta,\epsilon)$. We will show that the following algorithm returns the desired estimate with probability at least $1 - \delta$:

\begin{enumerate}
    \item Preprocessing:
    \begin{enumerate}
        \item Index the rows of $C$ by integers between 1 and $\ell$. Index columns of $C$ by points $v\in S$.
        \item Compute the vector $x^{(i)} = \sum_{v\in S} C_{iv} v$ for each $i\in \ell$, where $\ell = (c/\epsilon^2)\log(1/\delta)$ for the constant $c$ in Theorem \ref{thm:l1-sketch}.
    \end{enumerate}
    \item Given a query point $u\in \mathbb{R}^d$,
    \begin{enumerate}
        \item Let $y\in \mathbb{R}^{\ell}$ be a vector with $y_i = \langle u, x^{(i)}\rangle$ for each $i\in [\ell]$.
        \item Return $\textsc{RecoverNorm}(y,n,\delta,\epsilon)$
    \end{enumerate}
\end{enumerate}

\textbf{Approximation}. Let $w \in \mathbb{R}^n$ be the vector with $w_v = \langle u,v\rangle$ for each $v\in S$. It suffices to show that the number that a query returns is a $(1\pm\epsilon)$-approximation to $\|w\|_1$. By definition, $y_i = \sum_{v\in S} C_{iv} w_v$ for all $i\in [\ell]$ and $v\in S$, so $y = C w$. Therefore, by the \emph{Approximation} guarantee of Theorem \ref{thm:l1-sketch}, $\textsc{RecoverNorm}(y,n,\delta,\epsilon)$ returns a $(1\pm\epsilon)$-approximation to $\|w\|_1$ with probability at least $1 - \delta$, as desired.

\textbf{Preprocessing time}. Computing the matrix $C$ takes $\tilde{O}(n\ell)$ time by the \emph{Runtime} guarantee of Theorem \ref{thm:l1-sketch}. Computing the vectors $x^{(i)}$ for all $i\in [\ell]$ takes $O(d\ell n)$ time. This is all of the preprocessing steps, so the total runtime is $\tilde{O}(d\ell n)$, as desired.

\textbf{Query time}. Each query consists of $\ell$ inner products of $d$ dimensioal vectors and one call to $\textsc{RecoverNorm}$, for a total of $O(\ell d + \poly(\ell))$ work, as desired.
\end{proof}

We use this corollary to obtain a sampling algorithm as follows, where $n = |S_1| + |S_2|$:

\begin{enumerate}
    \item Preprocessing:
    \begin{enumerate}
        \item Use the Corollary \ref{cor:s-data-struct} data structure to $(1 \pm 1/(100\log n))$-approximate $\sum_{v\in S_2} |\langle u,v\rangle|$ for each $u\in S_1$. Let $t_u$ be this estimate for each $u\in S_1$. (one preprocess for $S\gets S_2$, $|S_1|$ queries).
        \item Form a balanced binary tree $\mathcal T$ of subsets of $S_2$, with $S_2$ at the root, the elements of $S_2$ at the leaves, and the property that for every parent-child pair $(P,C)$, $|C|\le 2|P|/3$.
        \item For every node $S$ in the binary tree, construct a $(1 \pm 1/(100\log n))$-approximate data structure for $S$.
    \end{enumerate}
    \item Sampling query:
    \begin{enumerate}
        \item Sample a point $u\in S_1$ with probability $t_u/(\sum_{a\in S_1} t_a)$.
        \item Initialize $S\gets S_2$. While $|S| > 1$,
        \begin{enumerate}
            \item Let $P_1$ and $P_2$ denote the two children of $S$ in $\mathcal T$
            \item Let $s_1$ and $s_2$ be the $(1 \pm 1/(100\log n))$-approximations to $\sum_{v\in P_1} |\langle u,v\rangle|$ and $\sum_{v\in P_2} |\langle u,v\rangle|$ respectively obtained from the data structure for $S$ computed during preprocessing.
            \item Reset $S$ to $P_1$ with probability $s_1/(s_1 + s_2)$; otherwise reset $S$ to $P_2$.
        \end{enumerate}
        \item Return the single element in $S$.
    \end{enumerate}
    \item $p_{uv}$ query:
    \begin{enumerate}
        \item Return the product of the $O(\log n)$ probabilities attached to ancestor nodes of the node $\{v\}$ in $\mathcal T$ for $u$, obtained from the preprocessing step (as during the sampling query)
    \end{enumerate}
\end{enumerate}

$s_i/(s_1 + s_2)$ is a $(1 \pm 1/(100\log n))^2$-approximation to $\Pr[v\in P_i | v\in S]$ for each $i\in \{1,2\}$. A $v\in S_2$ is sampled with probability proportional to the product of these conditional probabilities for the ancestors, for which $p_{uv}$ is a $(1 + 1/(100\log n))^{2\log n + 1}\le 2$-approximation. The total preprocessing time is proportional to the size of all sets in the tree, which is at most $O(|S_2|\log |S_2|)$. The total query time is also $\text{polylog}(|S_1| + |S_2|)\poly(d)$ due to the logarithmic depth of the query binary tree.

We now expand on this intuition to prove Lemma \ref{lem:simple-sampling-ds}:

\begin{proof}[Proof of Lemma \ref{lem:simple-sampling-ds}]
We show that the algorithm given just before this proof satisfies this lemma:

\textbf{Probability guarantee}. Consider a pair $u\in S_1$, $v\in S_2$. Let $A_0 = S_2, A_1,\hdots A_{k-1}, A_k = \{v\}$ denote the sequence of ancestor sets of the singleton set $\{v\}$ in $\mathcal{T}$. For each $i\in [k]$, let $B_i$ be the child of $A_{i-1}$ besides $A_i$ (unique because $\mathcal{T}$ is binary). For a node $X$ of $\mathcal{T}$, let $s_X$ be the $(1\pm 1/(100\log n))$-approximation to $\sum_{v\in X} |\langle u,v\rangle|$ used by the algorithm. The sampling probability $p_{uv}$ is the following product of probabilities:

$$p_{uv} = \frac{t_u}{\sum_{a\in S_1} t_a} \prod_{i=1}^k \frac{s_{A_i}}{s_{A_i} + s_{B_i}}$$
$s_{A_i} + s_{B_i}$ is a $(1\pm 1/(100\log n))$-approximation to $\sum_{v\in A_{i-1}} |\langle u,v\rangle|$ by the approximation guarantee of Corollary \ref{cor:s-data-struct}. $s_{A_i}$ is a $(1\pm 1/(100\log n))$-approximation to $\sum_{v\in A_i} |\langle u,v\rangle|$ by the approximation guarantee of Corollary \ref{cor:s-data-struct}. By these guarantees and the approximation guarantees for the $t_a$s for $a\in S_1$, $p_{uv}$ is a $(1 \pm 1/(100\log n))^{2k+2}\le (1\pm 1/2)$-approximation to

$$\frac{\sum_{b\in S_2} |\langle u,b\rangle|}{\sum_{a\in S_1,b\in S_2} |\langle a,b\rangle|}\prod_{i=1}^k \frac{\sum_{b\in A_i} |\langle u,b\rangle|}{\sum_{b\in A_{i-1}} |\langle u,b\rangle|} = \frac{|\langle u,v\rangle|}{\sum_{a\in S_1,b\in S_2} |\langle a,b\rangle|}$$
as desired, since $k\le \log n$.

\textbf{Preprocessing time}. Let $\delta = 1/n^{1000}$ and $\ell = (c/\epsilon^2)\log(1/\delta)$, where $c$ is the constant from Theorem \ref{thm:l1-sketch}. Approximating all $t_u$s takes $\tilde{O}(\ell d |S_2|) + |S_1|\poly(\ell d) = n \poly(\ell d)$ time by Corollary \ref{cor:s-data-struct}. Preprocessing the data structures for each node $S\in \mathcal{T}$ takes $\sum_{S\in \mathcal{T}} O(\ell d |S|)$ time in total. Each member of $S_2$ is in at most $O(\log n)$ sets in $\mathcal{T}$, since $\mathcal{T}$ is a balanced binary tree. Therefore, $\sum_{S\in \mathcal{T}} O(\ell d |S|)\le \tilde{O}(n\ell d)$, so the total preprocessing time is $\tilde{O}(n\ell d)$, as desired.

\textbf{Query time}. $O(\log n)$ of the data structures attached to sets in $\mathcal{T}$ are queried per sample query, for a total of $O(\log n)\poly(\ell d)$ time by Corollary \ref{cor:s-data-struct}, as desired.

\end{proof}

We use this data structure via a simple reduction to implement the following data structure, which suffices for our applications:

\begin{proposition}\label{prop:main-sampling-ds}
Given a family $\mathcal G$ of sets in $\mathbb{R}^d$, a collection of positive real numbers $\{\gamma_S\}_{S\in \mathcal G}$, and a set $S_2\subseteq \mathbb{R}^d$, there is a data structure that can be constructed in $\tilde{O}(d(|S_2| + \sum_{S\in \mathcal G} |S|))$ time that, in $\poly(d\log(|S_2| + \sum_{S\in \mathcal G} |S|))$ time per sample, independently samples pairs $u\in S$ for some $S\in \mathcal G$, $v\in S_2$ with probability $p_{uv}$, where

$$\frac{\gamma_S |\langle u,v\rangle|}{(\sum_{A\in \mathcal G} \gamma_A) \sum_{a\in A,b\in S_2} |\langle a,b\rangle|}$$
is a 2-approximation to $p_{uv}$. Furthermore, it is possible to query the number $p_{uv}$ in time $\poly(d\log(|S_2| + \sum_{S\in \mathcal G} |S|))$ time.
\end{proposition}

\begin{proof}
Define a new set $S_1\subseteq \mathbb{R}^{d+\log n}$ as follows, where $n = |S_2| + \sum_{S\in \mathcal G} |S|$:

$$S_1 = \cup_{S\in \mathcal G} \{f_S(u) \forall u\in S\}$$
where the function $f_S:\mathbb{R}^d\rightarrow \mathbb{R}^d$ is defined as follows for any set $S\in \mathcal G$:

$$f_S(u) = \frac{\gamma_S (u,\text{id}(u))}{\sum_{a\in S,b\in S_2} |\langle a,b\rangle|}$$
where $\text{id}: \cup_{S\in \mathcal G} S\rightarrow \{0,1\}^{\log n}$ is a function that outputs a unique ID for each element of $\cup_{S\in \mathcal G} S$. Define $f(u) = f_S(u)$ for the unique $S$ containing $u$ (Without loss of generality assume that exactly one $S$ contains u \Aaron{clarify?}). Let $S_2' = \{(x,0^{\log n}) \forall x\in S_2\}$. Construct the data structure $\mathcal D$ from Lemma \ref{lem:simple-sampling-ds} on the pair of sets $S_1,S_2'$. Now, sample a pair of $d$-dimensional vectors as follows:

\begin{enumerate}
    \item Sample:
    \begin{enumerate}
        \item Sample a pair $(x,(y,0^{\log n}))$ from $\mathcal D$, where $y\in S_2$ and $x\in \mathbb{R}^{d+\log n}$.
        \item Since the function $\text{id}$ outputs values that are not proportional to one other, the function $f$ is injective.
        \item Return the pair $(f^{-1}(x),y)$.
        \item (For the proof, let $w = f^{-1}(x)$ and let $S$ be the unique set for which $w\in S$)
    \end{enumerate}
\end{enumerate}

This data structure has preprocessing time $\tilde{O}(d(|S_2'| + |S_1|)) = \tilde{O}(d(|S_2| + \sum_{S\in \mathcal G} |S|))$ and sample query time $\poly(d\log n)$ by Lemma \ref{lem:simple-sampling-ds}, as desired. Therefore, we just need to show that it samples a pair $(w,y)$ with the desired probability. By the probability guarantee for Lemma \ref{lem:simple-sampling-ds} and the injectivity of the mapping $f$ combined over all $S\in \mathcal G$, $p_{wy}$ is 2-approximated by

\begin{align*}
\frac{|\langle x,(y,0^{\log n})\rangle|}{\sum_{a\in S_1,b\in S_2'} |\langle a,b\rangle|} &= \frac{\gamma_S |\langle w,y\rangle|}{\sum_{A\in \mathcal G} \sum_{p\in A,q\in S_2} |\langle p,q\rangle|}\\
\end{align*}
as desired.
\end{proof}

\subsection{Weighted IP graph sparsification}

In this section, we use the tools developed in the previous sections to sparsify weighted inner product graphs. To modularize the exposition, we define a partition of the edge set of a weighted inner product graph:

\begin{definition}
For a $w$-weighted graph $G$ and three functions on pairs of vertex sets $\zeta,\kappa,\delta$, a collection of vertex set family-vertex set pairs $\mathcal F$ is called a $(\zeta,\kappa,\delta)$-\emph{cover} for $G$ iff the following property holds:

\begin{enumerate}
    \item (Coverage) For any $e = \{u,v\}\in E(G)$, there exists a pair $(\mathcal G,S_1)\in \mathcal F$ and an $S_0\in \mathcal G$ for which $u\in S_0,v\in S_1$ or $u\in S_1,v\in S_0$, and
    $$\Reff_G(u,v)\le \frac{\delta(S_0,S_1)}{w_{uv}} + \frac{\kappa(S_0,S_1)}{\max_{x\in S_0,y\in S_1} w_{xy}} + \frac{\zeta(S_0,S_1)}{\sum_{x\in S_0,y\in S_1} w_{xy}}$$
\end{enumerate}

A $(\zeta,\kappa,\delta)$-cover is said to be $s$-\emph{sparse} if $\sum_{(\mathcal G,S_1)\in \mathcal{F}} \sum_{S_0\in \mathcal G} ((\delta(S_0,S_1) + \kappa(S_0,S_1))|S_0||S_1| + \zeta(S_0,S_1)) \le s$. A $(\zeta,\kappa)$-cover is said to be $w$-\emph{efficient} if $\sum_{(\mathcal G,S_1)\in \mathcal F} \left(|S_1| + \sum_{S_0\in \mathcal G} |S_0|\right)\le w$. When $\delta = 0$, we simplify notation to refer to $(\zeta,\kappa)$-covers instead.
\end{definition}

Given a $(\zeta,\kappa)$-cover for a weighted or unweighted inner product graph, one can sparsify it using Theorem \ref{thm:oversampling} and the sampling data structure from Proposition \ref{prop:main-sampling-ds}:

\begin{proposition}\label{prop:sparsify-given-cover}
Given a set $X\subseteq \mathbb{R}^d$ with $n = |X|$ and an $s$-sparse $w$-efficient $(\zeta,\kappa,\delta)$-cover $\mathcal F$ for the weighted inner product graph $G$ on $X$, and $\epsilon,\delta \in (0,1)$, there is an 
\begin{align*}
\poly(d,\log n,\log s,\log w,\log(1/\delta)) (s + w + n)/\epsilon^4
\end{align*}
time algorithm for constructing an $(1\pm\epsilon)$-sparsifier for $G$ with $O(n\log n/\epsilon^2)$ edges with probability at least $1 - \delta$.
\end{proposition}

\begin{algorithm}
\begin{algorithmic}[1]\caption{}
\Procedure{\OversamplingWithCover}{$X,\mathcal F,\epsilon,\delta$}
    \State (for analysis only: define $r_{uv}$ for each pair $u,v\in X$ as in proof)
    \State \textbf{Construct the Proposition \ref{prop:main-sampling-ds} data structure $\mathcal D_{(\mathcal G,S_1)}$ with $\gamma_S = \zeta(S,S_1)$ for each $S\in \mathcal G$}
    \State $t \leftarrow \sum_{u,v\in X} r_{uv}$
    \State $q \leftarrow C \cdot \epsilon^{-2} \cdot t \log t \cdot \log(1/\delta)$
    \State Initialize $H$ to be an empty graph
    \For{$i = 1 \to q$}
        \State Sample one $e = \{u,v\} \in X\times X$ with probability $r_e/t$ \textbf{by sampling $\{u,v\}$ uniformly or from some data structure $\mathcal D_{(\mathcal G,S_1)}$ (see proof for details)}
        \State Add that edge with weight $w_e t / (r_e q)$ to graph $H$ (note: $r_e$ can be computed in $\poly(d,\log w)$ by the $p_{uv}$ query time in Prop \ref{prop:main-sampling-ds})
    \EndFor
    \State \Return Spielman-Srivastava \cite{ss11} applied to $H$
\EndProcedure
\end{algorithmic}
\end{algorithm}

\begin{proof}

\textbf{Filling in algorithm details (the bolded parts)}. We start by filling in the details in the algorithm

\noindent $\OversamplingWithCover$. First, we define $r_{uv}$ for each pair of distinct $u,v\in X$. $\{u,v\}$ is a weighted edge in $G$ with weight $w_{uv}$. Define

$$r_{uv} = 2\sum_{(\mathcal G,S_1)\in \mathcal F}\sum_{S_0\in \mathcal G: u\in S_0,v\in S_1 \text{ or } v\in S_0,u\in S_1} \left(\delta(S_0,S_1) + \kappa(S_0,S_1) + \left(\sum_{A\in \mathcal G} \zeta(A,S_1)\right) p_{uv}^{(\mathcal G,S_1)}\right)$$
where $p_{uv}^{(\mathcal G,S_1)}$ is the probability $p_{uv}$ defined for the data structure $\mathcal D^{(\mathcal G,S_1)}$ in Proposition \ref{prop:main-sampling-ds}. Next, we fully describe how to sample pairs $\{u,v\}$ with probability proportional to $r_{uv}$. Notice that $t$ can be computed in $O(w)$ time because

\begin{align*}
t &= 2\sum_{u,v\in X} r_{uv}\\
&= 2\sum_{(\mathcal G,S_1)\in \mathcal F} \sum_{S_0\in \mathcal G} \sum_{u\in S_0,v\in S_1} \left(\delta(S_0,S_1) + \kappa(S_0,S_1) + \left(\sum_{A\in \mathcal G} \zeta(A,S_1)\right) p_{uv}^{(\mathcal G,S_1)}\right)\\
&= 2\sum_{(\mathcal G,S_1)\in \mathcal F} \left(\left(\sum_{A\in \mathcal G} \zeta(A,S_1)\right) + \sum_{S_0\in \mathcal G} |S_0||S_1|(\delta(S_0,S_1) + \kappa(S_0,S_1))\right)\\
\end{align*}
can be computed in $O(w)$ time. Sample a pair $\{u,v\}$ with probability equal to $r_{uv}/t$ as follows:

\begin{enumerate}
    \item Sample a pair:
    \begin{enumerate}
        \item Sample a Bernoulli $b\sim \text{Bernoulli}\left(\frac{1}{t}\sum_{(\mathcal G,S_1)\in \mathcal F} \left(\sum_{A\in \mathcal G} \zeta(A,S_1)\right)\right)$.
        \item If $b = 1$
        \begin{enumerate}
            \item Sample a pair $(\mathcal G,S_1)\in \mathcal F$ with probability proportional to $\sum_{A\in \mathcal G} \zeta(A,S_1)$.
            \item Sample the pair $(u,v)$ using the data structure $\mathcal D^{(\mathcal G,S_1)}$.
        \end{enumerate}
        \item Else
        \begin{enumerate}
            \item Sample a pair $(\mathcal G,S_1)\in \mathcal F$ with probability proportional to $\sum_{S_0\in \mathcal G} |S_0||S_1|(\delta(S_0,S_1) + \kappa(S_0,S_1))$.
            \item Sample an $S_0\in \mathcal G$ with probability proportional to $|S_0|(\delta(S_0,S_1) + \kappa(S_0,S_1))$.
            \item Sample $(u,v)\in S_0\times S_1$ uniformly.
        \end{enumerate}
    \end{enumerate}
\end{enumerate}

All sums in the above sampling procedure can be precomputed in $\poly(d) w$ time. After doing this precomputation, each sample from the above procedure takes $\poly(d,\log n,\log w)$ time by to Proposition \ref{prop:main-sampling-ds} for the last step in the if statement and uniform sampling from $[0,1]$ with intervals otherwise.

\textbf{Sparsifier correctness}. By the \emph{Coverage} guarantee of $\mathcal F$ and the approximation guarantee for the $p_{uv}$s in Proposition \ref{prop:main-sampling-ds}, $w_{uv}\Reff_G(u,v)\le r_{uv}$ for all $u,v\in X$. Therefore, Theorem \ref{thm:oversampling} applies and shows that the graph $H$ returned is a $(1\pm\epsilon)$-sparsifier for $G$ with probability at least $1 - \delta$. Spielman-Srivastava only worsens the approximation guarantee by a $(1+\epsilon)$ factor, as desired.

\underline{Number of edges in $H$}. It suffices to bound $q$. In turn, it suffices to bound $t$. Recall from above that

\begin{align*}
t &= 2\sum_{(\mathcal G,S_1)\in \mathcal F} \left(\left(\sum_{A\in \mathcal G} \zeta(A,S_1)\right) + \sum_{S_0\in \mathcal G} |S_0||S_1|(\delta(S_0,S_1) + \kappa(S_0,S_1))\right)\\
&= 2\sum_{(\mathcal G,S_1)\in \mathcal F} \left(\sum_{S_0\in \mathcal G} \left(|S_0||S_1|(\delta(S_0,S_1) + \kappa(S_0,S_1)) + \zeta(S_0,S_1)\right)\right)\\
&\le 2s
\end{align*}
since $\mathcal F$ is $s$-sparse. Therefore, $q\le \poly(d,\log s,\log 1/\delta) s/\epsilon^2$.

\textbf{Runtime}. We start by bounding the runtime to produce $H$. Constructing the data structure $\mathcal D_{(\mathcal G,S_1)}$ takes $\poly(d,\log n,\log w) (|S_1| + \sum_{S_0\in \mathcal G} |S_0|)$ time by Proposition \ref{prop:main-sampling-ds}. Therefore, the total time to construct all data structures is at most $\poly(d,\log n,\log w) w$. Computing $t$ and $q$, as discussed above, takes $O(w)$ time. $q\le \poly(d,\log s,\log 1/\delta) s/\epsilon^2$ as discussed above and each iteration of the for loop takes $\poly(d,\log n,\log w)$ time by the query complexity bounds of Proposition \ref{prop:main-sampling-ds}. Therefore, the total time required to produce $H$ is at most $\poly(d,\log n,\log s,\log w,\log 1/\delta) (s + w)/\epsilon^2$. Running Spielman-Srivastava requires an additional $\poly(\log s,\log n) (s/\epsilon^2 + n)/\epsilon^2$ time, for a total of $\poly(d,\log n,\log s,\log w,\log 1/\delta) (s + w)/\epsilon^4$ time, as desired.
\end{proof}

Therefore, to sparsify weighted inner product graphs, it suffices to construct an $\poly(d,\log n) n$-sparse, $\poly(d,\log n) n$-efficient $(\zeta,\kappa)$-cover. We break up this construction into a sequence of steps:

\subsubsection{\texorpdfstring{$(\zeta,\kappa)$}{}-cover for unweighted IP graphs}

We start by constructing covers for unweighted inner product graphs. The algorithm repeatedly peels off sets constructed using $\LowDiamSet$ and returns all pairs of such sets. The sparsity of the cover is bounded due to Proposition \ref{prop:lev-score-bound}. The efficiency of the cover is bounded thanks to a $\poly(d,\log n)$ bound on the number of while loop iterations, which in turn follows from the \emph{Size} guarantee of Proposition \ref{prop:weak-sampling}.

\begin{proposition}[Cover for unweighted graphs]\label{prop:unweighted-cover}
Given a set $X\subseteq \mathbb{R}^d$ with $|X| = n$, there is an $\poly(d,\log n)$-time algorithm $\UnweightedCover(X)$ that, with probability at least $1 - 1/\poly(n)$, produces an $\poly(d,\log n) n$-sparse $\poly(d,\log n) n$-efficient $(\zeta,\kappa)$-cover for the unweighted inner product graph $G$ on $X$.
\end{proposition}

\begin{algorithm}[!h]\caption{}
\begin{algorithmic}[1]
\Procedure{\UnweightedCover}{$X$}

    \State \textbf{Input}: $X\subseteq \mathbb{R}^d$
    
    \State \textbf{Output}: An sparse, efficient $(\zeta,\kappa)$ cover for the unweighted inner product graph $G$ on $X$
    
    \State $\mathcal U \gets \emptyset$
    
    \State $Y\gets X$
    
    \While{$Y\ne \emptyset$} \Comment{Finding expanders}
    
        \State Let $Q\gets \LowDiamSet(Y)$
    
        \State Add the set $Q$ to $\mathcal U$
        
        \State Remove the vertices $Q$ from $Y$
        
    \EndWhile
    
    \State \Return $\{(\mathcal U, S): \forall S\in \mathcal U\}$
    
\EndProcedure
\end{algorithmic}
\end{algorithm}

\begin{proof}
Let $\mathcal F = \UnweightedCover(X)$ and define the functions $\zeta,\kappa$ as follows: $\zeta(S_0,S_1) = 3$ and $\kappa(S_0,S_1) = C_2\left(\frac{3}{|S_0|} + \frac{3}{|S_1|}\right)$ for any pair of sets $S_0,S_1\subseteq X$. Recall that $C_2$ is defined in the statement of Proposition \ref{prop:efficient-low-eff-res}.

\textbf{Number of while loop iterations}. We start by showing that there are at most $\poly(d,\log n)$ while loop iterations with probability at least $1 - 1/\poly(n)$. By the \emph{Size} guarantee of Proposition \ref{prop:efficient-low-eff-res}, when $\LowDiamSet$ succeeds (which happens with probability $1 - 1/\poly(n)$), $Y$ decreases in size by a factor of at least $1 - 1/p(d,\log n)$ for some fixed constant degree polynomial $p$. Therefore, after $p(d,\log n) \log n = \poly(d,\log n)$ iterations, $Y$ will be empty, as desired. Therefore, $|\mathcal U|\le \poly(d,\log n)$.

\textbf{Runtime}. The runtime follows immediately from the bound on the number of while loop iterations and the runtime bound on $\LowDiamSet$ from Proposition \ref{prop:efficient-low-eff-res}.

\textbf{Coverage}. For each $S\in \mathcal U$, $\max_{u,v\in S} \Reff_G(u,v)\le \frac{C_2}{|S|}$ by the \emph{Low effective resistance diameter} guarantee of Proposition \ref{prop:efficient-low-eff-res}. Plugging this into Proposition \ref{prop:lev-score-bound} immediately shows that $\mathcal F$ is a $(\zeta,\kappa)$-cover for $G$.

\textbf{Sparsity bound}. We bound the desired quantity directly using the fact that $|\mathcal U|\le \poly(d,\log n)$:

\begin{align*}
    \sum_{(\mathcal U,S_1)\in \mathcal F} \sum_{S_0\in \mathcal U} \left(\kappa(S_0,S_1)|S_0||S_1| + \zeta(S_0,S_1)\right) &\le \poly(d,\log n)\sum_{(\mathcal U,S_1)\in \mathcal F} \sum_{S_0\in \mathcal U} (|S_0| + |S_1| + 1)\\
    &\le \poly(d,\log n) n\\
\end{align*}
as desired.
\end{proof}

\textbf{Efficiency bound}. Follows immediately from the bound on $|\mathcal U|$.

\subsubsection{\texorpdfstring{$(\zeta,\kappa)$}{}-cover for weighted IP graphs on bounded-norm vectors}

Given a covers for unweighted inner product graphs, it is easy to construct covers for weighted inner product graphs on bounded norm vectors simply by removing edge weights and producing the cover. Edge weights only differ by a factor of $O(d)$ in these two graphs, so effective resistances also differ by at most that amount. Note that the following algorithm also works for vectors with norms between $z$ and $2z$ for any real number $z$.

\begin{proposition}[Weighted bounded norm cover]\label{prop:bounded-cover}
Given a set $X\subseteq \mathbb{R}^d$ with $1\le \|u\|_2\le 2$ for all $u\in X$ and $|X| = n$, there is an $n\poly(d,\log n)$ time algorithm $\BoundedCover(X)$ that produces a $\poly(d,\log n)$-sparse, $\poly(d,\log n)$-efficient $(\zeta,\kappa)$-cover for the weighted inner product graph $G$ on $X$.
\end{proposition}

\begin{proof}
Let $G_0$ be the unweighted inner product graph on $X$. Let $G_1$ be the weighted graph $G$ with all edges that are not in $G_0$ deleted. Let $\mathcal F$ be the $(\zeta_0,\kappa_0)$-cover given by Proposition \ref{prop:unweighted-cover} for $G_0$. Let this cover be the output of $\BoundedCover(X)$. It suffices to show that $\mathcal F$ is a $(\zeta,\kappa)$-cover for $G$, where $\zeta = 8d \zeta_0$ and $\kappa = 8d \kappa_0$. By Rayleigh monotonicity,

$$\Reff_G(u,v)\le \Reff_{G_1}(u,v)$$
for all $u,v\in X$. Let $w_e$ denote the weight of the edge $e$ in $G_1$. For all edges $e$ in $G_1$, $\frac{1}{d+1}\le w_e\le 4$ by the norm condition on $X$. Therefore, for all $u,v\in X$,

$$\Reff_{G_1}(u,v)\le (d+1)\Reff_{G_0}(u,v)$$
By the \emph{Coverage} guarantee on $\mathcal F$, there exists a pair $(\mathcal G,S_1)$ and an $S_0\in \mathcal G$ for which $u\in S_0,v\in S_1$ or $v\in S_0,u\in S_1$ and

$$\Reff_{G_1}(u,v)\le (d+1) \left(\kappa_0(S_0,S_1) + \frac{\zeta_0(S_0,S_1)}{|S_0||S_1|}\right)$$
By the upper bound on the edge weights for $G_1$,

$$\Reff_{G_1}(u,v)\le 4(d+1) \left(\frac{\kappa_0(S_0,S_1)}{\max_{x\in S_0,y\in S_1} w_{xy}} + \frac{\zeta_0(S_0,S_1)}{\sum_{x\in S_0,y\in S_1}w_{xy}}\right)$$
Since $4(d+1)\le 8d$, $\mathcal F$ is a $(\zeta,\kappa)$-cover for $G$ as well, as desired.
\end{proof}

\subsubsection{\texorpdfstring{$(\zeta,\kappa)$}{}-cover for weighted IP graphs on vectors with norms in the set \texorpdfstring{$[1,2]\cup [z,2z]$}{} for any \texorpdfstring{$z > 1$}{}}

By the previous subsection, it suffices to cover the pairs $(u,v)$ for which $\|u\|_2\in [1,2]$ and $\|v\|_2\in [z,2z]$. This can be done by clustering using $\LowDiamSet$ on the $[z,2z]$-norm vectors. For each cluster $S_1$, let $\mathcal G = \{\{u\}: \forall u\in X \text{ with } \|u\|_2\in [1,2]\}$. This cover is sparse because of the fact that the clusters have low effective resistance diameter. It is efficient because of the small number of clusters.

\begin{proposition}[Two-scale cover]\label{prop:two-bounded-cover}
Given a set $X\subseteq \mathbb{R}^d$ for which $|X| = n$ and $\|u\|_2\in [1,2]\cup [z,2z]$ for all $u\in X$, there is a $\poly(d,\log n) n$-time algorithm $\TwoBoundedCover(X)$ that produces an $\poly(d,\log n) n$-sparse, $\poly(d,\log n)$-efficient $(\zeta,\kappa)$-cover for the weighted inner product graph $G$ on $X$.
\end{proposition}

\begin{algorithm}[!h]\caption{}
\begin{algorithmic}[1]
\Procedure{\TwoBoundedCover}{$X$}

    \State \textbf{Input}: $X\subseteq \mathbb{R}^d$, where $\|u\|_2\in [1,2]\cup [z,2z]$ for all $u\in X$
    
    \State \textbf{Output}: An sparse, efficient $(\zeta,\kappa)$ cover for the weighted inner product graph $G$ on $X$
    
    \State $X_{\text{low}}\gets \{u\in X: \|u\|_2\in [1,2]\}$
    
    \State $X_{\text{high}}\gets \{u\in X: \|u\|_2\in [z,2z]\}$
    
    \State $\mathcal U \gets \emptyset$
    
    \State $Y\gets X_{\text{high}}$
    
    \While{$Y\ne \emptyset$}
    
        \State Let $Q\gets \LowDiamSet(Y)$
    
        \State Add the set $Q$ to $\mathcal U$
        
        \State Remove the vertices $Q$ from $Y$
        
    \EndWhile
    
    \State \Return $\{(\{\{u\} \forall u\in X_{\text{low}}\}, S_1): \forall S_1\in \mathcal U\}$
    
    \State $ \cup \BoundedCover(X_{\text{low}}) \cup \BoundedCover(X_{\text{high}})$
    
\EndProcedure
\end{algorithmic}
\end{algorithm}

\begin{proof}
Suppose that the $\BoundedCover$s returned for $X_{\text{low}}$ and $X_{\text{high}}$ are $(\zeta_{\text{low}},\kappa_{\text{low}})$ and $(\zeta_{\text{high}},\kappa_{\text{high}})$-covers respectively. Recall the value $C_2\le \poly(d,\log n)$ from the statement of Proposition \ref{prop:efficient-low-eff-res}. Let $w_{uv} = |\langle u,v\rangle|$ denote the weight of the $u$-$v$ edge in $G$. Let $\mathcal F = \TwoBoundedCover(X)$ and define the functions $\zeta,\kappa$ as follows:

\[
  \zeta(S_0,S_1) =
  \begin{cases}
                                   \zeta_{\text{low}}(S_0,S_1) & \text{if $S_0,S_1\subseteq X_{\text{low}}$} \\
                                   \zeta_{\text{high}}(S_0,S_1) & \text{if $S_0,S_1\subseteq X_{\text{high}}$} \\
                                    3 & \text{if $S_0\subseteq X_{\text{low}}$ and $S_1\subseteq X_{\text{high}}$} \\
                                    \infty & \text{otherwise}
  \end{cases}
\]

\[
  \kappa(S_0,S_1) =
  \begin{cases}
                                   \kappa_{\text{low}}(S_0,S_1) & \text{if $S_0,S_1\subseteq X_{\text{low}}$} \\
                                   \kappa_{\text{high}}(S_0,S_1) & \text{if $S_0,S_1\subseteq X_{\text{high}}$} \\
                                    \frac{24 d C_2}{|S_1|} & \text{if $S_0\subseteq X_{\text{low}}$ and $S_1\subseteq X_{\text{high}}$} \\
                                    \infty & \text{otherwise}
  \end{cases}
\]

\textbf{Number of while loop iterations}. A $\poly(d,\log n)$-round bound follows from the \emph{Size} bound of Proposition \ref{prop:efficient-low-eff-res}. For more details, see the same part of the proof of Proposition \ref{prop:unweighted-cover}, which used the exact same algorithm for producing $\mathcal U$.

\textbf{Runtime}. Follows immediately from the runtime bounds of $\LowDiamSet$, $\BoundedCover$, and the number of while loop iterations.

\textbf{Coverage}. Consider a pair $u,v\in X$. We break the analysis up into cases:

\underline{Case 1: $u,v\in X_{\text{low}}$}. In this case, the \emph{Coverage} property of $\BoundedCover(X_{\text{low}})$ implies that the pair $(u,v)$ is covered in $\mathcal F$ by Rayleigh monotonicity (since $G_{\text{low}}$ is a subgraph of $G$, where $G_{\text{low}}$ is the weighted inner product graph for $X_{\text{low}}$).

\underline{Case 2: $u,v\in X_{\text{high}}$}. In this case, the \emph{Coverage} property of $\BoundedCover(X_{\text{high}})$ implies that the pair $(u,v)$ is covered in $\mathcal F$ by Rayleigh monotonicity.

\underline{Case 3: $u\in X_{\text{low}}$ and $v\in X_{\text{high}}$}. Since $\mathcal U$ is a partition of $X_{\text{high}}$, there is a unique pair $(\mathcal G,S_1)\in \mathcal F$ for which $\{u\}\in \mathcal G$ and $v\in S_1$. Let $H$ denote the unweighted inner product graph on $X_{\text{high}}$. By Rayleigh monotonicity, the fact that $\frac{z^2}{2d}\le\frac{z^2}{d+1}\le w_{xy}$ for all $\{x,y\}\in E(H)$, and the \emph{Low effective resistance diameter} guarantee of Proposition \ref{prop:efficient-low-eff-res},

\begin{align*}
    \Reff_G(x,y) &\le \Reff_{G_{\text{high}}}(x,y)\\
    &\le \frac{2d\Reff_H(x,y)}{z^2}\\
    &\le \frac{2dC_2}{z^2|S_1|}\\
\end{align*}
for any $x,y\in S_1$. Since $z > 1$, $w_{xy}\le 4z^2$ for all $x,y\in X$. Therefore,

$$\Reff_G(x,y) \le \frac{8dC_2}{\max_{p\in X_{\text{low}}, q\in S_1} w_{pq}}$$
for any $x,y\in S_1$. Therefore, Proposition \ref{prop:lev-score-bound} implies the desired \emph{Coverage} bound in this case.

\textbf{Sparsity bound}. We use the fact that $|\mathcal U|\le \poly(d,\log n)$ along with sparsity bounds for $\BoundedCover$ from Proposition \ref{prop:bounded-cover} to bound the sparsity of $\mathcal F$ as follows:

\begin{align*}
    \sum_{(\mathcal G, S_1)\in \mathcal F} \sum_{S_0\in \mathcal G} (\kappa(S_0,S_1) |S_0||S_1| + \zeta(S_0,S_1)) &= \text{Sparsity}(\BoundedCover(X_{\text{low}})) \\
    &+ \text{Sparsity}(\BoundedCover(X_{\text{high}})) \\
    &+ \sum_{u\in X_{\text{low}}, S_1\in \mathcal U} (\kappa(\{u\},S_1)|S_1| + \zeta(\{u\},S_1))\\
    &\le \poly(d,\log n) n + |X_{\text{low}}||\mathcal U| (24dC_2 + 3)\\
    &\le \poly(d,\log n) n\\
\end{align*}
as desired.

\textbf{Efficiency bound}. We use the efficiency bounds of Proposition \ref{prop:bounded-cover} along with the bound on $|\mathcal U|$:

\begin{align*}
    \sum_{(\mathcal G, S_1)\in \mathcal F} \left(|S_1| + \sum_{S_0\in \mathcal G} |S_0|\right) &= \text{Efficiency}(\BoundedCover(X_{\text{low}})) \\
    &+ \text{Efficiency}(\BoundedCover(X_{\text{high}})) \\
    &+ \sum_{S_1\in \mathcal U} (|S_1| + |X_{\text{low}}|)\\
    &\le \poly(d,\log n) n + |X_{\text{high}}| + |\mathcal U||X_{\text{low}}|\\
    &\le \poly(d,\log n) n\\
\end{align*}
as desired.
\end{proof}

\subsubsection{\texorpdfstring{$(\zeta,\kappa)$}{}-cover for weighted IP graphs with polylogarithmic dependence on norm}

We now apply the subroutine from the previous subsection to produce a cover for weighted inner product graphs on vectors with arbitrary norms. However, we allow the sparsity and efficiency of the cover to depend on the ratio $\tau$ between the maximum and minimum norm of points in $X$. To obtain this cover, we bucket vectors by norm and call $\TwoBoundedCover$ on all pairs of buckets.

\begin{proposition}[Log-dependence cover]\label{prop:log-cover}
Given a set of vectors $X\subseteq \mathbb{R}^d$ with $\tau = \frac{\max_{x\in X} \|x\|_2}{\min_{x\in X} \|x\|_2}$ and $n = |X|$, there is a $\poly(d,\log n,\log \tau) n$-time algorithm $\LogCover(X)$ that produces a $\poly(d,\log n,\log \tau) n$-sparse, $\poly(d,\log n,\log \tau) n$-efficient $(\zeta,\kappa)$-cover for the weighted inner product graph $G$ on $X$.
\end{proposition}

\begin{algorithm}[!h]\caption{}
\begin{algorithmic}[1]
\Procedure{\LogCover}{$X$}

    \State \textbf{Input}: $X\subseteq \mathbb{R}^d$
    
    \State \textbf{Output}: An sparse, efficient $(\zeta,\kappa)$ cover for the weighted inner product graph $G$ on $X$
    
    \State $d_{\min}\gets \min_{x\in X} \|x\|_2$, $d_{\max}\gets \max_{x\in X} \|x\|_2$
    
    \For{$i\in \{0,1,\hdots,\log \tau\}$}
    
        $X_i\gets \{x\in X: \|x\|_2\in [2^i d_{\min}, 2^{i+1}d_{\min})\}$
    
    \EndFor
    
    \State $\mathcal F = \emptyset$
    
    \For{each pair $i,j\in \{0,1,\hdots,\log \tau\}$}
    
        \State Add $\TwoBoundedCover(X_i\cup X_j)$ to $\mathcal F$
    
    \EndFor
    
    \State \Return $\mathcal F$
    
\EndProcedure
\end{algorithmic}
\end{algorithm}

\begin{proof}
\textbf{Coverage}. For any edge $\{u,v\}\in E(G)$, there exists a pair $i,j\in \{0,1,\hdots,\log \tau\}$ for which $u,v\in X_i\cup X_j$. Therefore, the \emph{Coverage} property for $\TwoBoundedCover(X_i\cup X_j)$ (which is part of $\mathcal F$) implies that the pair $\{u,v\}$ is covered by $\mathcal F$.

\textbf{Runtime, efficiency, and sparsity}. Efficiency and sparsity of $\mathcal F$ are at most the sum of the efficiencies and sparsities respectively of the constituent $\TwoBoundedCover$s, each of which are at most $\poly(d,\log n) n$ by Proposition \ref{prop:two-bounded-cover}. There are $O(\log^2 \tau)$ such covers in $\mathcal F$, so the efficiency and sparsity of $\mathcal F$ is at most $\poly(d,\log n)\log^2 \tau n\le \poly(d,\log n,\log \tau) n$, as desired. Runtime is also bounded due to the fact that there are at most $\log^2 \tau$ for loop iterations.
\end{proof}

\subsubsection{Desired \texorpdfstring{$(\zeta,\kappa,\delta)$}{}-cover}

Now, we obtain a cover for all $X$ with sparsity, efficiency, and runtime $\poly(d,\log n)$. To do this, we break up pairs to cover $\{u,v\}\in X\times X$ into two types. Without loss of generality, suppose that $\|u\|_2\le \|v\|_2$. The first type consists of pairs for which $\|v\|_2\le (dn)^{1000} \|u\|_2$. These pairs are covered using several $\LogCover$s. The total efficiency, sparsity, and runtime required for these covers is $\poly(d,\log n) n$ due to the fact that each vector is in at most $\poly(d,\log n)$ of these covers.

The second type consists of all other pairs, i.e. those with $\|v\|_2 > (dn)^{1000} \|u\|_2$. For these pairs, we take care of them via a clustering argument. We cluster all vectors in $X$ into $d+1$ clusters in a greedy fashion. Specifically, we sort vectors in decreasing order by norm and create a new cluster for a vector $x\in X$ if $|\langle x,y\rangle | < \frac{1}{d+1} \|x\|_2 \|y\|_2$ for the first vector $y$ in each cluster. Otherwise, we assign $x$ to an arbitrary cluster for which $|\langle x,y\rangle | \ge \frac{1}{d+1} \|x\|_2 \|y\|_2$ for first cluster vector $y$. We then cover the pair $\{u,v\}$ using the pair of sets $(\{u\},C)$, where $C$ is the cluster containing $v$. To argue that this satisfies the \emph{Coverage} property, we exploit the norm condition on the pair $\{u,v\}$. To bound efficiency, sparsity, and runtime, it suffices to bound the number of clusters, which is at most $d+1$ by Proposition \ref{prop:inner-product-dep}.

In order to define this algorithm, we use the notion of an \emph{interval family}, which is exactly the same as the one-dimensional interval tree from computational geometry \Aaron{Cite?}.

\begin{definition}
For a set $X$ and a function $f:X\rightarrow \mathbb{R}$, define the \emph{interval family} for $X$, denoted $\mathcal X = \IntervalFamily(X)$, to be a family of sets produced recursively by initializing $\mathcal X = \{X\}$ and repeatedly taking an element $S\in \mathcal X$, splitting it evenly into two subsets $S_0$ and $S_1$ for which $\max_{x\in S_0} f(x)\le \min_{x\in S_1} f(x)$, and adding $S_0$ and $S_1$ to $\mathcal X$ until $\mathcal X$ contains all singleton subsets of $X$.

$\mathcal X$ has the property that for any set $S\subseteq X$ consisting of all $x\in X$ for which $a\le f(x)\le b$ for two $a,b\in \mathbb{R}$, $S$ is the disjoint union of $O(\log |X|)$ sets in $\mathcal X$. Furthermore, each element in $X$ is in at most $O(\log |X|)$ sets in $\mathcal X$.
\end{definition}

\begin{proposition}[Desired cover]\label{prop:desired-cover}
Given a set $X\subseteq \mathbb{R}^d$ with $n = |X|$, there is a $\poly(d,\log n) n$-time algorithm $\DesiredCover(X)$ that produces a $\poly(d,\log n) n$-sparse, $\poly(d,\log n) n$-efficient $(\zeta,\kappa,\delta)$-cover for the weighted inner product graph $G$ on $X$.
\end{proposition}

\begin{algorithm}[!h]\caption{}
\begin{algorithmic}[1]
\Procedure{\DesiredCover}{$X$}

    \State \textbf{Input}: $X\subseteq \mathbb{R}^d$
    
    \State \textbf{Output}: An sparse, efficient $(\zeta,\kappa,\delta)$-cover for the weighted inner product graph $G$ on $X$, where $\zeta,\kappa$, and $\delta$ are defined in the proof of Proposition \ref{prop:desired-cover}.
    
    \State $\mathcal F\gets \emptyset$
    
    \State $\xi\gets (dn)^{1000}$
    
    \State $d_{\min}\gets \min_{x\in X} \|x\|_2$, $d_{\max}\gets \max_{x\in X} \|x\|_2$
    
    \Comment{Cover nearby norm pairs}
    
    \For{$i\in \{0,1,\hdots,\log (d_{\max}/d_{\min}) - \lceil \log \xi \rceil$}
    
        \State $X_i\gets \{x\in X: \|x\|_2\in [d_{\min}2^i,d_{\min}2^{i+1})$
        
        \State Add $\LogCover(X_i\cup X_{i+1}\cup \hdots\cup X_{i+\lceil \log \xi \rceil})$ to $\mathcal F$
    
    \EndFor
    
    \Comment{Create approximate basis for spread pairs}
    
    $B\gets \emptyset$
    
    \For{$x\in X$ in decreasing order by $\|x\|_2$}
    
        \If{there does not exist $y\in B$ for which $|\langle x,y\rangle| \ge \|x\|_2 \|y\|_2/(d+1)$}
        
            \State Add $x$ to $B$ and initialize a cluster $C_x = \{x\}$
        
        \Else
        
            \State Add $x$ to $C_y$ for an arbitrary choice of $y$ satisfying the condition
        
        \EndIf
    
    \EndFor
    
    \Comment{Cover the spread pairs}
    
    \For{$w\in B$}
    
        \State $\mathcal C_w\gets \IntervalFamily(C_w)$
    
        \For{each set $C\in \mathcal C_w$}
        
            \State $\mathcal G_C\gets$ the family of all singletons of $x\in X$ for which the disjoint union of $O(\log n)$ sets for $\{y\in C_w: \|y\|_2 \ge \xi \|u\|_2\}$ obtained from $\mathcal C_w$ contains $C$
        
            \State Add $(\mathcal G_C, C)$ to $\mathcal F$
        
        \EndFor
    
    \EndFor
    
    \Return $\mathcal F$
    
\EndProcedure
\end{algorithmic}
\end{algorithm}

\begin{proof}
We start by defining the functions $\zeta,\kappa$, and $\delta$. Let $\zeta_i$ and $\kappa_i$ denote the functions for which $\LogCover(Y_i)$ is a $(\zeta_i,\kappa_i)$-cover for the weighted inner product graph on $Y_i$, where $Y_i = X_i\cup X_{i+1}\cup \hdots\cup X_{i+\lceil \log \xi \rceil}$ for all $i\le \log(d_{\max}/d_{\min}) - \lceil \log \xi \rceil$. Let

\[
  \zeta(S_0,S_1) =
  \begin{cases}
                                   \sum_{i: S_0,S_1\subseteq Y_i, \zeta_i(S_0,S_1)\ne \infty} \zeta_i(S_0,S_1) & \text{if there exists $i$ for which $S_0,S_1\subseteq Y_i$} \\
                                   2 & \text{if $S_1\in \mathcal C_y$ for some $y\in B$ and $S_0 = \{x\}$ for some $x$}\\
                                   & \text{with $\|x\|_2\le \min_{a\in C} \|a\|_2/\xi$} \\
                                    \infty & \text{otherwise}
  \end{cases}
\]

\[
  \kappa(S_0,S_1) =
  \begin{cases}
                                   \sum_{i: S_0,S_1\subseteq Y_i, \kappa_i(S_0,S_1)\ne \infty} \kappa_i(S_0,S_1) & \text{if there exists $i$ for which $S_0,S_1\subseteq Y_i$} \\
                                   0 & \text{if $S_1\in \mathcal C_y$ for some $y\in B$ and $S_0 = \{x\}$ for some $x$}\\
                                   & \text{with $\|x\|_2\le \min_{a\in C} \|a\|_2/\xi$} \\
                                    \infty & \text{otherwise}
  \end{cases}
\]

\[
  \delta(S_0,S_1) =
  \begin{cases}
                                   0 & \text{if there exists $i$ for which $S_0,S_1\subseteq Y_i$} \\
                                   \frac{1}{|S_1|} & \text{if $S_1\in \mathcal C_y$ for some $y\in B$ and $S_0 = \{x\}$ for some $x$}\\
                                   & \text{with $\|x\|_2\le \min_{a\in C} \|a\|_2/\xi$} \\
                                    \infty & \text{otherwise}
  \end{cases}
\]

Before proving that the required guarantees are satisfied, we bound some important quantities.

\underline{Bound on $w_{yw}$ in terms of $w_{uy}$ for $y\in C_w$ if $\|y\|_2 \ge \xi \|u\|_2$}. By definition of $C_w$, $w_{yw}\ge \frac{1}{d+1} \|y\|_2 \|w\|_2$ for any $y\in C_w$. $w$ was the first member added to $C_w$, so $\|w\|_2\ge \|y\|_2$. By the norm assumption on $y$, $\|y\|_2 \ge \xi \|u\|_2$. By Cauchy-Schwarz, $\|y\|_2 \|u\|_2 \ge |\langle y,u\rangle| = w_{uy}$. Therefore,

$$w_{yw} \ge \frac{\xi}{d+1} w_{uy}$$

\underline{Bound on $\Reff_G(u,w)$ for $w\in B$}. We start by bounding the effective resistance between $u\in X$ and any $w\in B$ for which $\|w\|_2 > \xi \|u\|_2$. Recall that $w_{xy} = |\langle x,y\rangle|$ for any $x,y\in X$. Consider any $C\in \mathcal C_w$ for which $\min_{a\in C} \|a\|_2 > \xi \|u\|_2$. We show that

$$\Reff_G(u,w) \le \frac{2}{\sum_{y\in C} w_{uy}}$$
Consider all 2-edge paths of the form $u$-$y$-$w$ for $y\in C$. By assumption on $C$, $\|y\|_2\ge \xi \|u\|_2$ for any $y\in C$. Therefore, the bound on $w_{yw}$ applies:

$$w_{yw} \ge \frac{\xi}{d+1} w_{uy}$$
for any $y\in C$. By series-parallel reductions, the $u$-$w$ effective resistance is at most

\begin{align*}
    \Reff_G(u,w) &\le \frac{1}{\sum_{y\in C} \frac{1}{1/w_{uy} + 1/w_{yw}}}\\
    &\le \frac{1}{\sum_{y\in C} \frac{1}{(1 + (d+1)/\xi)/w_{uy}}}\\
    &\le \frac{2}{\sum_{y\in C} w_{uy}}\\
\end{align*}
as desired.

\underline{Bound on $\Reff_G(u,y)$ for $y\in C$}. Any $y\in C$ has the property that $\|y\|_2\ge \xi \|u\|_2$. Therefore, for $y\in C$, $\Reff_G(y,w)\le \frac{d+1}{\xi w_{uy}} \le \frac{1}{|C| w_{uy}}$. By the triangle inequality for effective resistance,

$$\Reff_G(u,y)\le \Reff_G(u,w) + \Reff_G(w,y)\le \frac{1}{|C| w_{uy}} + \frac{2}{\sum_{a\in C} w_{ua}}$$

\textbf{Coverage}. For any pair $\{u,v\}$ for which there exists $i$ with $u,v\in Y_i$, $\{u,v\}$ is still covered by $\mathcal F$ by the \emph{Coverage} property of $\LogCover(Y_i)$. Therefore, we may assume that this is not the case. Without loss of generality, suppose that $\|v\|_2 \ge \|u\|_2$. Then, by assumption, $\|v\|_2\ge \xi \|u\|_2$. By definition of the $C_w$s, there exists a $w\in B$ for which $v\in C_w$. By the first property of interval families, the set $\{x\in C_w: \|x\|_2 \ge \xi \|u\|_2\}$ is the disjoint union of $O(\log n)$ sets in $\mathcal C_w$. Let $C$ be the unique set among these for which $v\in C$. By our effective resistance bound,

\begin{align*}
    \Reff_G(u,v) &\le \frac{1}{|C| w_{uv}} + \frac{2}{\sum_{x\in C} w_{ux}}\\
    &= \frac{\delta(\{u\},C)}{w_{uv}} + \frac{\kappa(\{u\},C)}{\max_{x\in C} w_{ux}} + \frac{\zeta(\{u\},C)}{\sum_{x\in C} w_{ux}}\\
\end{align*}
so the coverage property for the pair $\{u,v\}$ is satisfied within $\mathcal F$ by the pair $(\mathcal G_C, C)$, as desired.

\textbf{Efficiency}. The efficiency of $\mathcal F$ is at most the efficiency of the $\LogCover$s and the remaining part for spread pairs. We start with the $\LogCover$s. By Proposition \ref{prop:log-cover},

\begin{align*}
    \sum_i \text{Efficiency}(\LogCover(Y_i)) &\le \sum_i \poly(d,\log |Y_i|,\log \xi) |Y_i|\\
    &\le \sum_i \poly(d,\log n) |Y_i|\\
    &\le (\log \xi + 1) \poly(d,\log n) \sum_i |X_i|\\
    &\le \poly(d,\log n) n\\
\end{align*}
Therefore, we just need to bound the efficiency of the remainder of $\mathcal F$. The efficiency of $\mathcal F$ is at most

\begin{align*}
    \text{Efficiency}(\mathcal F) &= \sum_{(\mathcal G,S_1)\in \mathcal F} \sum_{S_0\in \mathcal G} ((\delta(S_0,S_1) + \kappa(S_0,S_1))|S_0||S_1| + \zeta(S_0,S_1))\\
    &= \sum_i \text{Efficiency}(\LogCover(Y_i))\\
    &+ \sum_{w\in B}\sum_{C\in \mathcal C_w}\sum_{\{u\}\in \mathcal G_C} (\delta(\{u\},C)|C| + \zeta(\{u\},C))\\
    &\le \poly(d,\log n) n + \sum_{w\in B}\sum_{C\in \mathcal C_w} 3|\mathcal G_C|\\
\end{align*}
By the first property of interval families, each $x\in X$ is present as a singleton in at most $O(\log n)$ $\mathcal G_C$s for $C$ that are a subset of a given $C_w$. Therefore,

$$\text{Efficiency}(\mathcal F)\le \poly(d,\log n) n + \sum_{w\in B} O(\log n) n$$
By Proposition \ref{prop:inner-product-dep}, $|B|\le d+1$. Therefore, $\text{Efficiency}(\mathcal F)\le \poly(d,\log n) n$, as desired.

\textbf{Sparsity}. By Proposition \ref{prop:log-cover},

$$\sum_i \text{Sparsity}(\LogCover(Y_i)) \le \sum_i \poly(d,\log n,\log \xi) |Y_i|\le \poly(d,\log n) n$$
Therefore, we may focus on the remaining part for spread pairs. In particular,

\begin{align*}
    \text{Sparsity}(\mathcal F) &= \sum_{(\mathcal G,S_1)\in \mathcal F}(|S_1| + \sum_{S_0\in \mathcal G} |S_0|)\\
    &\le \poly(d,\log n) n\\
    &+ \sum_{w\in B}\sum_{C\in \mathcal C_w}(|C| + |\mathcal G_C|)\\
    &\le \poly(d,\log n) n + \sum_{w\in B}\sum_{C\in \mathcal C_w} |C|
\end{align*}
where the last inequality follows from the first property of interval families. By the second property of interval families, each element of $C_w$ is in at most $O(\log n)$ sets in $\mathcal C_w$, so $\sum_{C\in \mathcal C_w} |C|\le O(\log n) |C_w|$. Since $|B|\le d+1$, $\text{Sparsity}(\mathcal F)\le \poly(d,\log n) n$, as desired.

\textbf{Runtime}. The first for loop takes $\sum_i \poly(d,\log n) |Y_i|\le \poly(d,\log n) n$ by Proposition \ref{prop:log-cover}. The second for loop takes $O(d|B|n)\le \poly(d,\log n) n$ by the bound on $|B|$. The third for loop takes $\poly(d,\log n) n$ by the runtime for $\IntervalFamily$ and the two properties of interval families. Therefore, the total runtime is $\poly(d,\log n) n$, as desired.

\end{proof}

\subsubsection{Proof of Lemma \ref{lem:inner-sparsify}}

\begin{proof}[Proof of Lemma \ref{lem:inner-sparsify}]
Follows immediately from constructing the cover $\mathcal F$ given by Proposition \ref{prop:desired-cover} and plugging that into Proposition \ref{prop:sparsify-given-cover}.
\end{proof} 
\section{Hardness of Sparsifying and Solving Non-Multiplicatively-Lipschitz Laplacians} \label{sec:hardnessnonlip}

We now define some terms to state our hardness results:

\begin{definition}
For a decreasing function $f$ that is not $(C,L)$-multiplicatively Lipschitz, there exists a point $x$ for which $f(Cx) \le C^{-L} f(x)$. Let $x_0$ denote one such point.
A set of real numbers $S\subseteq \mathbb{R}_{\ge 0}$ is called \emph{$\rho$-discrete} for some $\rho > 1$ if for any pair $a,b\in S$ with $b > a$, $b\le \rho a$. A set of points $X\subseteq \mathbb{R}^d$ is called \emph{$\rho$-spaced} for some $\rho > 1$ if there is some $\rho$-discrete set $S\subseteq \mathbb{R}_{\ge 0}$ with the property that for any pair $x,y\in X$, $\|x - y\|_2\in S$. $S$ is called the \emph{distance set} for $X$.
\end{definition}

\begin{table}[!h]
\centering
\begin{tabular}{|l|l|l|l|l|l|}
    \hline
   Dim. & Thm. & $d$ & $g(p)$ & $\rho$ & Time \\ \hline
   Low & \ref{thm:low-spars-hard} & $c^{\log^* n}$ & $p$ & $1+16 \log( 10 (L^{1/(4c_0)}) ) / L$ & $O(n L^{1/(8c_0)})$ \\ \hline
   High & \ref{thm:high-spars-hard} & $\log n$ & $e^p$ &  $1+2 \log(10 (2^{L^{0.48}})) /L$ & $O(n 2^{L^{.48}})$\\ \hline
\end{tabular}\caption{Sparsification Hardness}\label{tab:spars-hard}
\end{table}

In this section, we show the following two hardness results:

\begin{theorem}[Low-dimensional sparsification hardness]\label{thm:low-spars-hard}
Consider a decreasing function $f:\mathbb{R}_{\ge 0}\rightarrow \mathbb{R}_{\ge 0}$
that is not $(\rho,L)$-multiplicatively lipschitz for some $L > 1$, where $c$ and $c_0$ are the constants given in Theorem \ref{thm:lownnhard} and $\rho = 1 +  2\log (10L^{1/(4c_0)}) / L$. There is no algorithm that, given a set of $n$ points $X$ in $d = c^{\log^* n}$ dimensions, returns a sparsifier of the $f$-graph for $X$ in less than $O(nL^{1/(8c_0)})$ time assuming {\sf SETH}.
\end{theorem}

\begin{theorem}[High-dimensional sparsification hardness]\label{thm:high-spars-hard}
Consider a decreasing function $f:\mathbb{R}_{\ge 0}\rightarrow \mathbb{R}_{\ge 0}$ 
that is not $(\rho,L)$-multiplicatively Lipschitz for some $L > 1$, where $\rho = 1 + 2\log (10 (2^{L^{0.48}}))/L$. There is no algorithm that, given a set of $n$ points $X$ in $d = O(\log n)$ dimensions, returns a sparsifier of the $f$-graph for $X$ in less than $O(n 2^{L^{.48}})$ time assuming {\sf SETH}.
\end{theorem}

Both of these results follow from the following reduction:

\begin{lemma}\label{lem:spars-reduc}
Consider a decreasing function $f:\mathbb{R}_{\ge 0}\rightarrow \mathbb{R}_{\ge 0}$ that is not $(\rho,L)$-multiplicatively Lipschitz for some $L > 1$, where $\rho = 1 + 2\log (10 n ) /L$ and $n > 1$. Suppose that there is an algorithm $\mathcal A$ that, when given a set of $n$ points $X\subseteq \mathbb{R}^d$, returns a 2-approximate sparsifier for the $f$-graph of $X$ with $O( n )$ edges in $\T(n, L, d)$ time. Then, there is an algorithm (Algorithm~\ref{alg:spars-reduc}) that, given two sets $A,B\subseteq \mathbb{R}^d$ for which $A\cup B$ is $\rho$-spaced with distance set $S$, $k\in S$, and $|A \cup B| = n$, returns whether or not $\min_{a\in A, b\in B} \|a - b\|_2 \le k$ in 
\begin{align*}
O(\T(|A\cup B|,L,d) + |A \cup B|)
\end{align*}
time.
\end{lemma}
The reduction described starts by scaling the points in $A \cup B$ by a factor of $x_0/k$ to obtain $\wt{A}$ and $\wt{B}$ respectively. Then, it sparsifies the $f$-graph for $\wt{A} \cup \wt{B}$. Finally, it computes the weight of the edges in the $\wt{A}$-$\wt{B}$ cut. Because $f$ is not multiplicatively Lipschitz and the distance set for $\wt{A} \cup \wt{B}$ is spaced, thresholding suffices for solving the $A\times B$ nearest neighbor problem.

\begin{proof}[Proof of Lemma \ref{lem:spars-reduc}]

Consider the following algorithm, \textsc{BichromaticNearestNeighbor} (Algorithm~\ref{alg:spars-reduc}), given below:

\begin{algorithm}\caption{}\label{alg:spars-reduc}
\begin{algorithmic}[1]
\Procedure{\textsc{BichromaticNearestNeighbor}}{$A,B,k$} \Comment{Lemma~\ref{lem:spars-reduc}}

    \State \textbf{Given}: $A,B\subset \mathbb{R}^d$ with the property that $A\cup B$ is $\rho$-spaced, where $\rho = 1 + 2(\log (10 n))/L$, and $k\in S$
    
    \State \textbf{Returns}: whether there are $a\in A, b\in B$ for which $\|a - b\|_2 \le k$
    
    \State $\wt{A} \gets \{a \cdot \sqrt{x_0} / k  ~|~ \forall a\in A\}$ \Comment{$\wt{A} \subset \R^d$}
    
    \State $\wt{B} \gets \{b \cdot \sqrt{x_0} / k  ~|~ \forall b\in B\}$ \Comment{$\wt{B} \subset \R^d$}
    
    \State $H\gets \mathcal A(\wt{A} \cup \wt{B})$ \Comment{$H$ is a 2-approximate sparsifier for the $f$-graph $G$ of $\wt{A} \cup \wt{B}$}
    
    \If{the total weight of edges between $\wt{A}$ and $\wt{B}$ in $H$ is at least $f(x_0)/2$}
        
        \State \Return $\mathsf{true}$
        
    \Else
    
        \State \Return $\mathsf{false}$
        
    \EndIf

\EndProcedure
\end{algorithmic}
\end{algorithm}

We start by bounding the runtime of this algorithm. Constructing $\wt{A}$ and $\wt{B}$ and calculating the total weight of edges between $\wt{A}$ and $\wt{B}$ takes $O( n )$ time since $H$ has $O( n )$ edges. Since the sparsification algorithm is only called once, the total runtime is therefore $\T(n, L, d) + O( n )$, as desired. For the rest of the proof, we may therefore focus on correctness.

First, suppose that $\min_{a\in A, b\in B} \|a - b\|_2 \le k$. There exists a pair of points $\wt{a} \in \wt{A}$, $\wt{b} \in \wt{B}$ with $\| \wt{a} - \wt{b} \|_2 \le \sqrt{x_0}$. Since $f$ is a decreasing function, the edge between $\wt{a}$ and $\wt{b}$ in $G$ has weight at least $f(x_0)$, which means that the total weight of edges in the $\wt{A}$-$\wt{B}$ cut in $G$ is at least $f(x_0)$. Since $H$ is a 2-approximate sparsifier for $G$, the total weight of edges in the $\wt{A}$-$\wt{B}$ cut is at least $f(x_0)/2$. This means that $\mathsf{true}$ is returned, as desired.

Next, suppose that $\min_{a\in A, b\in B} \|a - b\|_2 > k$. Since $A\cup B$ is $\rho$-spaced with distance set $S$ and $k\in S$, $\| a - b \|_2 \ge \rho \cdot k$ for all $a\in A$ and $b\in B$. Therefore, $\| \wt{a} - \wt{b} \|_2 \ge \rho \cdot \sqrt{x_0}$ for all $\wt{a} \in \wt{A}$ and $\wt{b} \in \wt{B}$. 

Since $f$ is decreasing and not $(\rho,L)$-multiplicatively Lipschitz, the weight of any edge between $\wt{A}$ and $\wt{B}$ in $G$ is at most 
\begin{align*}
f(\rho \cdot x_0) \le f(x_0)/(100 n^2).
\end{align*} 

The total weight of edges between $\wt{A}$ and $\wt{B}$ is therefore at most 
\begin{align*}
n^2 \cdot (f(x_0)/(100 n^2)) < f(x_0)/8.
\end{align*}
Since $H$ is a 2-approximate sparsifier for $G$, the total weight between $C$ and $D$ is at most $f(x_0)/4 < f(x_0)/2$, so the algorithm returns $\mathsf{false}$, as desired.
\end{proof}

We now prove the theorems:

\begin{proof}[Proof of Theorem \ref{thm:low-spars-hard}]
Consider an instance of $\ell_2$-bichromatic closest pair for $n = L^{1/(4c_0)}$ and $d = c^{\log^* n}$, where $c$ is the constant given in the dimension bound of Theorem \ref{thm:lownnhard} and $c_0$ is such that integers have bit length $c_0\log n$ in Theorem \ref{thm:lownnhard}. This consists of two sets of points $A,B\subseteq \mathbb{R}^d$ with $|A\cup B| = n$ for which we wish to compute $\min_{a\in A, b\in B} \|a - b\|_2$. By Theorem \ref{thm:lownnhard}, the coordinates of points in $A$ are also $c_0\log n$ bit integers. Therefore, the set $S$ of possible $\ell_2$ distances between points in $A$ and $B$ is a set of square roots of integers with $\log d + c_0\log n\le 2c_0\log n$ bits. Therefore, $S$ is a $\rho$-discrete (recall $\rho = 1 + (2\log (10n))/L$), since $1 + 1/n^{2c_0} > 1 + (2\log (10n))/L$. Furthermore, note that $f$ is not $(\rho,L)$-multiplicatively Lipschitz.

We now describe an algorithm for solving $\ell_2$-closest pair on $A\times B$. Use binary search on the values in $S$ to compute the minimum distance between points in $A,B$. For each query point $k\in S$, by Lemma \ref{lem:spars-reduc}, there is a 
\begin{align*}
\T(n,L,d) = O(\T(L^{1/(4c_0)},L,c^{\log^* L}) + L^{1/(4c_0)})
\end{align*}
-time algorithm for determining whether or not the closest pair has distance at most $k$. Therefore, there is a 
\begin{align*}
O( \log |S| \cdot (\T(L^{1/(4c_0)},L,c^{\log^* L}) + L^{1/4c_0})) = \tilde{O}(L^{1/(4c_0)}L^{1/(8c_0)}) < O( n^{3/2} )
\end{align*}
time algorithm for solving $\ell_2$-closest pair on pairs of sets with $n$ points. But this is impossible given {\sf SETH} by Theorem \ref{thm:lownnhard}, a contradiction. This completes the result.
\end{proof}

\begin{proof}[Proof of Theorem \ref{thm:high-spars-hard}]
\Aaron{Check constants}
Consider an instance of bichromatic Hamming nearest neighbor search for $n = 2^{L^{0.49}}$ and $d = c_1\log n$ for the constant $c_1$ in the dimension bound in Theorem \ref{thm:r18}. This consists of two sets of points $A , B \subseteq \mathbb{R}^d$ with $|A\cup B| = n$ for which we wish to compute $\min_{a\in A, b\in B} \|a - b\|_2$. The coordinates of points in $A$ and $B$ are 0-1. Therefore, the set $S$ of possible $\ell_2$ distances between points in $A$ and $B$ is the set of square roots of integers between 0 and $c_1\log n$, which differ by a factor of at least $1 + 1/(2c_1\log n) > \rho$ (recall $\rho = 1 + (2\log (10n))/L$). Therefore, $A\cup B$ is $\rho$-spaced. Note that $f$ is also no $(\rho,L)$-multiplicatively Lipschitz by definition.

We now give an algorithm for solving $\ell_2$-closest pair on $A\times B$. Use binary search on $S$. For each query $k\in S$, Lemma \ref{lem:spars-reduc} implies that one can check if there is a pair with distance at most $k$ in 
\begin{align*}
\T(n,L,d)
= & ~ O(\T(2^{L^{.49}},L,c_1 L^{.49}) + 2^{L^{.49}}) \\
\le & ~ O(2^{L^{.49} + L^{.48}}) \\ 
= & ~ n^{1 + o(1)}
\end{align*}
time on pairs of sets with $n$ points. But this is impossible given {\sf SETH} by Theorem \ref{thm:r18}. This completes the result.
\end{proof}

Next, we prove hardness results for solving Laplacian systems. In these hardness results, we insist that kernels are bounded:

\begin{definition}
Call a function $f:\mathbb{R}_{\ge 0}\rightarrow \mathbb{R}_{\ge 0}$ \emph{$g(p)$-bounded} for a function $g:\mathbb{R}_{\ge 0}\rightarrow \mathbb{R}_{\ge 0}$ iff for any pair $a,b > 0$ with $b > a$, $f(b) \ge f(a) \cdot g(b/a)$. Call a set $S\subseteq \mathbb{R}_{\ge 0}$ \emph{$\gamma$-bounded} iff $\max_{s\in S} s \le \gamma \min_{s\in S,s\ne 0} s$. A set of points $X\subseteq \mathbb{R}^d$ is called \emph{$\gamma$-boxed} iff the set of distances between points in $X$ is $\gamma$-bounded.
\end{definition}

\begin{table}[!h]
\centering
\begin{tabular}{|l|l|l|l|l|l|l|}
    \hline
   Dim. & Thm. & $d$ & $g(p)$ & $\rho$ & Time & $\epsilon$ \\ \hline
   Low & \ref{thm:low-lsolve-hard} & $c^{\log^* n}$ & $p$ & $1+16 \log( 10 (L^{1/(4c_0)}) ) / L$ & $n \log (g(\gamma)) L^{1/(64c_0)} $ & $1/ ( g(\gamma)^3 2^{\poly(\log n)} )$ \\ \hline
   High & \ref{thm:high-lsolve-hard} & $\log n$ & $e^p$ &  $1+2 \log(10 (2^{L^{0.48}})) /L$ & $n \log(g(\gamma)) 2^{L^{.48}}$ & $1/ ( g(\gamma)^3 2^{\poly(\log n)} )$ \\ \hline
\end{tabular}\caption{Linear System Hardness}\label{tab:lsolve-hard}
\end{table}

We show the following results:

\begin{theorem}[Partial low-dimensional linear system hardness]\label{thm:low-lsolve-hard}
Consider a decreasing $g(p)=p$-bounded function $f:\mathbb{R}_{\ge 0}\rightarrow \mathbb{R}_{\ge 0}$
that is not $(\rho,L)$-multiplicatively Lipschitz for some $L > 1$, where $c$ and $c_0$ are the constants given in Theorem \ref{thm:lownnhard} and $\rho = 1 + 16\log(10 (L^{1/(32c_0)}))/L$. Assuming {\sf SETH}, there is no algorithm that, given a $\gamma$-boxed set of $n$ points $X$ in $d = c^{\log^* n}$ dimensions with $f$-graph $G$ and a vector $b\in \mathbb{R}^n$, returns a $\epsilon = 1/(g(\gamma)^3 2^{\poly(\log n)})$-approximate solution $x\in \mathbb{R}^n$ to the geometric Laplacian system $L_G x = b$ in less than $O(n \log(g(\gamma)) L^{1/(64c_0)})$ time.
\end{theorem}

\begin{theorem}[Partial high-dimensional linear system hardness]\label{thm:high-lsolve-hard}
Consider a decreasing $g(p)=e^p$-bounded function $f:\mathbb{R}_{\ge 0}\rightarrow \mathbb{R}_{\ge 0}$ that is not $(\rho,L)$-multiplicatively Lipschitz for some $L > 1$, where $\rho = 1 + 16\log(10(2^{L^{0.48}}))/L$. Assuming {\sf SETH}, there is no algorithm that, given a $\gamma$-boxed set of $n$ points $X$ in $d = O(\log n)$ dimensions with $f$-graph $G$ and a vector $b\in \mathbb{R}^n$, returns a $\epsilon = 1/(g(\gamma)^3 2^{\poly(\log n)})$-approximate solution $x\in \mathbb{R}^n$ to the geometric Laplacian system $L_G x = b$ in less than $O(n \log(g(\gamma)) 2^{L^{.48}})$ time assuming {\sf SETH}.
\end{theorem}

To prove these theorems, we use the following reduction from bichromatic nearest neighbors:

\begin{lemma}\label{lem:lsolve-reduc}
Consider a decreasing $g(p)$-bounded function $f:\mathbb{R}_{\ge 0}\rightarrow \mathbb{R}_{\ge 0}$ that is not $(\rho,L)$-multiplicatively Lipschitz for some $L > 1$, where $\rho = 1 + 16(\log (10 n ))/L$ and $n > 1$. Suppose that there is an algorithm $\mathcal A$ that, given a $\gamma$-boxed set of $n$ points $X$ in $d$ dimensions with $f$-graph $G$ and a vector $b\in \mathbb{R}^n$, returns an $\epsilon = 1/(g(\gamma)^3 2^{\poly(\log n)})$-approximate solution $x\in \mathbb{R}^n$ to the geometric Laplacian system $L_G x = b$ in $\T(n,L,g(\gamma),d)$ time. Then, there is an algorithm that, given two sets $A,B\subseteq \mathbb{R}^d$ for which $A\cup B$ is $\rho$-spaced with $|S|^{O(1)}$-bounded distance set $S$\Aaron{determine $O(1)$}, $k\in S$, and $|A\cup B| = n$, returns whether or not $\min_{a\in A, b\in B} \|a - b\|_2 \le k$ in 
\begin{align*}
    O((\log n)\T(n,L,g(|S|^{O(1)}),d) + |A\cup B|)
\end{align*} 
time.
\end{lemma}

This reduction works in a similar way to the effective resistance data structure of Spielman and Srivastava \cite{ss11}, but with minor differences due to the fact that their data structure requires multiplication by incidence matrix of the graph, which in our case is dense. Our reduction uses Johnson-Lindenstrauss to embed the points
\begin{align*}
v_s = L^{\dagger} b_s
\end{align*}
for vertices $s$ in the graph $G$ into $O(\log n)$ dimensions in a way that distorts the distances
\begin{align*}
b_{st}^\top (L^{\dagger})^2 b_{st}
\end{align*}
for vertices $s,t$ in $G$ by a factor of at most 2. After computing this embedding, we build an $O(\log n)$-approximate nearest neighbor data structure on the resulting points. This allows us to determine whether or not a vertex in $A$ has a high-weight edge in $G$ to $B$ in almost-constant time. After looping through all of the edges in $A$ in total time $n^{1 + o(1)}$, we determine whether or not there are any high-weight edges between $A$ and $B$ in $G$, allowing us to answer the bichromatic nearest neighbors decision problem.

We start by proving a result that links norms of $v_s$ to effective resistances:

\begin{proposition}\label{prop:res2-res}
In an $n$-vertex graph $G$ with vertices $s$ and $t$,
\begin{align*}
0.5 \cdot (\Reff_G(s,t))^2 \le \|v_s - v_t\|_2^2 \le n \cdot (\Reff_G(s,t))^2
\end{align*}
\end{proposition}

\begin{proof}

{\bf Lower bound.}
Let $x = v_t - v_s = L_G^{\dagger} b_{st} \in \R^n$. By definition,
\begin{align*}
\| v_s - v_t \|_2^2 
= & ~ \| x \|_2^2\\
\ge & ~ (x_s)^2 + (x_t)^2\\
= & ~ (x_s)^2 + (x_s - b_{st}^\top L_G^{\dagger} b_{st})^2\\
\ge & ~ (b_{st}^\top L_G^{\dagger} b_{st})^2/2,
\end{align*}
where third step follows from $x_t = x_s - b_{st}^\top L_G^\dagger b_{st}$.

Thus we complete the proof of the lower bound.

{\bf Upper bound.}
Next, we prove the upper bound. The maximum and minimum coordinates of $x$ are $x_t$ and $x_s$ respectively. By definition of the pseudoinverse\Aaron{Justify further in preliminaries?}, $\text{image}(L_G^{\dagger}) = \text{image}(L_G)$. Therefore, $\textbf{1}^\top x = 0$, $x_s \le 0$, and $x_t \ge 0$. $x_s \le 0$ implies that for all $i\in [n]$, 
\begin{align*}
x_i \le x_s + b_{st}^\top L_G^{\dagger} b_{st} \le b_{st}^\top L_G^{\dagger} b_{st}.
\end{align*}
$x_t \ge 0$ implies that for all $i\in [n]$, 
\begin{align*}
x_i \ge x_t - b_{st}^\top L_G^{\dagger} b_{st} \ge -b_{st}^\top L_G^{\dagger} b_{st}.
\end{align*}
Therefore, $|x_i|\le b_{st}^\top L_G^{\dagger} b_{st} = \Reff_G(s,t)$ for all $i\in [n]$. Summing across $i\in [n]$ yields the desired upper bound.
\end{proof}

Furthermore, the minimum effective resistance of an edge across a cut is related to the maximum weight edge across the cut:

\begin{proposition}\label{prop:weight-res}
In an $m$-edge graph $G$ with vertex set $S$,
\begin{align*}
\min_{s\in S,t\notin S} \Reff_G(s,t)\le \min_{e\in \partial S} (1/w_e)\le m \min_{s\in S,t\notin S} \Reff_G(s,t) .
\end{align*}
\end{proposition}

\begin{proof}

{\bf Lower bound.}
The lower bound on $\min_e 1/w_e$ follows immediately from the fact that for any edge $e = \{s,t\}$, $\Reff_G(s,t)\le r_e$.

{\bf Upper bound.}
For the upper bound, recall that

\begin{align*}
\Reff_G(s,t) &= \min_{f\in \mathbb{R}^m : B^\top f = b_{st}} \sum_{e\in E(G)} f_e^2/w_e\\
&\ge \min_{f\in \mathbb{R}^m : B^\top f = b_{st}} \sum_{e\in \partial S} f_e^2/w_e\\
&\ge \left(\min_{f\in \mathbb{R}^m : B^\top f = b_{st}} \sum_{e\in \partial S} f_e^2\right)\left(\min_{e\in \partial S} 1/w_e\right)\\
&\ge \frac{1}{|\partial S|} \left(\min_{f\in \mathbb{R}^m : B^\top f = b_{st}} \sum_{e\in \partial S} |f_e|\right)^2\left(\min_{e\in \partial S} 1/w_e\right)
\end{align*}
where the first step follows from definition of effective resistance, 
the third step follows from taking $w$ out, and the last step follows from Cauchy-Schwarz.

Since $s\in S$ and $t\notin S$, $\sum_{e\in \partial S} |f_e|\ge 1$. Therefore,
\begin{align*}
\Reff_G(s,t) \ge \frac{1}{m}\min_{e\in \partial S} 1/w_e
\end{align*}
completing the upper bound.
\end{proof}

\begin{proposition}\label{prop:lapl-error}
Consider an $n$-vertex connected graph $G$ with edge weights $\{w_e\}_{e\in G}$, two matrices $Z,\tilde{Z}\in \mathbb{R}^{k\times n}$ with rows $\{z_i\}_{i=1}^k$ and $\{\tilde{z}_i\}_{i=1}^k$ respectively for $k\le n$, and $\epsilon\in (1/n,1)$. Suppose that both of the following properties hold:
\begin{enumerate}
    \item $\|z_i - \tilde{z}_i\|_{L_G} \le 0.01 n^{-12} w_{\min}^2 w_{\max}^{-2}\|z_i\|_{L_G}$ 
for all $i\in [k]$, where $w_{\min}$ and $w_{\max}$ are the minimum and maximum weights of edges in $G$ respectively
    \item $(1 - \epsilon/10) \cdot \| L_G^{\dagger} b_{st}\|_2 \le \|Zb_{st}\|_2 \le (1 + \epsilon/10) \cdot \|L_G^{\dagger}b_{st}\|$ for any vertices $s,t$ in $G$
\end{enumerate}

Then for any vertices $s,t$ in $G$,
\begin{align*}
(1 - \epsilon) \cdot \|L_G^{\dagger} b_{st}\|_2 \le \|\tilde{Z} b_{st}\|_2 \le (1 + \epsilon) \cdot \|L_G^{\dagger} b_{st}\|_2.
\end{align*}
\end{proposition}

\begin{proof}
We start by bounding
\begin{align*}
((z_i - \tilde{z}_i)^\top b_{st})^2
\end{align*}
for each $i\in [k]$. Since $G$ is connected, there is a path from $s$ to $t$ consisting of edges $e_1,e_2,\hdots,e_{\ell}$ in that order, where $\ell\le n$. By the Cauchy-Schwarz inequality,
\begin{align*}
((z_i - \tilde{z}_i)^\top b_{st})^2 
\le & ~ n \sum_{j=1}^{\ell} ((z_i - \tilde{z}_i)^\top b_{e_j})^2\\
\le & ~ (n/w_{\min}) \cdot \|z_i - \tilde{z}_i\|_{L_G}^2\\
\le & ~ 0.01 n^{-23} w_{\min}^3 w_{\max}^{-4} \cdot \|z_i\|_{L_G}^2\\
\end{align*}
where the last step from property 1 in proposition statement.

By the upper bound on $\|Zb_{s't'}\|_2$ for any vertices $s',t'$ in $G$, $(z_i^\top b_e)^2\le (1 + \epsilon)^2 b_e^\top (L_G^{\dagger})^2 b_e$ for all edges $e$ in $G$, so 
\begin{align*}
0.01 n^{-23} w_{\min}^3 w_{\max}^{-4} \cdot \|z_i\|_{L_G}^2 
= & ~ 0.01 n^{-23} w_{\min}^3 w_{\max}^{-4} \cdot \sum_{e\in E(G)} w_e((z_i)^\top b_e)^2\\
\le & ~ 0.01 n^{-23} w_{\min}^3 w_{\max}^{-3} \cdot \sum_{e\in E(G)} ((z_i)^\top b_e)^2\\
\le & ~ 0.01 n^{-23} w_{\min}^3 w_{\max}^{-3} \cdot \sum_{e\in E(G)} (1 + \epsilon)^2 b_e^\top (L_G^{\dagger})^2 b_e \\
\le & ~ 0.01 n^{-21} w_{\min}^3 w_{\max}^{-3} (1 + \epsilon)^2 \max_{e\in E(G)} (b_e^\top (L_G^{\dagger})^2 b_e)\\
\le & ~ 0.04 n^{-21} w_{\min}^3 w_{\max}^{-3} \max_{e\in E(G)} (b_e^\top (L_G^{\dagger})^2 b_e) .
\end{align*}
where the first step follows from $\| z_i \|_{L_G}^2 = \sum_{e \in E} z_i^\top w_e b_e b_e^\top z_i $, the second step follows from $w_{\max} = \max_{e\in G} w_e$, and the third step follows from $(z_i^\top b_e)^2 \leq (1+\epsilon)^2 b_e^\top (L_G^\dagger) b_e$, the forth step follows from summation has at most $n^2$ terms, the last step follows from $(1+\epsilon)^2 \leq 4$, $\forall \epsilon \in (0,1)$.

\Aaron{Cite Prop 6.9 instead of going through this paragraph}
$L_G^{\dagger} b_e$ is a vector that is maximized and minimized at the endpoints of $e$. Furthermore, $\textbf{1}\in \text{kernel}(L_G^{\dagger})$ by definition of the pseudoinverse. Thus, we know $L_G^\dagger b_e $ has both positive and negative coordinates and that $\|L_G^{\dagger}b_e\|_{\infty} \leq \max_{i\neq j} | (L_G^\dagger b_e)_i - (L_G^\dagger b_e)_j | \le b_e^\top L_G^{\dagger} b_e$ .

We have
\begin{align*}
& ~ 0.04 n^{-21} w_{\min}^3 w_{\max}^{-3} \cdot \max_{e\in E(G)} b_e^\top (L_G^{\dagger})^2 b_e \\
\le & ~ 0.04 n^{-20} w_{\min}^3 w_{\max}^{-3} \cdot \max_{e\in E(G)} (b_e^\top L_G^{\dagger} b_e)^2\\
\le & ~ 0.4 n^{-20} w_{\min} w_{\max}^{-3}\\
\le & ~ 0.4 n^{-16} w_{\min} w_{\max}^{-1} \cdot (b_{st}^\top L_G^{\dagger} b_{st})^2 \\
\le & ~ 0.8 n^{-16} w_{\min} w_{\max}^{-1} \cdot \| L_G^\dagger b_{st} \|_2^2
\end{align*}
where the first step follows from $\max_{e \in E(G)} b_e^\top (L_G^\dagger)^2 b_e = \max_{e\in E} \| L_G^\dagger b_e \|_2^2 \leq n \max_{e\in E} \| L_G^\dagger b_e \|_{\infty}^2 \leq n \max_{e \in E(G)} ( b_e^\top L_G^\dagger b_e )^2$, the second step follows from $\max_e (b_e^\top L_G^\dagger b_e)^2 \leq w_{\min}^{-2}$ (the lower bound in Proposition \ref{prop:weight-res}), and the third step follows from $w_{\max}^{-2} \leq n^4 (b_{st} L_G^\dagger b_{st})^2$ (the upper bound in Proposition \ref{prop:weight-res}), and the last step follows from $(b_{st}^\top L_G^{\dagger} b_{st})^2 \le 2 \| L_G^\dagger b_{st} \|_2^2$ (Since $L_G^{\dagger} b_{st}$ is maximized and minimized at $t$ and $s$ respectively).

Combining these inequalities shows that
\begin{align*}
((z_i - \tilde{z}_i)^\top b_{st})^2\le 0.8 n^{-16} w_{\min}w_{\max}^{-1} \cdot \|L_G^{\dagger}b_{st}\|_2^2
\end{align*}
Summing over all $i$ and using the fact that $k\le n$, $\epsilon > 1/n$, and $w_{\min}\le w_{\max}$ shows that
\begin{align*}
    \|(Z - \tilde{Z})b_{st}\|_2^2
    = & ~ \sum_{i=1}^k ((z_i - \tilde{z}_i)^\top b_{st})^2\\
    \leq & ~ 0.8 k n^{-16} w_{\min}w_{\max}^{-1} \cdot \|L_G^{\dagger}b_{st}\|_2^2 \\
    \leq & ~ 0.8 \epsilon^2 n^{-13} w_{\min} w_{\max}^{-1} \cdot \|L_G^{\dagger}b_{st}\|_2^2 ~ \\
    \le & ~ 0.8 \epsilon^2 n^{-13} \cdot \|L_G^{\dagger} b_{st}\|_2^2
\end{align*}

Combining this with the given upper and lower bounds on $\| Z b_{st} \|_2$ using the triangle inequality yields the desired result.
\end{proof}

\begin{proof}[Proof of Lemma \ref{lem:lsolve-reduc}]
Consider the following algorithm \textsc{BichromaticNearestNeighbor}, given below:
\begin{algorithm}\caption{  }\label{alg:bichromatic_nearest_neighor_proj}
\begin{algorithmic}[1]
\Procedure{\textsc{BichromaticNearestNeighborProj}}{$A,B,k$} \Comment{Lemma~\ref{lem:lsolve-reduc}}

    \State \textbf{Given}: $A,B\in \mathbb{R}^d$ with the property that $A\cup B$ is $\rho$-spaced with $|S|^{O(1)}$-bounded distance set\Aaron{Determine $O(1)$}
    , where $\rho = 1 + 16(\log (10 n ))/L$, and $k\in S$
    
    \State \textbf{Returns}: whether there are $a\in A, b\in B$ for which $\|a - b\|_2 \le k$
    
    \State $C\gets \{a \cdot \sqrt{x_0} / k ~|~ \forall a\in A\}$ 
    
    \State $D\gets \{b \cdot \sqrt{x_0} / k ~|~ \forall b\in B\}$
    
    \State $\ell\gets 200\log n$

    \State Construct matrix $P \in \R^{\ell \times d}$, where each entry is $1/\sqrt{\ell}$ with prob $1/2$ and $-1/\sqrt{\ell}$ with prob $1/2$ 
    
    \State $\tilde{Z}_{i,*} \gets \mathcal A( C \cup D , P_{i,*} )$ for each $i\in [\ell]$ \Comment{$P_{i,*},\tilde{Z}_{i,*}$ denotes row $i$ of $P \in \R^{\ell \times d} ,\tilde{Z} \in \R^{\ell \times n}$}
    
    \State $\widehat{C}\gets \{\tilde{Z} \cdot b_c ~|~ \forall c\in C\}$ \Comment{$b_c\in \mathbb{R}^n$ denotes the indicator vector of the vertex $c$, i.e. $(b_c)_c = 1$ and $(b_c)_i = 0$ for all $i\ne c$}
    
    \State $\widehat{D}\gets \{\tilde{Z} \cdot b_c ~|~ \forall c\in D\}$
    
    \State $t\gets (\log n)$-approximation to closest $\widehat{C}$-$\widehat{D}$ $\ell_2$-distance using Theorem \ref{thm:l2-ann}
    
    \If{$t \le 3\sqrt{n}(\log n)/f(x_0)$}
        
        \State \Return $\mathsf{true}$
        
    \Else
    
        \State \Return $\mathsf{false}$
        
    \EndIf

\EndProcedure
\end{algorithmic}
\end{algorithm}

First, we bound the runtime of $\textsc{BichromaticNearestNeighborProj}$ (Algorithm~\ref{alg:bichromatic_nearest_neighor_proj}). Computing the sets $C,D$ and the matrix $P$ trivially takes $\tilde{O}( n )$ time. Computing the matrix $\tilde{Z}$ takes 
\begin{align*}
O(\log n) \T(n,L,g(|S|^{O(1)}),d)
\end{align*}
time since the point set $C\cup D$ is $|S|^{O(1)}$-boxed. Computing $\widehat{C}$ and $\widehat{D}$ takes $\tilde{O}( n )$ time, as computing $\widehat{Z}b_i$ takes $O(\log n)$ time for each $i\in C\cup D$ since $b_i$ is supported on just one vertex. Computing $t$ takes $n^{1 + o(1)}$ time by Theorem \ref{thm:l2-ann}. In particular, one computes $t$ by preprocessing a $O(\log n)$-approximate nearest neighbors data structure on $\widehat{D}$ (takes $n^{1+o(1)}$ time), queries the data structure on all points in $\widehat{C}$ (takes $n (n^{o(1)}) = n^{1+o(1)}$ time), and returns the minimum of all of the queries. The subsequent if statement takes constant time. Therefore, the reduction takes \begin{align*}
O((\log n) \cdot \T(n,L,|S|^{O(1)},d) + n^{1 + o(1)}
\end{align*}
time overall, as desired.

Next, suppose that there exists $a\in A$ and $b\in B$ for which $\| a - b \|_2 \le k$. We show that the reduction returns $\mathsf{true}$ with probability at least $1 - 1/n$. Let $G$ denote the $f$-graph on $C\cup D$. By definition of $C$ and $D$ and the fact that $f$ is decreasing, there exists a pair of points in $C$ and $D$ with edge weight at least $f(x_0)$. 

By the lower bound on Proposition \ref{prop:weight-res}, 
\begin{align*}
\exists
p\in C,q\in D \text{~s.t.~} \Reff_G(p,q)\le 1/f(x_0).
\end{align*}

By the upper bound of Proposition \ref{prop:res2-res}, 
\begin{align*}
b_{pq}^\top  (L_G^{\dagger})^2 b_{pq} \le n \cdot (\Reff_G(p,q))^2 \le n/f(x_0)^2 .
\end{align*}

Let $Z \in \R^{\ell \times n}$ be defined as $Z = P \cdot L_G^{\dagger}$. By Theorem \ref{thm:jl} with $\epsilon = 1/2$ applied to the collection of vectors $\{L_G^{\dagger} b_s\}_{s\in C\cup D}$ and projection matrix $P$, 
\begin{align*}
\frac{1}{2} \|L_G^{\dagger}b_{st}\|_2 \le \|Z b_{st}\|_2 \le \frac{3}{2} \|L_G^{\dagger}b_{st}\|_2
\end{align*}
for all pairs $s,t\in C\cup D$ with high probability. Therefore, the second input guarantee of Proposition \ref{prop:lapl-error} is satisfied with high probability. Furthermore, for each $i\in [\ell]$, $\tilde{z}_i$, the $i$th row of $\tilde{Z}$, satisfies the first input guarantee by the output error guarantee of the algorithm $\mathcal A$. Therefore, Proposition \ref{prop:lapl-error} applies and shows that
\begin{align*}
\|\tilde{Z} \cdot b_{pq}\|_2 \le \frac{9}{4} \|L_G^{\dagger} \cdot b_{pq}\|_2 \le 3\sqrt{n} / f(x_0).
\end{align*}
This means that there exists of vectors $a\in \widehat{C},b\in \widehat{D}$ with 
\begin{align*}
\|a - b\|_2 \le 3\sqrt{n} / f(x_0).
\end{align*}
By the approximation guarantee of the nearest neighbors data structure, $t \le ( 3\sqrt{n} / f(x_0) )\log n $ and the reduction returns $\mathsf{true}$ with probability at least $1 - 1/n$, as desired.

Next, suppose that there do not exist $a\in A$ and $b\in B$ for which $\|a - b\|_2 \le k$. We show that the reduction returns $\mathsf{false}$ with probability at least $1 - 1/n$. Since $k\in S$ and $A\cup B$ is $\rho$-spaced, $\|a - b\|_2 \ge \rho \cdot k$ for all $a\in A$ and $b\in B$. Therefore, all edges between $C$ and $D$ in the $f$-graph $G$ for $C\cup D$ have weight at most $f(\rho x_0)\le \frac{f(x_0)}{100 n^{16}}$ since $f$ is not $(\rho,L)$-multiplicatively Lipschitz. 

By the upper bound of Proposition \ref{prop:weight-res}, 
\begin{align*}
\Reff_G(s,t) \ge \frac{1}{ n^2 f(\rho x_0) }\ge \frac{ 100 n^{14} }{f(x_0)}
\end{align*}
for any pair of vertices $s\in C, t\in D$. 

By the lower bound of Proposition \ref{prop:res2-res}, 
\begin{align*}
b_{st}^\top (L_G^{\dagger})^2 b_{st} \ge (\Reff_G(s,t))^2/2 \ge (\frac{50n^{14}}{f(x_0)})^2
\end{align*}
for all $s\in C, t\in D$.

Recall from the discussion of the $\|a - b\|\le k$ case that Theorem \ref{thm:jl} and Proposition \ref{prop:lapl-error} apply. Therefore, with probability at least $1 - 1/n$, by the lower bound of Proposition \ref{prop:lapl-error},
\begin{align*}
\|\tilde{Z}b_{st}\|_2 \ge \frac{1}{4} \| L_G^{\dagger} b_{st} \|_2 \ge \frac{12n^{14}}{f(x_f)}
\end{align*}
for any $s\in C,t\in D$. 

Therefore, by the approximate nearest neighbors guarantee, 
\begin{align*}
t \ge \frac{12 n^{14}}{ (\log n) \cdot f(x_0) } > \frac{ 3\sqrt{n} (\log n) }{ f(x_0) },
\end{align*}
so the algorithm returns $\mathsf{false}$ with probability at least $1 - 1/n$, as desired.
\end{proof}

We now prove the theorems:

\begin{proof}[Proof of Theorem \ref{thm:low-lsolve-hard}]
Consider an instance of bichromatic $\ell_2$-closest pair for $n = L^{1/(32c_0)}$ and $d = c^{\log^* n}$, where $c$ is the constant given in the dimension bound of Theorem \ref{thm:lownnhard} and $c_0$ is such that integers have bit length $c_0\log n$ in Theorem \ref{thm:lownnhard}. This consists of two sets of points $A,B\subseteq \mathbb{R}^d$ with $|A\cup B| = n$ for which we wish to compute $\min_{a\in A, b\in B} \|a - b\|_2$. By Theorem \ref{thm:lownnhard}, the coordinates of points in $A$ are also $c_0\log n$ bit integers. Therefore, the set $S$ of possible $\ell_2$ distances between points in $A$ and $B$ is a set of square roots of integers with $\log d + c_0\log n\le 2c_0\log n$ bits. Therefore, $S$ is a $\rho$-discrete (recall $\rho = 1 + (16\log (10n))/L$) since $1 + 1/n^{2c_0} > 1 + (16\log n)/L$. Furthermore, note that $f$ is not $(\rho,L)$-multiplicatively Lipschitz by assumption.

We now describe an algorithm for solving $\ell_2$-closest pair on $A\times B$. Use binary search on the values in $S$ to compute the minimum distance between points in $A,B$. $S$ consists of integers between 1 and $n^{c_0}$ (pairs with distance 0 can be found in linear time), so $\gamma \le n^{c_0}$. For each query point $k\in S$, by Lemma \ref{lem:lsolve-reduc}, there is a 
\begin{align*}
O((\log n)\cdot \T(n,L,g(\gamma),d) + n) = O((\log n) \cdot \T(L^{1/(32c_0)},L,L^{1/32},c^{\log^* L}) + L^{1/(32c_0)})
\end{align*}
-time algorithm for determining whether or not the closest pair has distance at most $k$. Therefore, there is a 
\begin{align*}
O((\log n) \T(L^{1/(32c_0)},L,L^{1/32},c^{\log^* L}) + L^{1/(32c_0)}) = \tilde{O}(L^{1/(32c_0)}L^{1/(64c_0)} \log(L^{1/32})) < \tilde{O}(n^{3/2})
\end{align*}
time algorithm for solving $\ell_2$-closest pair on pairs of sets with $n$ points. But this is impossible given {\sf SETH} by Theorem \ref{thm:lownnhard}, a contradiction. This completes the result.
\end{proof}

\begin{proof}[Proof of Theorem \ref{thm:high-lsolve-hard}]
\Aaron{Check constants}
Consider an instance of bichromatic Hamming nearest neighbor search for $n = 2^{L^{0.49}}$ and $d = c_1\log n$ for the constant $c_1$ in the dimension bound in Theorem \ref{thm:r18}. This consists of two sets of points $A,B\subseteq \mathbb{R}^d$ with $|A\cup B| = n$ for which we wish to compute $\min_{a\in A, b\in B} \|a - b\|_2$. The coordinates of points in $A$ and $B$ are 0-1. Therefore, the set $S$ of possible $\ell_2$ distances between points in $A$ and $B$ is the set of square roots of integers between 0 and $c_1\log n$, which differ by a factor of at least $1 + 1/(16c_1\log n) > \rho$ (recall $\rho = 1 + (16\log (10n))/L$). Therefore, $A\cup B$ is $\rho$-spaced. Note that $f$ is not $(\rho,L)$-multiplicatively Lipschitz by assumption. Furthermore, $A\cup B$ is $\sqrt{c_1\log n}\le L^{.25}$-boxed.

We now give an algorithm for solving $\ell_2$-closest pair on $A\times B$. Use binary search on $S$. For each query $k\in S$, Lemma \ref{lem:lsolve-reduc} implies that one can check if there is a pair with distance at most $k$ in 
\begin{align*}
\T(n,L,d)
= & ~ O(\T(2^{L^{.49}},L,e^{L^{.25}},2^{c_1 L^{.49}}) + 2^{L^{.49}}) \\
\le & ~ O(2^{L^{.49} + L^{.48}}) \\ 
= & ~ n^{1 + o(1)}
\end{align*}
time on pairs of sets with $n$ points. But this is impossible given {\sf SETH} by Theorem \ref{thm:r18}. This completes the result.
\end{proof} 
\section{Fast Multipole Method}\label{sec:fastmm}
The fast multipole method (FMM) was described as one of the top-10 most important algorithms of the 20th century \cite{dongarra2000guest}. It is a numerical technique that was developed to speed up calculations of long-range forces in the $n$-body problem in physics. In 1987, FMM was first introduced by Greengard and Rokhlin \cite{gr87}, based on the multipole expansion of the vector Helmholtz equation. By treating the interactions between far-away basis functions using the FMM, the corresponding matrix elements do not need to be explicitly computed or stored. This is technique allows us to improve the naive $O(n^2)$ matrix-vector multiplication time to $o(n^2)$.

Since Greengard and Rokhlin invented FMM, the topic has attracted researchers from many different fields, including physics, math, and computer science  \cite{gr87,g88,gr88,gr89,g90,gs91,emrv92,g94,gr96,bg97,d00,ydgd03,ydd04,m12}.

We first give a quick overview of the high-level ideas of FMM in Section~\ref{sec:fastmm_overview}. In Section~\ref{sec:fastmm_gaussian_kernel}, we provide a complete description and proof of correctness for the fast Gaussian transform, where the kernel function is the Gaussian kernel. Although a number of researchers have used FMM in the past, most of the previous papers about FMM either focus on the low-dimensional or low-error cases. We therefore focus on the superconstant-error, high dimensional case, and carefully analyze the joint dependence on $\eps$ and $d$. We believe that our presentation of the original proof in Section~\ref{sec:fastmm_gaussian_kernel} is thus of independent interest to the community. In Section~\ref{sec:fastmm_other}, we give the analogous results for other kernel functions used in this paper.


\subsection{Overview}\label{sec:fastmm_overview}
We begin with a description of high-level ideas of the Fast Multipole Method (FMM). Let $\k : \R^d \times \R^d \rightarrow \R$ denote a kernel function. The inputs to the FMM are $N$ sources $s_1, s_2, \cdots, s_N \in \R^d$  and $M$ targets $t_1, t_2, \cdots, t_M$. For each $i \in [N]$, the source $s_i$ has a strength $q_i$. Suppose all sources are in a `box' ${\cal B}$ and all the targets are in a `box' ${\cal C}$. The goal is to evaluate
\begin{align*}
u_j = \sum_{i=1}^N \k (s_i, t_j ) q_i, ~~~ \forall j \in [M]
\end{align*}
Intuitively, if $\k$ has some nice property (e.g. smooth), we can hope to approximate $\k$ in the following sense
\begin{align*}
\k(s,t) \approx \sum_{p=0}^{P-1} B_p(s) \cdot C_p(t), ~~~ s \in {\cal B} , t \in {\cal C}
\end{align*}
where $P$ is a small positive integer, usually called the \emph{interaction rank} in the literature.

Now, we can construct $u_i$ in two steps:
\begin{align*}
v_p = \sum_{i \in {\cal B}}  B_p (s_i) q_i, ~~~ \forall p = 0,1, \cdots, P-1,
\end{align*}
and
\begin{align*}
\wt{u}_j = \sum_{p=0}^{P-1} C_p (t_j) v_p, ~~~ \forall i \in [M].
\end{align*}

Intuitively, as long as ${\cal B}$ and ${\cal C}$ are well-separated, then $\wt{u}_j$ is very good estimation to $u_j$ even for small $P$, i.e., $|\wt{u}_j - u_j | < \epsilon$.

Recall that, at the beginning of this section, we assumed that all the sources are in the the same box ${\cal B}$ and ${\cal C}$. This is not true in general. To deal with this, we can discretize the continuous space into a batch of boxes ${\cal B}_1, {\cal B}_2, \cdots $ and ${\cal C}_1, {\cal C}_2, \cdots $. For a box ${\cal B}_{l_1}$ and a box ${\cal C}_{l_2}$, if they are very far apart, then the interaction between points within them is small, and we can ignore it. If the two boxes are close, then we deal wit them efficiently by truncating the high order expansion terms in $\k$ (only keeping the first $\log^{O(d)}(1/\epsilon)$). For each box, we will see that the number of nearby relevant boxes is at most $\log^{O(d)}(1/\epsilon)$.

\subsection{\texorpdfstring{$\k (x,y) = \exp ( -  \| x - y \|_2^2 ) $}{}, Fast Gaussian transform}\label{sec:fastmm_gaussian_kernel}

Given $N$ vectors $s_1, \cdots s_N \in \R^d$, $M$ vectors $t_1, \cdots, t_M \in \R^d$ and a strength vector $q \in \R^n$, Greengard and Strain \cite{gs91} provided a fast algorithm for evaluating discrete Gauss transform
\begin{align*}
G(t_i) = \sum_{j=1}^N q_j e^{ - \| t_i - s_j \|^2 / \delta }
\end{align*}
for $i \in [M]$ in $O(M+N)$ time. In this section, we re-prove the algorithm described in \cite{gs91}, and determine the exact dependences on $\epsilon$ and $d$ in the running time.

By shifting the origin and rescaling $\delta$, we can assume that the sources $s_j$ and targets $t_i$ all lie in the unit box ${\cal B}_0 = [0,1]^d$. 

Let $t$ and $s$ lie in $d$-dimensional Euclidean space $\R^d$, and consider the Gaussian 
\begin{align*}
e^{- \| t - s \|_2^2 } = e^{ - \sum_{i=1}^d (t_i - s_i)^2  }
\end{align*}

We begin with some definitions.
\begin{definition}[one-dimensional Hermite polynomial]
The Hermite polynomials $\wt{h}_n : \R \rightarrow \R$ is defined as follows
\begin{align*}
\wt{h}_n(t) = (-1)^n e^{t^2} \frac{ \mathrm{d}^n }{ \mathrm{d} t } e^{-t^2}
\end{align*}
\end{definition}

\begin{definition}[one-dimensional Hermite function]
The Hermite functions $h_n : \R \rightarrow \R$ is defined as follows
\begin{align*}
h_n(t) = e^{-t^2} \wt{h}_n(t)
\end{align*}
\end{definition}

We use the following Fact to simplify $e^{-(t-s)^2/\delta}$.
\begin{fact}
For $s_0 \in \R$ and $\delta > 0$, we have
\begin{align*}
e^{ - (t-s)^2 / \delta } = \sum_{n=0}^{\infty} \frac{1}{n !} \cdot \left( \frac{ s - s_0 }{ \sqrt{\delta} } \right)^n \cdot h_n \left( \frac{ t - s_0 }{ \sqrt{\delta} } \right)
\end{align*}
and
\begin{align*}
e^{ - (t-s)^2 / \delta } =  e^{ - (t - s_0)^2 / \delta } \sum_{n=0}^{\infty} \frac{1}{n !} \cdot \left( \frac{ s - s_0 }{ \sqrt{\delta} } \right)^n \cdot \wt{h}_n \left( \frac{t - s_0}{ \sqrt{\delta} } \right) .
\end{align*}
\end{fact}
\begin{proof}

\begin{align*}
e^{ - (t-s)^2 / \delta } 
= & ~ e^{ - ( t - s_0 - (s - s_0) )^2 / \delta } \\
= & ~ \sum_{n=0}^{\infty} \frac{1}{n !} \left( \frac{ s - s_0 }{ \sqrt{\delta} } \right)^n h_n \left( \frac{ t - s_0 }{ \sqrt{\delta} } \right) \\
= & ~ e^{ - (t - s_0)^2 / \delta } \sum_{n=0}^{\infty} \frac{1}{n !} \left( \frac{ s - s_0 }{ \sqrt{\delta} } \right)^n \wt{h}_n \left( \frac{t - s_0}{ \sqrt{\delta} } \right). 
\end{align*}

\end{proof}

Using Cramer's inequality, we have the following standard bound.
\begin{lemma}
For any constant $K \leq 1.09$, we have
\begin{align*}
|\wt{h}_n(t) | \leq K \cdot 2^{n/2} \cdot \sqrt{ n !} \cdot e^{t^2 /2}
\end{align*}
and
\begin{align*}
| h_n(t) | \leq K \cdot 2^{n/2} \cdot \sqrt{n !} \cdot e^{-t^2/2} .
\end{align*}
\end{lemma}

Next, we will extend the above definitions and observations to the high dimensional case. To simplify the discussion, we define multi-index notation. A multi-index $\alpha = (\alpha_1,\alpha_2, \cdots, \alpha_d)$ is a $d$-tuple of nonnegative integers, playing the role of a multi-dimensional index. For any multi-index $\alpha \in \R^d$ and any $t\in \R^t$, we write
\begin{align*}
\alpha ! = ~ \prod_{i=1}^d (\alpha_i ! ), ~~~~
t^{\alpha} = ~ \prod_{i=1}^d t_i^{\alpha_i} , ~~~~
D^{\alpha} = ~ \partial_1^{\alpha_1} \partial_2^{\alpha_2} \cdots \partial_d^{\alpha_d} .
\end{align*}
where $\partial_i$ is the differentiatial operator with respect to the $i$-th coordinate in $\R^d$. For integer $p$, we say $\alpha \geq p$ if $\alpha_i \geq p, \forall i \in [d]$.

We can now define:
\begin{definition}[multi-dimensional Hermite polynomial]
We define function $\wt{H}_{\alpha} : \R^d \rightarrow \R$ as follows:
\begin{align*}
\wt{H}_{\alpha}(t) = \prod_{i=1}^d \wt{h}_{\alpha_i} (t_i) .
\end{align*}
\end{definition}

\begin{definition}[multi-dimensional Hermite function]
We define function $H_{\alpha} : \R^d \rightarrow \R$ as follows:
\begin{align*}
H_{\alpha}(t) = \prod_{i=1}^d h_{\alpha_i}(t_i). 
\end{align*}
It is easy to see that $ H_{\alpha}(t) = e^{ - \| t \|_2^2 } \cdot \wt{H}_{\alpha}(t)$
\end{definition}

The Hermite expansion of a Gaussian in $\R^d$ is
\begin{align}\label{eq:hermite_expansion_of_gaussian}
e^{- \| t - s \|_2^2 } = \sum_{ \alpha \geq 0 } \frac{ ( t - s_0 )^{\alpha} }{ \alpha ! } h_{\alpha} ( s - s_0 ).
\end{align}
Cramer's inequality generalizes to
\begin{lemma}[Cramer's inequality]\label{lem:cramer_inequality}
Let $K<(1.09)^d$, then
\begin{align*}
 | \wt{H}_{\alpha} (t) | \leq K \cdot e^{ \| t \|_2^2 / 2 } \cdot 2^{ \| \alpha \|_1 /2 } \cdot \sqrt{ \alpha ! } 
\end{align*}
and 
\begin{align*}
 | H_{\alpha} (t) | \leq K \cdot e^{ - \| t \|_2^2 / 2 } \cdot 2^{ \| \alpha \|_1 /2 } \cdot \sqrt{ \alpha ! } 
\end{align*}
\end{lemma}

The Taylor series of $H_{\alpha}$ is
\begin{align}\label{eq:taylor_series_of_H_t}
H_{\alpha}(t) = \sum_{\beta \geq 0} \frac{ (t-t_0)^{\beta} }{ \beta ! } (-1)^{\| \beta \|_1} H_{\alpha + \beta} (t_0)
\end{align}

\subsubsection{Estimation}

We first give a definition
\begin{definition}\label{def:G_t}
Let ${\cal B}$ denote a box with center $s_{\cal B}$ and side length $r \sqrt{2\delta}$ with $r < 1$.
If source $s_j$ is in box ${\cal B}$, we say $j \in {\cal B}$. Then the Gaussian evaluation from the sources in box ${\cal B}$ is,
\begin{align*}
G(t) = \sum_{j \in {\cal B}} q_j \cdot e^{ - \| t - s_j \|_2^2 /\delta  }.  
\end{align*}
The Hermite expansion of $G(t)$ is 
\begin{align}\label{eq:hermite_expansion_of_G_t}
G(t) = \sum_{\alpha \geq 0} A_{\alpha} \cdot h_{\alpha} \left( \frac{ t - s_{\cal B} }{ \sqrt{\delta} } \right),  
\end{align}
where the coefficients $A_{\alpha}$ are defined by
\begin{align}\label{eq:def_A_alpha}
A_{\alpha} = \frac{1}{\alpha !} \sum_{j \in {\cal B}} q_j \cdot \left( \frac{ s_j - s_{\cal B} }{ \sqrt{\delta} } \right)^{\alpha} 
\end{align}
\end{definition}

The rest of this section will present a batch of Lemmas that bound the error of the function truncated at certain degree of Taylor and Hermite expansion.
\begin{lemma}\label{lem:fast_gaussian_lemma_1}
Let $p$ denote an integer, let $\Err_H(p)$ denote the error after truncating the series $G(t)$ (as defined in Def.~\ref{def:G_t}) after $p^d$ terms, i.e.,
\begin{align*}
\Err_H(p) = \sum_{\alpha \geq p} A_{\alpha} \cdot H_{\alpha} \left( \frac{t - s_{\cal B}}{ \sqrt{\delta} } \right).
\end{align*}
 Then we have
\begin{align*}
| \Err_H(p) | \leq K \cdot \sum_{j \in {\cal B}} |q_j| \cdot \left( \frac{1}{p !} \right)^{d/2} \cdot \left( \frac{ r^{ p + 1 } }{ 1 - r } \right)^d,
\end{align*}
where $K = (1.09)^d$.
\end{lemma}
\begin{proof}
Using Eq.~\eqref{eq:hermite_expansion_of_gaussian} to expand each Gaussian (see Definition~\ref{def:G_t}) in the 
\begin{align*}
G(t) = \sum_{j \in {\cal B}} q_j \cdot e^{ - \| t - s_j \|_2^2 / \delta }
\end{align*}
into a Hermite series about $s_{\cal B}$
\begin{align*}
\sum_{\alpha \geq 0} \left( \frac{1}{\alpha !} \sum_{j \in {\cal B}} q_j \cdot \left( \frac{ s_j - s_{\cal B} }{ \sqrt{\delta} } \right)^{\alpha} \right) H_{\alpha} \left( \frac{t - s_{\cal B}}{ \sqrt{\delta} } \right)
\end{align*}
and swap the summation over $\alpha$ and $j$ to obtain
\begin{align*}
 \sum_{j \in {\cal B}} q_j \sum_{\alpha \geq 0}  \frac{1}{\alpha !} \cdot \left( \frac{ s_j - s_{\cal B} }{ \sqrt{\delta} } \right)^{\alpha} \cdot H_{\alpha} \left( \frac{t - s_{\cal B}}{ \sqrt{\delta} } \right)
\end{align*}
The truncation error bound follows from Cramer's inequality (Lemma~\ref{lem:cramer_inequality}) and the formula for the tail of a geometric series.

\end{proof}

The next Lemma shows how to convert a Hermite expansion about $s_{\cal B}$ into a Taylor expansion about $t_{\cal C}$. The Taylor series converges rapidly in a box of side length $r \sqrt{2 \delta}$ about $t_{\cal C}$, where $r < 1$.
\begin{lemma}\label{lem:fast_gaussian_lemma_2}
The Hermite expansion of $G(t)$ is
\begin{align*}
G(t) = \sum_{\alpha \geq 0} A_{\alpha} \cdot H_{\alpha} \left( \frac{ t - s_{\cal B} }{ \sqrt{\delta} } \right)
\end{align*}
has the following Taylor expansion, at an arbitrary point $t_0$ :
\begin{align}\label{eq:taylor_expansion_of_G_t}
G(t) = \sum_{\beta \geq 0} B_{\beta} \left( \frac{ t - t_0 }{ \sqrt{\delta} } \right)^{\beta} .
\end{align}
where the coefficients $B_{\beta}$ are defined as 
\begin{align}\label{eq:def_B_beta}
B_{\beta} = \frac{ (-1)^{ | \beta | } }{ \beta ! } \sum_{\alpha \geq 0} A_{\alpha} \cdot H_{\alpha + \beta} \left( \frac{ s_{\cal B} - t_0 }{ \sqrt{\delta} } \right).
\end{align}
Let $\Err_T(p)$ denote the error by truncating the Taylor series after $p^d$ terms, in the box ${\cal C}$ with center $t_{\cal C}$ and side length $r \sqrt{2\delta}$, i.e.,
\begin{align*}
\Err_T(p) = \sum_{\beta \geq p} B_{\beta} \left( \frac{ t - t_{\cal C} }{ \sqrt{\delta} } \right)^{\beta}
\end{align*}
Then 
\begin{align*} 
| \Err_T(p) | \leq K \cdot Q_B \cdot \left( \frac{1}{p!} \right)^{d/2} \left( \frac{ r^{p+1} }{1-r} \right)^d .
\end{align*}
\end{lemma}
\begin{proof}
Each Hermite function in Eq.~\eqref{eq:hermite_expansion_of_G_t} can be expanded into a Taylor series by means of Eq.~\eqref{eq:taylor_series_of_H_t}. The expansion in Eq.~\eqref{eq:taylor_expansion_of_G_t} is obtained by swapping the order of summation. 

The truncation error bound can be proved as follows. Using Eq.~\eqref{eq:def_A_alpha} for $A_{\alpha}$, we can rewrite $B_{\beta}$:
\begin{align*}
B_{\beta} = & ~ \frac{ (-1)^{|\beta|} }{ \beta ! } \sum_{\alpha \geq 0} A_{\alpha} H_{\alpha + \beta} \left( \frac{ s_{\cal B} - t_{\cal C} }{ \sqrt{\delta} } \right) \\
= & ~ \frac{ (-1)^{|\beta|} }{ \beta ! } \sum_{\alpha \geq 0} \left( \frac{1}{ \alpha ! } \sum_{j \in {\cal B} } q_j \left( \frac{ s_j - s_{\cal B} }{ \sqrt{\delta} } \right)^{\alpha} \right) H_{\alpha + \beta} \left( \frac{ s_{\cal B} - t_{\cal C} }{ \sqrt{\delta} } \right) \\
= & ~ \frac{ (-1)^{|\beta|} }{ \beta ! } \sum_{j \in {\cal B} } q_j \sum_{\alpha \geq 0 } \frac{1}{ \alpha !} \left( \frac{s_j - s_{\cal B} }{ \sqrt{\delta} } \right)^{\alpha} \cdot H_{\alpha + \beta} \left( \frac{ s_{\cal B} - t_{\cal C} }{ \sqrt{\delta} } \right)
\end{align*}
By Eq.~\eqref{eq:taylor_series_of_H_t}, the inner sum is the Taylor expansion of $H_{\beta} ( (s_j - t_{\cal C}) / \sqrt{\delta} )$. Thus
\begin{align*}
B_{\beta} = \frac{ (-1)^{\|\beta\|_1} }{ \beta ! } \sum_{j \in {\cal B} } q_j \cdot H_{\beta} \left( \frac{ s_j - t_{\cal C} }{ \sqrt{\delta} } \right)
\end{align*}
and Cramer's inequality implies
\begin{align*}
| B_{\beta} | \leq \frac{1}{\beta !} K \cdot Q_B 2^{\|\beta\|_1/2} \sqrt{\beta !} \leq K Q_B \frac{ 2^{ \| \beta \|_1 / 2 } }{ \sqrt{\beta !} }
\end{align*}
The truncation error follows from summation the tail of a geometric series.
\end{proof}

For the purpose of designing our algorithm, we'd like to make a variant of Lemma~\ref{lem:fast_gaussian_lemma_2} in which the Hermite series is truncated before converting it to a Taylor series. This means that in addition to truncating the Taylor series itself, we are also truncating the finite sum formula in Eq.~\eqref{eq:def_B_beta} for the coefficients. 

\begin{lemma}\label{lem:fast_gaussian_lemma_3}
Let $G(t)$ be defined as Def~\ref{def:G_t}. For an integer $p$, let $G_{p}(t)$ denote the Hermite expansion of $G(t)$ truncated at $p$,
\begin{align*}
G_p(t) = \sum_{\alpha \leq p} A_{\alpha} H_{\alpha} \left( \frac{ t - s_{\cal B} }{ \sqrt{\delta} } \right).
\end{align*}
The function $G_p(t)$ has the following Taylor expansion about an arbitrary point $t_0$:
\begin{align*}
G_p(t) = \sum_{\beta \geq 0} C_{\beta} \cdot \left( \frac{ t - t_0 }{ \sqrt{\delta} } \right)^{\beta},
\end{align*}
where the the coefficients $C_{\beta}$ are defined as
\begin{align}\label{eq:def_C_beta}
C_{\beta}= \frac{ (-1)^{\| \beta \|_1} }{ \beta ! } \sum_{\alpha \leq p} A_{\alpha} \cdot H_{\alpha + \beta} \left( \frac{s_{\cal B} - t_{\cal C} }{ \sqrt{\delta} } \right).
\end{align}
Let $\Err_T(p)$ denote the error in truncating the Taylor series after $p^d$ terms, in the box ${\cal C}$ with center $t_{\cal C}$ and side length $r \sqrt{2\delta} $, i.e.,
\begin{align*}
\Err_T(p ) = \sum_{\beta \geq p} C_{\beta} \left( \frac{ t - t_{\cal C} }{ \sqrt{\delta} } \right)^{\beta}.
\end{align*}
Then, we have
\begin{align*}
| \Err_T(p) | \leq K' \cdot Q_B \left( \frac{1}{p !} \right)^{d/2} \left( \frac{r^{p+1}}{1-r} \right)^d
\end{align*}
where $K' \leq 2K$ and $r \leq 1/2$. 
\end{lemma}
\begin{proof}

We can write $C_{\beta}$ in the following way:
\begin{align*}
C_{\beta} = & ~ \frac{ (-1)^{ \| \beta \|_1 } }{ \beta ! } \sum_{j \in {\cal B}} q_j \sum_{\alpha \leq p} \frac{1}{ \alpha ! } \left( \frac{ s_j - s_{\cal B} }{ \sqrt{\delta} } \right)^{\alpha} \cdot H_{\alpha + \beta} \left( \frac{ s_{\cal B} - t_{\cal C} }{ \sqrt{\delta} } \right) \\
= & ~ \frac{ (-1)^{\| \beta \|_1} }{ \beta ! } \sum_{j \in {\cal B}} q_j \left( \sum_{\alpha \geq 0} - \sum_{\alpha > p} \right) \frac{1}{ \alpha ! } \left( \frac{ s_j - s_{\cal B} }{ \sqrt{\delta} } \right)^{\alpha} \cdot H_{\alpha + \beta} \left( \frac{ s_{\cal B} - t_{\cal C} }{ \sqrt{\delta} } \right) \\
= & ~ B_{\beta} - \frac{ (-1)^{ \| \beta \|_1 } }{ \beta ! } \sum_{j \in {\cal B}} q_j \sum_{\alpha > p} \frac{1}{ \alpha ! } \left( \frac{ s_j - s_{\cal B} }{ \sqrt{\delta} } \right)^{\alpha} \cdot H_{\alpha + \beta} \left( \frac{ s_{\cal B} - t_{\cal C} }{ \sqrt{\delta} } \right) \\
= & ~ B_{\beta} + (C_{\beta} - B_{\beta})
\end{align*}
Next, we have
\begin{align}\label{eq:split_Err_T_p_into_two_terms}
| \Err_T(p) | \leq \left| \sum_{\beta \geq p} B_{\beta} \left( \frac{ t - t_{\cal C} }{ \sqrt{\delta} } \right)^{\beta} \right| + \left| \sum_{\beta \geq p} (C_{\beta} - B_{\beta}) \cdot \left( \frac{ t - t_{\cal C} }{ \sqrt{\delta} } \right)^{\beta} \right|
\end{align}
Using Lemma~\ref{lem:fast_gaussian_lemma_2}, we can upper bound the first term in the Eq.~\eqref{eq:split_Err_T_p_into_two_terms} by,
\begin{align*}
K \cdot Q_B \left( \frac{1}{ p ! } \right)^{d/2} \cdot \left( \frac{ r^{p+1} }{1-r} \right)^d
\end{align*}
To bound the second term in Eq.~\eqref{eq:split_Err_T_p_into_two_terms}, we can do the following
\begin{align*}
 & ~ \left| \sum_{\beta \geq p} (C_{\beta} - B_{\beta}) \cdot \left( \frac{ t - t_{\cal C} }{ \sqrt{\delta} } \right)^{\beta} \right| \\
\leq & ~ Q_B \cdot \sum_{\beta \geq p} \left| \Big( \frac{ t - t_{\cal C} }{ \sqrt{\delta} } \Big)^{\beta} \right| \cdot \frac{1}{\beta !} \sum_{\alpha > p} \frac{1}{\alpha !} \left| \Big( \frac{ s_j - s_{\cal B} }{ \sqrt{\delta} } \Big)^{\alpha} \right| \cdot \left| H_{\alpha + \beta} \left( \frac{s_{\cal B} -  t_{\cal C} } { \sqrt{\delta} } \right) \right| \\
\leq & ~ K Q_{B} \sum_{\alpha > p} \sum_{\beta > p} \frac{ r^{\| \alpha \|_1 } }{  \sqrt{ \alpha ! }} \cdot \sqrt{ \frac{ (\alpha +\beta) ! }{ \alpha ! \beta ! } } \cdot \frac{ r^{\| \beta \|_1 } }{  \sqrt{ \beta ! }}
\end{align*}
Finally, the proof is complete since we know that
\begin{align*}
\frac{ (\alpha + \beta) ! }{  \alpha ! \beta ! } \leq 2^{ \| \alpha + \beta \|_1 }.
\end{align*}
\end{proof}

The proof of the following Lemma is almost identical. We omit the details here.
\begin{lemma}\label{lem:fast_gaussian_lemma_4}
Let $G_{s_j} (t)$ be defined as
\begin{align*}
G_{s_j}(t) = q_j \cdot e^{- \| t - s_j \|_2^2 / \delta}
\end{align*}
has the following Taylor expansion at $t_{\cal C}$
\begin{align*}
G_{s_j}(t) = \sum_{\beta \geq 0} {\cal B}_{\beta} \left( \frac{ t - t_{\cal C} }{ \sqrt{\delta} } \right)^{\beta},
\end{align*}
where the coefficients $B_{\beta}$ is defined as
\begin{align*}
B_{\beta} = q_j \cdot \frac{ (-1)^{ \| \beta\|_1 } }{ \beta ! } \cdot H_{\beta} \left( \frac{ s_j - t_{\cal C} }{ \sqrt{\delta} } \right)
\end{align*}
and the error in truncation after $p^d$ terms is
\begin{align*}
| \Err_T(p) | = \left| \sum_{\beta \geq p} B_{\beta} \left( \frac{t - t_{\cal C}}{ \sqrt{\delta}}  \right)^{\beta} \right| \leq K \cdot q_j \cdot \left( \frac{1}{p!} \right)^{d/2} \cdot \left( \frac{ r^{p+1} }{ 1 -r } \right)^d
\end{align*}
for $r < 1$.
\end{lemma}

\subsubsection{Algorithm}

The algorithm is based on subdividing $B_0$ into smaller boxes with sides of length $r \sqrt{2\delta}$ parallel to the axes, for a fixed $r \leq 1/2$. We can then assign each source $s_j$ to the box ${\cal B}$ in which it lies and each target $t$, to the box ${\cal C}$ in which it lies. 

For each target box ${\cal C}$, we need to evaluate the total field due to sources in all boxes. Since boxed ${\cal B}$ have side lengths $r \sqrt{2\delta}$, only a fixed number of source boxes ${\cal B}$ can contribute more than $Q \epsilon$ to the field in a given target box ${\cal C}$, where $Q = \| q\|_1$ and $\epsilon$ is the precision parameter. If we cut off the sum over all ${\cal B}$ after including the $(2k+1)^d$ nearest boxes to ${\cal C}$, it incurs an error which can be upper bounded as follows 
\begin{align}\label{eq:box_error}
\sum_{j : \| t - s_j \|_{\infty} \geq k r \sqrt{2\delta} } |q_j| \cdot e^{-\| t - s_j \|_2^2 / \delta} 
\leq & ~ \sum_{j : \| t - s_j \|_{\infty} \geq k r \sqrt{2\delta} } |q_j| \cdot e^{-\| t - s_j \|_{\infty}^2 / \delta} \notag\\
\leq & ~  \sum_{j : \| t - s_j \|_{\infty} \geq k r \sqrt{2\delta}} |q_j| \cdot e^{ - ( k \cdot r \sqrt{2\delta} )^2 / \delta } \notag \\
\leq & ~ Q \cdot e^{ - 2r^2 k^2}
\end{align}
where the first step follows from $\| \cdot \|_2 \geq \| \cdot \|_{\infty}$, the second step follows from $\| t - s_j\|_{\infty} \geq k r \sqrt{2\delta}$, and the last step follows from a straightforward calculation.

For a box ${\cal B}$ and a box ${\cal C}$, there are several possible ways to evaluate the interaction between ${\cal B}$ and ${\cal C}$. We mainly need the following three techniques:
\begin{enumerate}
	\item $N_{\cal B}$ Gaussians, accumulated in Taylor series via definition $B_{\beta}$ in Lemma~\ref{lem:fast_gaussian_lemma_4}
	\item Hermite series, directly evaluated
	\item Hermite series, accumulated in Taylor series in Lemma~\ref{lem:fast_gaussian_lemma_3} 
\end{enumerate}
Essentially, having any two of the above three techniques is sufficient to give an algorithm that runs in $(M+N) \log^{O(d)} (\| q\|_1 /\epsilon)$ time.

In the next a few paragraphs, we explain the details of the three techniques.

\paragraph*{Technique 1.} Consider a fixed source box ${\cal B}$. For each target box ${\cal C}$ within range, we must compute $p^d$ Taylor series coefficients
\begin{align*}
C_{\beta} ( {\cal B} ) = \frac{ (-1)^{|\beta|} }{ \beta ! } \sum_{ j \in {\cal B} } H_{\beta} \left( \frac{ s_j - t_C }{ \sqrt{\delta} } \right).
\end{align*}
Each coefficient requires $O(N_{\cal B})$ work to evaluate, resulting in a net cost $O(p^d N_{\cal B})$. Since there are at most $(2k+1)^d$ boxes within range, the total work for forming all the Taylor series is $O( (2k+1)^d p^d N )$. Now, for each target $t_i$, one must evaluate the $p^d$-term Taylor series corresponding to the box in which $t_i$ lies. The total running time of algorithm is thus
\begin{align*}
O( (2k+1)^d p^2 N ) + O(p^d M).
\end{align*}

\paragraph*{Technique 2.} We form a Hermite series for each box ${\cal B}$ and evaluate it at all targets. Using Lemma~\ref{lem:fast_gaussian_lemma_1}, we can rewrite $G(t)$ as
\begin{align*}
G(t) = & ~ \sum_{ {\cal B} } \sum_{ j \in {\cal B} } q_j \cdot e^{ - \| t - s_j \|_2^2 / \delta } \\ 
= & ~ \sum_{ {\cal B} } \sum_{ \alpha \geq 0 } A_{\alpha} (B) H_{\alpha} \left( \frac{ t - s_{\cal B} }{ \sqrt{\delta} } \right) + \Err_H(p)
\end{align*}
where $|\Err_H(p)| \leq \epsilon$ and
\begin{align}\label{def:A_alpha_cal_B}
A_{\alpha}  ( {\cal B} ) = \frac{1}{\alpha !} \sum_{j \in {\cal B}} q_j \cdot \left( \frac{ s_j - s_{\cal B} }{ \sqrt{\delta} } \right)^{\alpha}.
\end{align}
To compute each $A_{\alpha}( {\cal B})$ costs $O(N_{\cal B})$ time, so forming all the Hermite expansions takes $O(p^d N)$ time. Evaluating at most $(2k + 1)^d$ expansions at each target $t_i$ costs $O( (2k+1)^d p^d )$ time per target, so this approach takes
\begin{align*}
O(p^d N) + O(  (2k+1)^d p^d M )
\end{align*}
time in total.

\paragraph*{Technique 3.} Let $N(B)$ denote the number of boxes. Note that $N(B) \leq \min ( ( r \sqrt{2\delta} )^{-d/2} , M )$.

 Suppose we accumulate all sources into truncated Hermite expansions and transform all Hermite expansions into Taylor expansions via Lemma~\ref{lem:fast_gaussian_lemma_3}. Then we can approximate the function $G(t)$ by
\begin{align*}
G(t) = & ~ \sum_{{\cal B}} \sum_{j \in {\cal B}} q_j \cdot e^{ - \| t - s_j \|_2^2 / \delta } \\
= & ~ \sum_{\beta \leq p} C_{\beta} \left( \frac{ t - t_{\cal C} }{ \sqrt{\delta} } \right)^{\beta} + \Err_T(p) + \Err_H(p)
\end{align*}
where $|\Err_H(p)| + |\Err_T(p)| \leq Q \cdot \epsilon$,
\begin{align*}
C_{\beta} = \frac{ (-1)^{\| \beta \|_1 } }{ \beta ! } \sum_{{\cal B}} \sum_{\alpha \leq p} A_{\alpha} ({\cal B}) H_{\alpha + \beta} \left( \frac{ s_{ {\cal B } } - t_{ {\cal C} } }{ \sqrt{\delta} } \right)
\end{align*}
and the coefficients $A_{\alpha}({\cal B})$ are defined as Eq.~\eqref{def:A_alpha_cal_B}. Recall in Part 2, it takes $O(p^d N)$ time to compute all the Hermite expansions, i.e., to compute the coefficients $A_{\alpha} ({\cal B})$ for all $\alpha \leq p$ and all sources boxes ${\cal B}$.

Making use of the large product in the definition of $H_{\alpha+\beta}$, we see that the time to compute the $p^d$ coefficients of $C_{\beta}$ is only $O(d p^{d+1})$ for each box ${\cal B}$ in the range. Thus, we know for each target box ${\cal C}$, the running time is
\begin{align*}
O( (2k +1)^d d p^{d+1} ).
\end{align*} 

Finally we need to evaluate the appropriate Taylor series for each target $t_i$, which can be done in $O(p^d M)$ time. Putting it all together, this technique 3 takes time
\begin{align*}
O( (2k+1)^d d p^{d+1} N(B) ) + O(p^d N) + O(p^d M).
\end{align*}

\subsubsection{Result}

Finally, in order to get $\epsilon$ additive error for each coordinate, we will choose $k = O(\log (\| q \|_1 /\epsilon))$ and $p = O(\log( \| q \|_1 / \epsilon ))$.

\begin{theorem}[fast Gaussian transform]\label{thm:fast_gaussian_transform}
Given $N$ vectors $s_1, \cdots, s_N \in \R^d$, $M$ vectors $t_1, t_2, \cdots, t_M$ $ \in \R^d$, a number $\delta > 0$, and a vector $q \in \R^n$, let function $G : \R^d \rightarrow \R$ be defined as $G(t) = \sum_{i=1}^N q_i \cdot e^{ -\| t - s_i \|_2^2/\delta }$. 
There is an algorithm that runs in 
\begin{align*}
O \left( (M + N) \log^{O(d)} ( \| q \|_1 / \epsilon ) \right)
\end{align*}
time, and outputs $M$ numbers $x_1, \cdots, x_M$ such that for each $j \in [M]$
\begin{align*}
 G(t_j) - \epsilon \leq x_j \leq G(t_j) +\epsilon .
\end{align*}
\end{theorem}

The proof of fast Gaussian transform also implies a result for the online version:
\begin{theorem}[online version]\label{thm:fast_gaussian_transform_online}
Given $N$ vectors $s_1, \cdots, s_N \in \R^d$, a number $\delta > 0$, and a vector $q \in \R^n$, let function $G : \R^d \rightarrow \R$ be defined as $G(t) = \sum_{i=1}^N q_i \cdot e^{ -\| t - s_i \|_2^2/\delta }$. 
There is an algorithm that takes 
\begin{align*}
O \left( N \log^{O(d)} ( \| q \|_1 / \epsilon ) \right)
\end{align*}
time to do preprocessing, and then for each $t \in \R^d$, takes $O(\log^{O(d)} ( \| q \|_1 / \epsilon ))$ time to output a number $x$ such that 
\begin{align*}
 G(t) - \epsilon \leq x \leq G(t) +\epsilon .
\end{align*}
\end{theorem}

\subsection{Generalization}\label{sec:fastmm_general}

The fast multipole method described in the previous section works not only for the Gaussian kernel, but also for $\k(u,v) = f(\|u,v\|_2^2)$ for many other functions $f$. As long as $f$ has the following properties, the result of Theorem~\ref{thm:fast_gaussian_transform} also holds for $f$:
\begin{itemize}
    \item $f : \R \rightarrow \R_{+}$.
    \item $f$ is non-increasing, i.e., if $x \geq y \geq 0$, then $f(x) \leq f(y)$.
    \item $f$ is decreasing fast, i.e., for any $\epsilon \in (0,1)$, we have $f( \Theta(\log (1/\epsilon) ) ) \leq \epsilon$.
    \item $f$'s Hermite expansion and Taylor expansions are truncateable: If we only keep $\log^d(1/\epsilon)$ terms of the polynomial for $\k$, then the error is at most $\epsilon$.
\end{itemize}

Let us now sketch how each of these properties is used in discritizing the continuous domain into a finite number of boxes. First, note that Eq.\eqref{eq:box_error} holds more generally for any function $f$ with these properties. Indeed, we can bound the error as follows (note that $Q = \| q \|_1$):
\begin{align}\label{eq:box_error_general}
\sum_{j : \| t - s_j \|_{\infty} \geq k r \sqrt{2\delta} } |q_j| \cdot f( \| t - s_j \|_2 / \sqrt{\delta} ) 
\leq & ~ \sum_{j : \| t - s_j \|_{\infty} \geq k r \sqrt{2\delta} } |q_j| \cdot f( \| t - s_j \|_{\infty}/ \sqrt{\delta} ) \notag\\
\leq & ~  \sum_{j : \| t - s_j \|_{\infty} \geq k r \sqrt{2\delta}} |q_j| \cdot f( \sqrt{2} k r ) \notag \\
\leq & ~ Q \cdot f( \sqrt{2} k r ) \notag \\
\leq & ~ \epsilon
\end{align}
where the first step follows from $\| \cdot \|_2 \geq \| \cdot \|_{\infty}$ and that $f$ is non-increasing, the second step follows from $\| t - s_j\|_{\infty} \geq k r \sqrt{2\delta}$ and that  $f$ is non-increasing, and the last step follows from the fact that $f$ is decreasing fast, and choosing $k = O(\log (Q/\epsilon) / r)$.

We next give an example of how the truncatable expansions property is used. Here we only show to generalize Definition~\ref{def:G_t} to Definition~\ref{def:G_t_general} and generalize Lemma~\ref{lem:fast_gaussian_lemma_1} to Definition~\ref{lem:fast_gaussian_lemma_1_general}; the other Lemmas in the proof can be extended in a similar way.

\begin{definition}\label{def:G_t_general}
Let ${\cal B}$ denote a box with center $s_{\cal B}$ and side length $r \sqrt{2\delta}$ with $r < 1$.
If source $s_j$ is in box ${\cal B}$, we say $j \in {\cal B}$. Then the Gaussian evaluation from the sources in box ${\cal B}$ is,
\begin{align*}
G(t) = \sum_{j \in {\cal B}} q_j \cdot f( \| t - s_j \|_2 / \sqrt{\delta} ).
\end{align*}
The Hermite expansion of $G(t)$ is 
\begin{align}\label{eq:hermite_expansion_of_G_t_general}
G(t) = \sum_{\alpha \geq 0} A_{\alpha} \cdot h_{\alpha} \left( \frac{ t - s_{\cal B} }{ \sqrt{\delta} } \right),  
\end{align}
where the coefficients $A_{\alpha}$ are defined by
\begin{align}\label{eq:def_A_alpha_general}
A_{\alpha} = \frac{1}{\alpha !} \sum_{j \in {\cal B}} q_j \cdot \left( \frac{ s_j - s_{\cal B} }{ \sqrt{\delta} } \right)^{\alpha} 
\end{align}
\end{definition}

\begin{definition}\label{lem:fast_gaussian_lemma_1_general}
Let $p$ denote an integer, let $\Err_H(p)$ denote the error after truncating the series $G(t)$ (as defined in Def.~\ref{def:G_t_general}) after $p^d$ terms, i.e.,
\begin{align*}
\Err_H(p) = \sum_{\alpha \geq p} A_{\alpha} \cdot H_{\alpha} \left( \frac{t - s_{\cal B}}{ \sqrt{\delta} } \right).
\end{align*}
We say $f$ is Hermite truncateable, if 
 Then we have
\begin{align*}
| \Err_H(p) | \leq p^{-\Omega(pd)}
\end{align*}
where $K = (1.09)^d$.
\end{definition}

\subsection{\texorpdfstring{$\k (x,y) = 1 / \| x - y \|_2^2$}{}}\label{sec:fastmm_other}

Similar ideas yield the following algorithm:
\begin{theorem}[FMM, \cite{bg97,m12}]\label{thm:fmminverse}
Given $n$ vectors $x_1, x_2, \cdots, x_n \in \R^d$, let matrix $A \in \R^{n \times n}$ be defined as $A_{i,j} = 1 / \| x_i - x_j \|_2^2$. For any vector $h \in \R^n$, in time $O( n \log^{O(d)}( \| u \|_1 /\epsilon) )$, we can output a vector $u$ such that
\begin{align*}
 (A h)_i - \epsilon \leq  u_i \leq  (A h)_i + \epsilon.
\end{align*}
\end{theorem}

\section{Neural Tangent Kernel} \label{sec:ntk}

In this section, we show that the popular Neural Tangent Kernel $\k$ from theoretical Deep Learning can be rearranged into the form $\k(x,y) = f(\|x-y\|_2^2)$ for an appropriate analytic function $f$, so our results in this paper apply to it.
We first define the kernel. 
\begin{definition}[Neural Tangent Kernel, \cite{jgh18}]\label{def:neural_tangent_kernel}
Given $n$ points $x_1, x_2, \cdots, x_n \in \R^d$, and any activation function $\sigma : \R \rightarrow \R$, the neural tangent kernel matrix $\k \in \R^{n \times n}$ can be defined as follows, where $\mathcal{N}(0,I_d)$ denotes the Gaussian distribution:
\begin{align*}
\k_{i,j}:= \int_{\mathcal{N}(0,I_d)} \sigma'(w^\top x_i)\sigma'(w^\top x_j) x_i^\top x_j \d w .
\end{align*}
\end{definition}

In the literature of convergence results for deep neural networks \cite{ll18,dzps19,als18,als19,sy19,bpsw20,lsswy20}, it is natural to assume that all the data points are on the unit sphere, i.e., for all $i \in [n]$ we have $\| x_i \|_2 = 1$ and datas are separable i.e., for all $i\neq j$, $\| x_i - x_j \|_2 \geq \delta$. One of the most standard and common used activation functions in neural network training is ReLU activation, which is $\sigma(x) = \max\{x,0\}$. Using Lemma~\ref{lem:ntk_relu}, we can figure out the corresponding kernel function. By Theorem~\ref{thm:hardnessapprox}, the multiplication task for neural tangent kernels is hard. In neural network training, the multiplication can potentially being used to speed the neural network training procedure(See \cite{sy19}).

In the following lemma, we compute the kernel function for ReLU activation function.
\begin{lemma}\label{lem:ntk_relu}
If $\sigma(x) = \max\{x,0\}$, then the Neural Tangent Kernel can be written as $\k(x,y) = f(\|x-y\|_2^2)$ for
\begin{align*}
f(x) = \frac{1}{\pi} ( \pi - \cos^{-1} ( 1 - 0.5 x ) ) \cdot (1 - 0.5 x)
\end{align*}
\end{lemma}
\begin{proof}

First, since $\| x_i \|_2 = \|x_j\|_2$, we know that
\begin{align*}
\| x_i - x_j \|_2^2 = \| x_i \|_2^2 - 2 \langle x_i , x_j \rangle + \| x_j \|_2^2 = 2 - 2\langle x_i, x_j \rangle.
\end{align*}

By definition of $\sigma$, we know 
\begin{align*}
\sigma(x) =  
\begin{cases}
1, & \text{~if~} x > 0; \\
0, & \text{~otherwise.}
\end{cases}
\end{align*}
Using properties of the Gaussian distribution ${\cal N}(0,I_d)$, we have
\begin{align*}
\int_{\mathcal{N}(0,I_d)} \sigma'(w^\top x_i)\sigma'(w^\top x_j) \d w 
= & ~ \frac{1}{\pi} (\pi - \cos^{-1} ( x_i x_j ) ) \\
= & ~ \frac{1}{\pi} (\pi - \cos^{-1} ( 1 - 0.5 \| x_i - x_j \|_2^2 ) )
\end{align*}

We can rewrite $\k_{i,j}$ as follows:
\begin{align*}
\k_{i,j} 
= & ~ \int_{\mathcal{N}(0,I_d)} \sigma'(w^\top x_i)\sigma'(w^\top x_j) x_i^\top x_j \d w \\
= & ~ (1-0.5 \| x_i - x_j \|_2^2) \cdot \int_{\mathcal{N}(0,I_d)} \sigma'(w^\top x_i)\sigma'(w^\top x_j) \d w \\
= & ~ (1-0.5 \| x_i - x_j \|_2^2) \cdot \frac{1}{\pi} ( \pi - \cos^{-1} ( 1 - 0.5 x ) )\\
= & ~ f( \| x_i - x_j \|_2^2 ) .
\end{align*}

\end{proof}

\begin{lemma}
Given $n$ data points $x_1, \dots, x_n \in \R^d$ on unit sphere. For any activation function $\sigma : \R \rightarrow \R$, the corresponding Neural Tangent Kernel $\k(x_i,x_j)$ is a function of $\| x_i - x_j \|_2^2$.
\end{lemma}
\begin{proof}
Note that $w\sim \mathcal{N}(0,I_d)$, so we know $(x_i^\top w,x_j^\top w)\sim\mathcal{N}(0, \Sigma_{i,j})$, where the covariance matrix 
\begin{align*}
\Sigma_{i,j}=\begin{bmatrix}
	x_i^\top x_i & x_i^\top x_j\\
	x_j^\top x_i & x_j^\top x_j
\end{bmatrix} = \begin{bmatrix}
	1 & x_i^\top x_j\\
	x_i^\top x_j & 1
\end{bmatrix} \in \R^{2 \times 2},
\end{align*}
since $\|x_i\|_2=\|x_j\|_2 = 1$. 
Thus, 
\begin{align*}
\k(x_i,x_j) = \E_{(a,b)\sim\mathcal{N}(0,\Sigma_{i,j})}[\sigma'(a)\sigma'(b)]x_i^\top x_j = g(x_i^\top x_j)
\end{align*}
for some function $g$.

 Note $x_i^\top x_j = -\frac{1}{2}\|x_i-x_j\|^2 + 1$, so $\k(x_i,x_j)=f( \| x_i-x_j \|_2^2 )$ for some function $f$, which completes the proof.
\end{proof}


\fi


\addcontentsline{toc}{section}{References}
\bibliographystyle{alpha}
\bibliography{ref}

\end{document}